\documentclass[preprint,12pt,3p,authoryear]{elsarticle}
\usepackage[colorlinks, linkcolor=skyblue, citecolor=skyblue]{hyperref}
\usepackage{booktabs} 
\usepackage{array} 
\usepackage{subcaption}

\usepackage{multirow}
\usepackage{tabularx}
\usepackage{lineno}
\usepackage{algorithm}
\usepackage{algorithmic}    % 用于 \STATE, \FOR, \IF, \REQUIRE, \ENSURE 等

\newcommand*\patchAmsMathEnvironmentForLineno[1]{%
  \expandafter\let\csname old#1\expandafter\endcsname\csname #1\endcsname
  \expandafter\let\csname oldend#1\expandafter\endcsname\csname end#1\endcsname
  \renewenvironment{#1}%
     {\linenomath\csname old#1\endcsname}%
     {\csname oldend#1\endcsname\endlinenomath}}%
\newcommand*\patchBothAmsMathEnvironmentsForLineno[1]{%
  \patchAmsMathEnvironmentForLineno{#1}%
  \patchAmsMathEnvironmentForLineno{#1*}}%
\AtBeginDocument{%
\patchBothAmsMathEnvironmentsForLineno{equation}%
\patchBothAmsMathEnvironmentsForLineno{align}%
\patchBothAmsMathEnvironmentsForLineno{flalign}%
\patchBothAmsMathEnvironmentsForLineno{alignat}%
\patchBothAmsMathEnvironmentsForLineno{gather}%
\patchBothAmsMathEnvironmentsForLineno{multline}%
}
\usepackage{amsmath}
\usepackage{amssymb}%% The amsthm package provides extended theorem environments
\usepackage{amsthm}

\newtheorem{theorem}{Theorem}[section]
\newtheorem{lemma}[theorem]{Lemma}
\newtheorem{proposition}[theorem]{Proposition}
\newtheorem{corollary}[theorem]{Corollary}

\theoremstyle{definition}
\newtheorem{definition}[theorem]{Definition}

\theoremstyle{remark}
\newtheorem{remark}[theorem]{Remark}
\journal{Sustainable Cities and Society}

\begin{document}

\begin{frontmatter}

\title{Urban Congestion Patterns under High Electric Vehicle Penetration: A Case Study of 10 U.S. Cities}
%\tnotetext[label0]{This is only an example}

\author[label1,label2,label3]{Xiaohan Xu}

\author[label4]{Wei Ma\corref{cor1}}

\author[label5]{Zhiheng Shi}

\author[label4]{Xiaotong Xu}

\author[label2,label1,label3]{Bin He\corref{cor1}}

\author[label2,label5]{Kairui Feng\corref{cor1}}

\address[label1]{College of Electronics and Information Engineering, Tongji University, Shanghai 201804, China}
\address[label2]{State Key Laboratory of Autonomous Intelligent Unmanned Systems, Shanghai 201210, China}
\address[label3]{Frontiers Science Center for Intelligent Autonomous Systems, Shanghai 201210, China}
\address[label4]{Department of Civil and Environmental Engineering, The Hong Kong Polytechnic University, Hong Kong 999077, China}
\address[label5]{Shanghai Innovation Institute, Shanghai 200030, China}

\cortext[cor1]{Corresponding author: wei.w.ma@polyu.edu.hk; hebin@tongji.edu.cn; kelvinfkr@tongji.edu.cn}

 % \linenumbers
 % \renewcommand\linenumberfont{\normalfont\small}
%\linenumbers

\begin{abstract}

With the global energy transition and the rapid penetration of electric vehicles (EVs), the widening travel cost gap between EVs and gasoline vehicles (GVs) increasingly affects commuters' route choices and may reshape urban congestion patterns. Existing research remains in its preliminary exploratory phase. On the one hand, multi-class models do not account for fixed user class scenarios, which may not align with actual commuters; on the other hand, there is a significant lack of systematic quantitative analysis based on real-world complex road networks across multiple cities. As a result, the congestion effects induced by heterogeneous GV–EV cost structures may be mischaracterized or substantially underestimated. To address these limitations, this paper proposes a multi-user equilibrium (MUE) assignment model for mixed GV-EV traffic, constructs a dual algorithm with convergence guarantees, and designs multi-dimensional evaluation metrics for congestion patterns. Using 10 representative U.S. cities as a case study, this research explores the evolution trends of traffic congestion under different EV penetration scenarios based on real city-level road networks and block-level commuter origin-destination (OD) demand. The results show that full EV penetration reduces average system travel time by 2.27\%--10.78\% across the 10 cities, with New Orleans achieving the largest reduction (10.78\%) and San Francisco the smallest (2.27\%), but the effectiveness of alleviating congestion exhibits significant urban heterogeneity. Moreover, for cities with sufficient network redundancy, benefits are primarily concentrated during the low to medium EV penetration stage (0--0.5), though cities with topological constraints (e.g., San Francisco, Boston) show more limited improvements throughout all penetration levels. This paper can provide a foundation for formulating differentiated urban planning, sustainable transportation development, and congestion management policies.

\end{abstract}

\begin{keyword}
Electric vehicle \sep  Gasoline vehicle \sep Urban congestion pattern \sep Multi-user equilibrium \sep  Dual algorithm
\end{keyword}

\end{frontmatter}

% \linenumbers
% \renewcommand\linenumberfont{\normalfont\small}
%\linenumbers

\section{Introduction}

The rapid advancement of electric vehicle (EV) technology and the accelerated global energy transition are profoundly reshaping urban transportation systems. With ongoing reductions in battery costs, the expansion of charging infrastructure, and progressively stricter energy conservation and emission mitigation policies, the share of EVs in the global passenger car stock is projected to increase steadily in the coming years \citep{orsi2021sustainability, zeng2024optimal}, making the coexistence of EVs and gasoline vehicles (GVs) a salient characteristic of urban transportation systems \citep{zeng2025optimizing}. In this mixed fleet structure, the growing gap in unit travel costs between EVs and GVs is bound to influence travelers' route choices and travel decisions \citep{liu2024eco}, thereby exerting a systemic impact on the spatiotemporal distribution of urban traffic congestion. On the one hand, improvements in electric drivetrain efficiency and the decarbonization of the power grid have made the unit cost per kilometer of EVs more advantageous compared to traditional GVs \citep{taiebat2022widespread}. On the other hand, governments worldwide have reinforced this transition through a combination of policies, including purchase subsidies, exemptions from license plate quotas, and preferential road access \citep{lu2022analysis, hu2025gas}. Therefore, a comprehensive exploration of how the widespread penetration of EVs reshapes urban traffic distribution and congestion patterns is crucial for gaining a deeper understanding of their impact on transportation systems.

In this context, differences in travel utilities between EV and GV users constitute the fundamental driver of their heterogeneous route choices. Existing studies primarily employ the user equilibrium (UE) \citep{sheffi1985urban}  model to characterize the path selection behavior of different vehicle types, thereby providing a foundation for understanding the evolution of urban traffic congestion patterns. Early studies primarily focused on mixed networks of cars and public transit, incorporating both modes into a single equilibrium framework to analyze their mutual influence \citep{florian1977traffic}. Subsequently, \citet{abdulaal1979methods} proposed an approach that jointly solves mode split and traffic assignment within a single equilibrium framework, in which travelers' decisions on both mode choice and route choice are modeled simultaneously, allowing them to trade off generalized travel costs among private cars, public transit, and other modes while ensuring that mode and route choices are jointly in equilibrium. For example, \citet{zhang2020practical} proposed a multi-class equilibrium model that jointly considers walking, public transit, private vehicles, and transfer combinations, providing a systematic tool for analyzing the traffic diversion and congestion mechanisms within multi-modal transit systems. \citet{liu2025integrating} focused on a multi-class travel system integrating ride-hailing and public transit, proposing a stochastic UE (SUE) model comprising private cars, public transit, and ride-hailing.

Traditional multi-user equilibrium (MUE) models typically assume travelers can freely choose different travel modes before departure, differing from the fixed user classes in GV-EV mixed traffic scenarios. Moreover, MUE assignment for GV-EV mixed traffic flows remains in its preliminary exploratory phase. \citet{liang2024vehicle} proposed an EV-UE-constrained vehicle-to-grid framework, integrating EV travel and charging behavior into distribution network planning. \citet{sun2025low} characterized a modified mixed UE model incorporating EV flows and a comprehensive tri-level robust planning of coupled transportation and power distribution networks. These studies indicate that explicitly embedding the traffic assignment process of EVs within power-transportation coupled planning has become a significant development direction for integrated analysis of low-carbon energy and transportation systems. However, existing research typically treats EVs as a single vehicle category among numerous options for travelers. Consequently, an MUE assignment model for GV-EV mixed traffic flows that systematically addresses the GV-EV cost differential and its impact on urban congestion patterns remains to be developed. Furthermore, available results predominantly focus on idealized or small-to-medium-scale networks, lacking systematic quantitative analysis of real urban road networks and actual GV-EV mixed demand. In particular, there is a dearth of empirical evidence based on real data regarding the reshaping of congestion patterns due to differences in operating costs and path redistribution.

In recent years, EVs have achieved high penetration rates in multiple regions worldwide. For example, the International Energy Agency \citep{IEA2025EV} indicates that in 2024, EV sales in China accounted for nearly 50\% of the country's total new vehicle sales. However, existing attention to the impact of widespread EV penetration on transportation systems remains insufficient. Therefore, to address these gaps, this paper investigates urban congestion patterns under high EV penetration. We develop an MUE assignment model capable of characterizing the competitive relationship between GVs and EVs, and design a convergent iterative algorithm based on a dual solution framework for convex optimization. Subsequently, a set of congestion pattern evaluation metrics is devised, comprising average travel time, potential savings, volume over capacity (VOC), road utilization rate, link congested time, and difference in delay factor. Finally, taking 10 representative U.S. cities as examples, the study systematically simulates and compares the evolution trends of transportation systems under different EV penetration scenarios based on real road networks and commuting origin-destination (OD) demand data. The research results can inform differentiated policy recommendations for sustainable urban transportation planning, road design, and congestion management. Therefore, the main contributions of this paper include the following:

{
\begin{enumerate}
\item \textbf{Fixed-class MUE model with convergent dual algorithm.} Unlike classical multi-class models where users freely switch between modes, we formulate a fixed-class MUE model specifically tailored for GV-EV scenarios, where the key distinction lies in treating vehicle class as exogenously determined by vehicle ownership rather than endogenous mode choice. We prove the existence and uniqueness of equilibrium, and develop a bi-conjugate Frank-Wolfe algorithm with guaranteed convergence.

\item \textbf{Multi-dimensional congestion pattern metrics.} We design six complementary metrics---average travel time, potential savings, VOC, road utilization rate, link congested time, and difference in delay factor---to characterize congestion evolution from system-level efficiency, link-level saturation, and spatial distribution perspectives.

\item \textbf{Empirical analysis across 10 heterogeneous U.S. cities.} Using real OSM road networks and Census block-level commuting demand, we reveal that: (i) full EV penetration reduces system travel time by 2.27\%--10.78\%; (ii) benefits concentrate in the 0--0.5 penetration range with diminishing returns thereafter; (iii) congestion reshaping patterns vary substantially with urban topology.
\end{enumerate}
}

The remainder of this paper is organized as follows: Section 2 reviews existing literature. Section 3 introduces the methodology, including the problem description, the MUE assignment model, and evaluation metrics. Section 4 presents an analysis using 10 U.S. cities as case studies. Section 5 concludes the paper and outlines future work.

\section{Literature Review}

\subsection{The Impact of EV Penetration on Urban Transportation Systems}
In recent years, against the backdrop of global carbon peaking and carbon neutrality goals and energy transition, the promotion of EVs has been regarded as a key driver for reducing emissions in the transportation sector and transforming the energy structure. A large body of research shows that, compared with conventional GVs, EVs offer substantial advantages in terms of life-cycle greenhouse gas emissions, primary energy consumption, and local air pollution, and therefore constitute one of the key technological pathways for deep decarbonization of the transport sector \citep{li2025study, kobashi2025enabling}. The rapid penetration of EVs is primarily driven by two groups of factors. On the one hand, improvements in electric drivetrain efficiency, the gradual decarbonization of power systems, and the continuous decline in battery costs have jointly reduced the generalized cost per kilometer of travel, enabling EVs to approach—and in some cases surpass—GVs in terms of life-cycle economic competitiveness \citep{kersey2022rapid}. On the other hand, governments around the world have deployed a mix of policies—such as purchase subsidies, vehicle registration or purchase tax exemptions, preferential road access, and large-scale investment in charging infrastructure—that substantially lower the adoption threshold for consumers and have driven EV ownership onto a rapidly increasing trajectory \citep{pardo2021sustainable, tilly2024sustainable}. Therefore, economic incentives, infrastructure accessibility, and the level of urbanization jointly shape the spatiotemporal diffusion of EVs \citep{maybury2022mathematical, hu2025gas}.

As EV penetration has risen in recent years, an increasing number of studies have begun to examine the feedback effects of EV adoption on urban transport systems—particularly how EVs may alter people's travel behavior (trip frequency, route/mode choice) and how those changes can affect traffic flow distribution, congestion patterns, and transportation system efficiency. \citet{morton2022effect} suggested that the reduction in operating and ownership costs of EVs will induce users to increase travel frequency or extend travel distances, thereby increasing the overall traffic demand density. Meanwhile, research has revealed that the introduction of EVs may alter the structure of travel modes.  For example, \citet{kosmidis2023electric} stated that when private EVs, slow mobility (walking/cycling), and public transit coexist in a multi-modal system, some trips that might otherwise be made by slow mobility or public transit will be replaced by EVs. In turn, the change in travel behavior will influence or reshape traffic congestion patterns, which are not always positive or consistent. On the one hand, if urban planning integrates charging infrastructure deployment, public transit corridor development, and travel demand management, the introduction of EVs could alleviate some of the congestion and emissions pressures traditionally dominated by GVs. For example, \citet{taamneh2025prospects} discovered that integrating EVs—particularly autonomous EVs—into urban transportation systems has the potential to improve traffic efficiency, reduce congestion, and enhance travel safety. On the other hand, if the widespread penetration of EVs lacks appropriate institutional design—such as traffic demand management, equitable infrastructure deployment, and charging incentive mechanisms—it may lead to a dual increase in traffic volume and energy system load due to higher travel frequency and mode substitution. \citet{grigorev2021will} indicated that under scenarios of high EV penetration rates coupled with inadequate or unsynchronized expansion of charging infrastructure, new congestion and energy consumption bottlenecks are likely to emerge at fast-charging stations and within transportation networks.

Although an increasing number of studies have begun to focus on the impact of EV penetration on urban transportation systems, several key limitations remain that require further exploration. Most studies remain at the level of qualitative judgments or macro-level trend descriptions, lacking detailed and quantitative insights into the micro-mechanisms of congestion pattern reshaping. In particular, how the traffic flow redistribution caused by EV penetration specifically alters the operational conditions of each road segment or key node within the network, as well as the spatial distribution and evolution of these changes, has not been fully explored. More importantly, the potential differential and diverse impacts brought about by the inherent heterogeneity of cities (such as terrain conditions, network topology, road hierarchy configuration, and initial congestion patterns) are generally overlooked. A systematic study integrating real, complex road network systems from multiple cities is still absent, limiting the generalizability of the conclusions and their policy implications.

\subsection{Multi-Class Traffic Assignment Model}

The cornerstone of traffic assignment research is the UE theory, which is widely used to allocate traffic demand under given traffic flow conditions \citep{van1980most}. The UE model is based on Wardrop's First Principle \citep{wardrop1952correspondence}, which states that traffic flows eventually stabilize when each traveler seeks to minimize their individual travel cost. \citet{daganzo1977stochastic} proposed the SUE traffic assignment model within a stochastic utility framework, incorporating travelers' random perceptions of path travel times into path selection probabilities. \citet{lou2010robust} applied the satisfaction mechanism to route choice problems (where travelers select acceptable but not necessarily shortest paths) and developed the boundedly rational UE (BRUE) flow model. \citet{zhang2015modeling} developed the inertial UE (IUE) model, which emphasizes the stability of travel habits, i.e., travelers may repeatedly use a particular route due to familiarity or convenience rather than performing a global optimization each time. \citet{ding2023status} presented a status quo-dependent UE (SDUE) model that simultaneously accounts for travelers' bounded rationality, behavioral inertia, and adaptive value of time.

To address increasingly complex urban transportation systems, the research paradigm for traffic assignment models has expanded from single-class to multi-class integrated analysis. This multi-class assignment model aims to systematically characterize the competitive, cooperative, and infrastructure-sharing relationships among various classes such as private cars, public transit, and non-motorized transportation. By constructing a generalized cost function incorporating attributes such as time, monetary cost, and comfort, these models can reveal the path choice behavior of heterogeneous travelers within multi-modal networks. The extension to multiple user classes was pioneered by \citet{dafermos11969traffic}, who established the theoretical foundations for heterogeneous user equilibrium. \citet{florian1977traffic} proposed a joint car-public transit traffic equilibrium model, which separately models car traffic flow and public transit ridership while characterizing their interactions through shared road segment impedance, making it one of the seminal early contributions to multi-class traffic assignment. \citet{florian2002multi} further developed a network equilibrium model for multi-class and variable demand scenarios, uniformly representing various modes such as private cars and public transit as a variational inequality problem. Subsequently, the multi-class equilibrium model has been continuously expanded upon the traditional “car-public transit” framework, gradually enabling the joint consideration of multiple travel modes—including walking, public transit, private cars, and ride-hailing services—along with their transfer combinations \citep{sun2021multi, najmi2023multimodal, liu2025integrating}.  This advancement provides analytical methods for studying traffic diversion and congestion mechanisms within multi-class transportation systems. Most of the aforementioned multi-class traffic assignment models focus on the coupled analysis of public transit with other modes (such as private cars, non-motorized travel, etc.). However, coordinating the fixed schedule operation of public transit with the continuous flow of private cars remains a significant challenge.

With the widespread adoption of EVs, an increasing number of studies in recent years have begun to examine their impacts on urban transportation systems, and incorporating EVs into traffic assignment models has emerged as a major research focus. Early work by \citet{dong2014electric} and \citet{he2014network} established frameworks for incorporating range anxiety into traffic assignment, showing that limited driving range fundamentally alters EV route choice behavior.\ \citet{davazdah2023nonlinear} proposed a mixed complementary traffic assignment model applicable to networks involving EVs and GVs, where the demand for each vehicle category depends on path characteristics and the availability of charging stations along the route. \citet{zeng2025optimizing} developed a joint model for GV-EV mixed flow and charging station siting based on SUE, analyzing the impact of different siting schemes on balanced flow distribution, travel costs, and environmental benefits. \citet{zhang2025impact} constructed a traffic assignment model targeting GV-EV mixed networks to quantitatively assess the impact of EV penetration on balanced traffic distribution and road network carbon emissions. Their findings indicate that under certain conditions, EVs may alter spatial patterns of congestion and emission hotspots through path redistribution. Our model abstracts away range constraints, which is reasonable for urban commute trips (typically $<$30 km) well within modern EV ranges (300+ km). A second major line of research adopts the perspective of coupled transportation–power systems \citep{sheng2023emission}, embedding EV traffic assignment into distribution network planning or operational optimization. Within this framework, UE or SUE models are employed to characterize EVs’ route and charging choices on the road network, and the resulting redistribution of traffic flows is then used to assess impacts on distribution network loading, system resilience, and renewable energy integration \citep{qiao2022distributed, liang2024vehicle}.

In summary, current research on the integration of EVs into traffic flow assignment remains largely in a relatively early exploratory phase. On the one hand, existing research predominantly focuses on the potential impact of EV penetration on travel demand, whereas our study concentrates more on the relatively stable OD structure of daily trips within typical commuting scenarios dominated by household charging. Previous frameworks do not account for fixed user class scenarios, and they also offer relatively limited insights into the intricate reshaping of congestion patterns resulting from the coexistence and games between EVs and conventional GVs within the same road network. On the other hand, current simulation validations remain confined to small to medium scale idealized/designed test networks, failing to adequately capture the complexity of congestion evolution at actual urban scales. Against this backdrop, how high EV penetration alters traffic flow spatial distribution and congestion patterns within real urban road networks characterized by significant heterogeneity remains missing from systematic quantitative assessments. Consequently, the applicability and guidance value of related conclusions in urban transportation management and planning practices are somewhat constrained.

\section{Methodology}

Before proceeding, we formally define three key concepts used throughout this paper:

\begin{definition}[Congestion Pattern]\label{def:congestion_pattern} A \textbf{congestion pattern} refers to a multi-dimensional characterization of traffic congestion across a road network, comprising: (i) system-level efficiency measured by average travel time $T_{MUE}$ and potential savings $P_s$; (ii) link-level saturation measured by VOC distribution; and (iii) spatial distribution measured by road utilization rate, link congested time, and difference in delay factor across road hierarchies. We present all three dimensions (travel time, VOC, spatial distribution) separately rather than as a composite index, allowing readers to assess trade-offs. When travel time decreases but VOC increases, we interpret this as “congestion redistribution”---traffic concentrates on fewer high-capacity links that still provide faster travel than dispersed local roads. System efficiency improves even as backbone utilization rises.
\end{definition}

\begin{definition}[Plateau Phase]\label{def:plateau_phase}
A \textbf{plateau phase} is an interval $[R_e^{(i)}, R_e^{(i+1)}]$ of EV penetration rates where the marginal improvement in system travel time is negligible, formally defined as:
\begin{equation}
    \left| \frac{\partial T_{MUE}}{\partial R_e} \right| < \epsilon, \quad \forall R_e \in [R_e^{(i)}, R_e^{(i+1)}],
\end{equation}
where $\epsilon > 0$ is a threshold (empirically $\epsilon = 0.07$ min per 10\% penetration is chosen based on the empirical distribution of sensitivity values across our 10-city sample, representing approximately one standard deviation below the mean gradient magnitude). Plateau phases typically span 5--15\% of the penetration range before transition to the next active set configuration, as observed empirically in our 10-city sample. This is due to the sequential route displacement mechanism: EVs progressively displace GVs route-by-route, and even minimal EV introduction can fully capture certain cost-advantaged routes. Consequently, pure plateaus (with strictly zero marginal benefit) are rare; instead, most penetration intervals exhibit at least some incremental displacement.
\end{definition}

\begin{definition}[Transition Zone]\label{def:transition_zone}
The \textbf{transition zone} (also termed ``cascading benefit zone'') is the penetration range where the system exhibits pronounced, nonlinear efficiency gains with $|\partial T_{MUE}/\partial R_e| > 3\epsilon$. The underlying mechanism is as follows: at $R_e = 0$, no displacement occurs and marginal benefit is zero; as $R_e$ increases past a critical threshold, a cascading effect emerges---each additional EV simultaneously displaces GV flows across multiple congested links, triggering network-wide flow redistribution and rapidly amplifying system benefits. Beyond this zone, the network approaches EV-dominated equilibrium, and marginal returns diminish as fewer GV routes remain available for displacement.
\end{definition}

We note that these categories represent idealized types on a continuum. Cities may exhibit hybrid characteristics, and the classification is sensitive to threshold choice. The trichotomy serves as a useful heuristic for policy discussion rather than a strict partition.

\subsection{Problem Description}

The traditional UE model typically does not explicitly account for heterogeneity across vehicle types, thus struggling to reflect how cost differences between GVs and EVs affect urban congestion patterns. However, the operational cost difference between GVs and EVs has become a key factor influencing route choices for users. To elucidate the core mechanism of this phenomenon, we first introduce a typical dual-route network problem (Fig. \ref{fig:1}). As shown in Fig. \ref{fig:1}a), the two routes, $a$ and $b$, from the common origin (home) to the destination (work) possess distinct characteristic attributes: length, free-flow speed, and capacity. Route $a$ is 6.0 miles in length, with a design capacity of 120 vehicles per hour (vph) and a free-flow speed of 30 miles per hour (mph). Route $b$ is 7.5 miles in length, with a design capacity of 200 vph and a free-flow speed of 40 mph. We assume that travel time costs for routes $a$ and $b$ are determined by a simple linear Bureau of Public Roads (BPR) function \citep{manual1964bureau}, given by:
\begin{equation}
    t\left ( x_{a}  \right ) =\frac{6}{30}\left ( 1+\left ( \frac{x_{a}}{120}  \right )  \right ),  
\end{equation}
\begin{equation}
t\left ( x_{b}  \right ) =\frac{7.5}{40}\left ( 1+\left ( \frac{x_{b}}{200}  \right )  \right ), 
\end{equation}
where $t\left ( x_{a}  \right )$ and $t\left ( x_{b}  \right )$ denote the travel time costs for routes $a$ and $b$, respectively; $x_{a}$ and $x_{b}$  denote the traffic flow on routes $a$ and $b$, respectively.

Consider a total demand of $q = 100$ vehicles for this OD pair. Under the traditional UE model (minimizing travel time only), the equilibrium allocation assigns 31.2 vehicles to route $a$ and 68.8 vehicles to route $b$, with both routes achieving equal travel times of 15.12 minutes. We now extend the model to incorporate vehicle operating costs. Let $C_m$ denote the per-mile operating cost for vehicle class $m \in \{gv, ev\}$, and let $\gamma$ represent the value of time (\$/min). The generalized travel cost on route $k$ becomes:
\begin{equation}
    C_k^m = \gamma \cdot t_k(x_k) + C_m \cdot l_k,
\end{equation}
where $l_k$ is the route length. For GVs with $C_{gv} = \$0.6$/mile and $\gamma = \$0.3$/min, the equilibrium shifts: route $a$ now carries 50.4 vehicles (travel time: 17.04 min) while route $b$ carries 49.6 vehicles (travel time: 14.04 min), yielding a system average travel time of 15.55 minutes.

To quantify the effect of gasoline price fluctuations (such as those triggered by geographic location, economic policies, or supply-side shocks) on traffic flow distribution, Fig. \ref{fig:1}b) and \ref{fig:1}c) present the UE assignment results for the dual-route system under varying GV costs. For every \$0.1 per mile increase in GV cost, the number of vehicles on route $a$ increases by 3.2, with a corresponding travel time increase of 0.32 minutes, while route $b$ experiences a reduction of 3.2 vehicles and a decrease in travel time of 0.18 minutes. When the initial GV cost is 0, route $b$ carries 68.8\% of the traffic flow due to its shorter baseline travel time. As GV costs rise, drivers increasingly prioritize the cost-effectiveness of travel when selecting routes. Consequently, the advantage of the shorter route $a$ becomes more pronounced, leading to a potential shift of more vehicles to route $a$, while route $b$ is burdened with higher travel costs. This shift in route selection triggers changes in congestion patterns on route $a$ and route $b$, subsequently affecting the overall travel time: on route $a$, the travel time increases due to intensified congestion caused by the growing number of vehicles, while route $b$ experiences a decline in travel time as reduced traffic volume facilitates smoother flow conditions. The system's average travel time demonstrates a progressive increase with rising GV costs. Under conditions of low GV costs, travel time exhibits a relatively gradual increase, in contrast to the highly pronounced temporal expansion observed under elevated GV cost scenarios. The variation in potential savings similarly reflects the diminishing returns of GV cost increases on the travel efficiency of the transportation system. This metric quantifies the potential travel time efficiency gain as a percentage, calculated by comparing the system's performance at a specific GV cost against the benchmark of maximum and minimum average travel times. At lower GV costs (less than \$0.2/mile), potential savings reach high levels approaching 100\%. As GV costs continue to rise to higher levels (particularly above \$0.8/mile), the percent savings drop sharply.
\begin{figure}[ht]
    \centering
    \includegraphics[width=\textwidth]{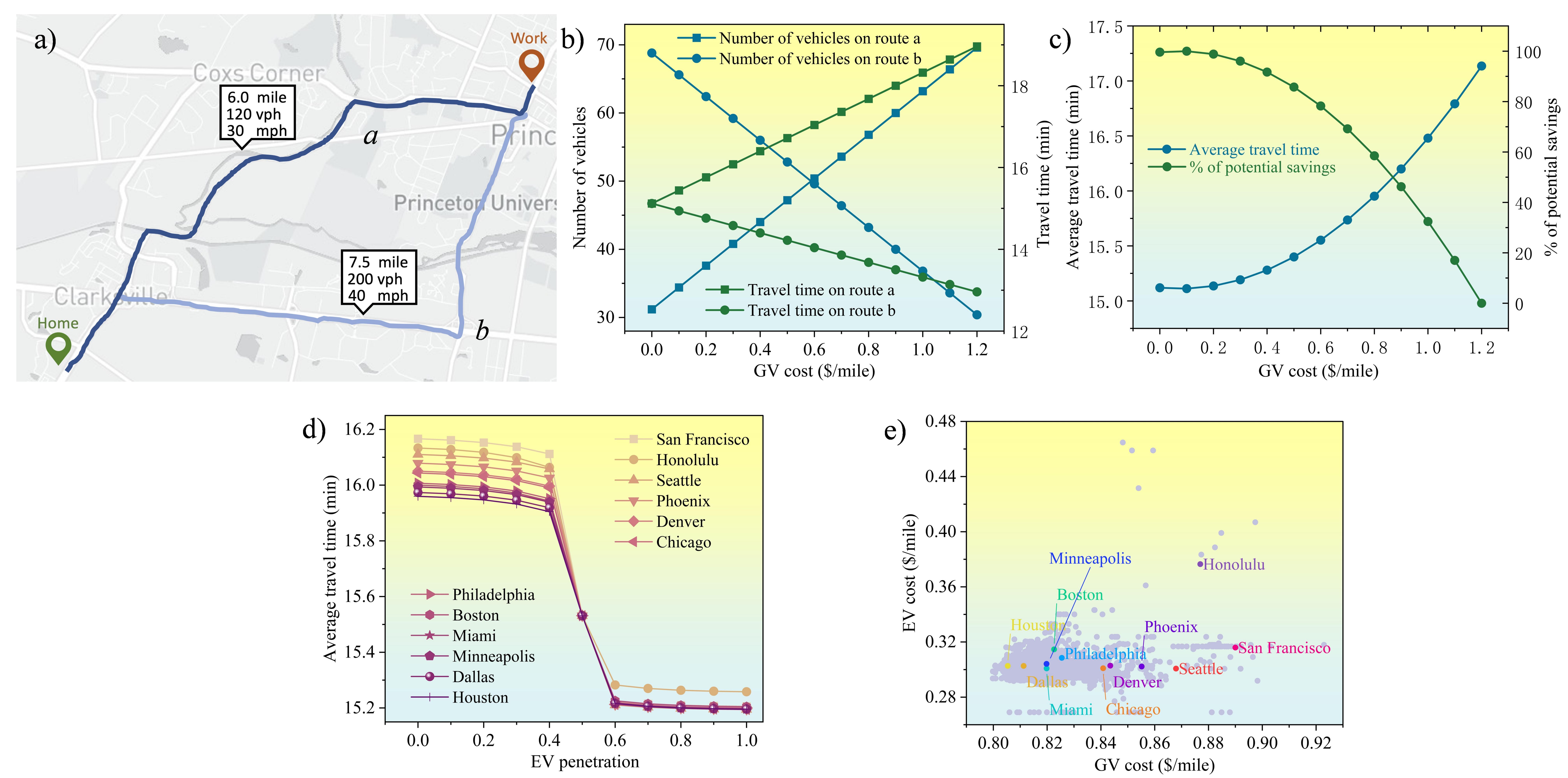}
    \caption{Illustrative dual-route network and sensitivity analysis. (a) Network topology with heterogeneous route attributes; b) Changes in the number of vehicles and travel time on the two routes as the GV cost varies. c) Changes in the average travel time and percentage of potential savings for the system as the GV cost varies. d) The impact of EV penetration on average travel time for the system in a set of representative U.S. cities. e) The distribution of GV and EV costs across U.S. cities.}
    \label{fig:1}
\end{figure}

Furthermore, what would happen if the vehicles were switched to EVs? We systematically estimate the per-mile operating cost for GVs and EVs across most U.S. cities, as shown in Fig. \ref{fig:1}e). Quantitative analysis demonstrates a clear economic advantage for EVs, exhibiting per-mile costs between \$0.269 and \$0.465, while conventional GVs maintain significantly higher operating costs ranging from \$0.799 to \$0.923 per mile. The data analysis also reveals significant differences among cities in terms of the operating costs of GVs and EVs. For example, in Houston, both vehicle types exhibit relatively low costs; in San Francisco, GV costs are notably higher; while in Honolulu, both GV and EV costs are simultaneously high. Fig. \ref{fig:2} tabulates the spatial distribution characteristics of GV and EV costs in the United States. Fig. \ref{fig:2}a) illustrates the distribution patterns of GV and EV costs, where variations in color intensity represent the difference in per-mile driving costs for the two types of vehicles. Certain regions in the west (particularly California) and parts of the northeast have higher costs for both GVs and EVs. High gasoline prices, high electricity prices, and high living costs in these regions result in high operating costs for transportation vehicles. In parts of the southern and midwestern regions, such as Texas, Arkansas, and Tennessee, the costs of GVs and EVs are relatively low. These states are located in energy-rich regions, particularly Texas, one of the largest oil-producing states in the United States with abundant supplies of natural gas and oil at comparatively low prices, as well as a fairly low living cost. We also observe that the two enclaves—Hawaii and Alaska—exhibit significantly higher costs for both types of vehicles, primarily due to their reliance on imported energy supplies. Fig. \ref{fig:2} b) presents the regional aggregated differences in GV and EV costs through additive and subtractive methods. Notably, several western regions (including California, Oregon, and Washington) exhibit higher total costs for both GV and EV operations, while maintaining a pronounced cost differential between the two propulsion systems, as visually represented by reddish hues in the mapping. In the southeastern and south-central regions, EV operational costs are markedly lower than those of GVs, while maintaining comparatively low total costs—a configuration cartographically represented by distinctive light-blue hues in the corresponding visualization. Fig. \ref{fig:2}c) shows the distribution of the ratio ($R$) between GV cost and EV cost. Regional analysis reveals elevated GV-to-EV cost ratios across select western and northern territories, whereas certain cities in the southern/eastern regions and Alaska exhibit lower cost ratios. Fig. \ref{fig:2}d) visualizes the percentage distribution of potential cost savings after converting GV costs to EV costs, revealing the electrification potential of cities. Regions with high potential savings are typically concentrated in parts of the northeast and west, such as Pennsylvania, Washington, and Oregon. As EVs become more widespread, these areas stand to achieve substantial transportation cost savings. In certain east-central and southern states such as Michigan and Texas, the potential savings are limited, with the economic benefits of electrification being less pronounced than in other regions.

\begin{figure}[ht]
    \centering
    \includegraphics[width=\textwidth]{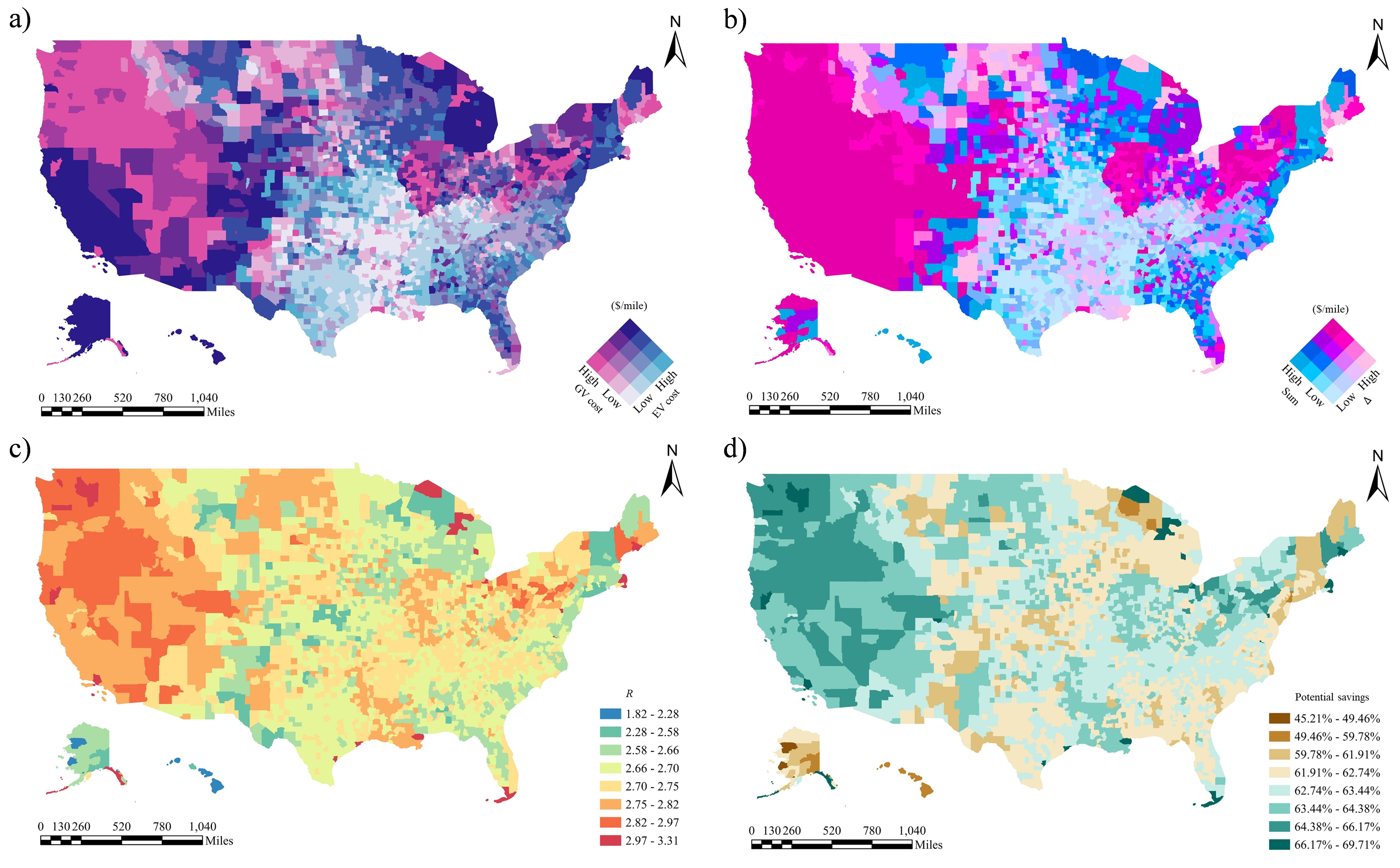}
    \caption{The spatial distribution of GV and EV costs in the United States. a) Spatial distribution of GV and EV costs. b) Spatial distribution of the sum and subtraction of GV and EV costs. c) Spatial distribution of the ratio of GV cost to EV cost. d) Spatial distribution of potential savings from converting GV cost to EV cost.}
    \label{fig:2}
\end{figure}

For a set of representative U.S. cities, we incorporate EVs into the dual-route system to evaluate the impact of EV integration on traffic congestion under different cost structures. The evolution of the system's average travel time under different EV penetration scenarios is shown in Fig. \ref{fig:1}d). For these cities, the average travel time shows a significant downward trend as EV penetration increases, thus alleviating traffic congestion. San Francisco, Honolulu, and Seattle show higher initial average travel times, while Honolulu has the highest final average travel time, primarily due to differences in GV and EV costs across these cities. Consistent with Definition \ref{def:plateau_phase}, two plateau phases are observed: at low penetration ($R_e < 0.3$) where GVs dominate the equilibrium, and at high penetration ($R_e > 0.7$) where EVs dominate. In these regimes, marginal EV additions yield negligible travel time improvements. The intermediate range $R_e \in [0.3, 0.7]$ corresponds to the transition zone (Definition \ref{def:transition_zone}), where the system exhibits pronounced efficiency gains as path saturation events trigger network-wide flow redistribution. 

The preceding analysis demonstrates that operating cost heterogeneity between GVs and EVs fundamentally alters equilibrium flow patterns even in simple networks. While this illustrative example uses $\beta=1$ for analytical tractability, the qualitative insights---that cost differentials induce flow redistribution and that network topology mediates congestion response---generalize to realistic BPR parameters ($\beta \in [1.2, 2]$) as demonstrated in our empirical analysis (Section \ref{Case_Study}). This motivates several research questions for real-world urban networks: (i) How does EV penetration affect system-level travel efficiency across cities with heterogeneous network topologies? (ii) Do congestion benefits exhibit nonlinear patterns (e.g., plateau phases) as penetration increases? (iii) How does congestion redistribute spatially---does it concentrate on specific corridor types? To address these questions, we develop a fixed-class MUE model that explicitly captures the GV-EV cost differential within real urban road networks. 

\subsection{UE Assignment Model}
The UE principle \citep{sheffi1985urban} in traffic assignment models serves as a classic theoretical framework for analyzing traffic distribution on road networks. This model is typically used to forecast traffic flow on each road based on given traffic demand \citep{feng2021reconstructing}. When the equilibrium state is reached, the travel costs on all utilized paths become equal and minimized. That is, for any given OD pair, the travel costs on unused paths are greater than or equal to those on utilized paths. Therefore, the mathematical expression of the UE assignment model can be represented as:
\begin{equation}\label{3}
  \min_{x} ~ Z\left ( x \right ) =\sum_{a}\int_{0}^{x_{a} } t_{a} \left ( w \right ) dw,  
\end{equation}
subject to
\begin{align}
\sum_{k \in \psi_{rs}} f_{k}^{rs} &= q^{rs},   \quad \forall r \in R, \forall s \in S, 
\label{4}\\
f_{k}^{rs} &\ge 0,  \quad \forall k \in \psi_{rs},\ r \in R,\ \forall s \in S, \label{5}\\
x_a        &= \sum_{r} \sum_{s} \sum_{k} f_{k}^{rs}\,\delta_{ka}^{rs},\quad \forall k \in \psi_{rs},\ r \in R,\ \forall s \in S, \label{6}
\end{align}
where $x$ is the set of the traffic flow volume on all roads (links) in the network; $x_{a}$ represents the traffic flow volume on road $a$; $t_{a}$ denotes the road traffic impedance function, i.e., the travel time on road $a$; $q^{rs}$ represents the travel demand from the origin $r$ to destination $s$; $f_{k}^{rs}$ denotes the traffic flow volume of the $k^{th}$ route (path) between $r$ and $s$; $\psi_{rs}$ denotes the set of all routes between $r$ and $s$; $R$ and $S$ denote the sets of all origins and destinations in the traffic network, respectively; and  $\delta_{ka}^{rs}$ is a Boolean function indicating whether the $k^{th}$ route between $r$ and $s$ passes through road $a$. The objective function defined in Eq. \eqref{3} minimizes the sum of the integral of travel time costs for all roads.  The flow conservation constraint is defined by Eq. \eqref{4}. The relationship between road flow and route flow is represented by Eq. \eqref{6}. The constraint in Eq. \eqref{5} ensures that the traffic flow on each route is non-negative. The above equation can determine the traffic flow on each road based on the traffic demand. Therefore, the total travel time for the $k^{th}$ route can be calculated by summing the travel times along this route for each road. We denote it as $c_{k}^{rs}$, which can be specifically expressed as:
\begin{equation}\label{7}
 c_{k}^{rs} = \sum_{a} t_a \delta_{ka}^{rs}, \quad \forall k \in \psi_{rs}, \ r \in R, \ \forall s \in S.
\end{equation}

The above UE model has been proven to be a strictly convex problem, yielding a unique solution for traffic flow assignment on roads \citep{sheffi1985urban}. However, it is extremely challenging to find the optimal solution for the assignment problem in large-scale complex urban traffic networks, as the computational cost increases substantially. This equilibrium problem solution method can be traced back to the convex combination algorithm proposed by \citet{frank1956algorithm}. Based on this foundation, \citet{florian1976application} made improvements and refinements to the algorithmic framework. To reduce computational burden, a bi-conjugate Frank-Wolfe algorithm \citep{mitradjieva2013stiff,zill2019toll} is adopted to find the optimal solution for the assignment model. 

\subsection{MUE Assignment Model for Mixed GVs and EVs}

We adopt a comparative static framework, computing equilibrium for each exogenously specified penetration rate $R_e$. This approach answers the question: `Given a particular GV-EV fleet composition, what traffic pattern emerges?' rather than modeling the endogenous adoption dynamics. Following standard practice in transportation modeling, we assume class-specific but intra-class homogeneous value of time. Future work may incorporate heterogeneous user preferences.

\subsubsection{Motivation and Distinction from Classical Multi-class Models}

The integration of electric vehicles into urban transportation networks necessitates a fundamental rethinking of traffic assignment models. Classical multi-class user equilibrium (MUE) models \citep{daganzo1983stochastic,yang2005mathematical} typically consider multiple user classes distinguished by their value of time or route choice preferences, where users may freely switch between classes based on their individual utility maximization. Such models implicitly assume modal flexibility—a commuter choosing between driving alone versus carpooling, or selecting different departure times.

In contrast, the GV-EV mixed traffic scenario presents a fundamentally different structure: vehicle class membership is exogenously fixed. A GV owner cannot instantaneously become an EV user in response to changing traffic conditions, and vice versa. This distinction has profound implications for the mathematical structure of the equilibrium problem:

\begin{itemize}
    \item \textbf{Classical Multi-class UE:} Users optimize over both route choice and class membership. The equilibrium condition requires that no user can reduce their cost by switching either their route or their class.
    
    \item \textbf{Fixed-Class MUE (This Work):} Users optimize over route choice only, with class membership predetermined. Each class has its own demand constraint that must be satisfied independently.
\end{itemize}

This fixed-class structure introduces additional complexity: the two user populations interact through shared road capacity but maintain separate feasibility constraints. The resulting equilibrium is characterized by coupled variational inequalities rather than a single unified optimization problem, and standard convexity arguments from single-class UE do not directly apply. We emphasize that “fixed-class” refers to the equilibrium computation for a given $R_e$, not the absence of long-term adoption dynamics. Our analysis is comparative static: we compare equilibria across different $R_e$ values, representing different stages of fleet transition. The dynamics of how $R_e$ evolves over time (influenced by policy, cost, etc.) is beyond this paper's scope.

\subsubsection{Cost Structure for GVs and EVs}

The operational cost heterogeneity between GVs and EVs constitutes the primary driver of differentiated route choice behavior. We model the per-kilometer cost for each vehicle class comprehensively.

The GV cost $C_{gv}$ (\$/km) comprises fuel cost $C_{gv}^{fuel}$, maintenance $C_{gv}^{maint}$, fixed costs (registration, taxes, loans) $C_{gv}^{fix}$, depreciation $C_{gv}^{dep}$, insurance $C_{gv}^{ins}$, additional taxes $C_{gv}^{add}$, and environmental externalities $C_{gv}^{env}$, as well as the conversion factor ${r_{dis}}$ (miles to kilometers). Among them, the fuel cost component is:
\begin{equation}\label{eq:fuel_cost}
    C_{gv}^{fuel} = \frac{P_{gas}}{MPG_{gv}},
\end{equation}
with $P_{gas}$ denoting gasoline price (\$/gallon), and $MPG_{gv}$ denoting fuel efficiency (typically 25 miles/gallon).

Similarly, the EV cost $C_{ev}$ (\$/km) includes electricity cost $C_{ev}^{ele}$, maintenance $C_{ev}^{maint}$, fixed costs $C_{ev}^{fix}$, depreciation $C_{ev}^{dep}$, insurance $C_{ev}^{ins}$, additional taxes $C_{ev}^{add}$, environmental costs $C_{ev}^{env}$, and subsidy allocations $C_{ev}^{sub}$, as well as the conversion factor ${r_{dis}}$. Among them, the electricity cost component is:
\begin{equation}\label{eq:ele_cost}
    C_{ev}^{ele} = \frac{P_{ele} \cdot \kappa_{gal}}{MPGe_{ev}},
\end{equation}
with $P_{ele}$ denoting electricity price (\$/kWh), $\kappa_{gal} = 33.7$ kWh/gallon (energy equivalence), and $MPGe_{ev}$ denoting EV efficiency (typically 110 miles/gallon-equivalent). The assumed fuel efficiency of 25 MPG for GVs and 110 MPGe for EVs reflects the U.S. fleet average per EPA fuel economy data\footnote{\url{https://www.fueleconomy.gov}}.

\subsubsection{Mathematical Formulation of the Fixed-Class MUE}

We now develop a rigorous mathematical framework for the fixed-class MUE problem. Consider a transportation network $\mathcal{G} = (\mathcal{N}, \mathcal{A})$ with node set $\mathcal{N}$ and link set $\mathcal{A}$. Let $\mathcal{W}$ denote the set of OD pairs, and for each $(r,s) \in \mathcal{W}$, let $\psi_{rs}$ denote the set of available paths.

\textbf{Decision Variables.} For vehicle class $m \in \mathcal{M} := \{gv, ev\}$, let $f_m = (f_{k,m}^{rs})_{k \in \psi_{rs}, (r,s) \in \mathcal{W}}$ denote the path flow vector. We concatenate these into a single vector:
\begin{equation}
    f = \begin{pmatrix} f_{gv} \\ f_{ev} \end{pmatrix} \in \mathbb{R}^{n_{gv} + n_{ev}} =: \mathbb{R}^n,
\end{equation}
where $n_m = \sum_{(r,s) \in \mathcal{W}} |\psi_{rs}|$ is the total number of path variables for class $m$.

\textbf{Feasible Region.} Unlike classical multi-class models where total OD demand can be freely distributed across classes, our model requires \textit{separate demand conservation} for each class:
\begin{equation}\label{eq:demand_constraint}
    \sum_{k \in \psi_{rs}} f_{k,m}^{rs} = d_m^{rs}, \quad \forall (r,s) \in \mathcal{W}, \quad \forall m \in \mathcal{M},
\end{equation}
where $d_m^{rs}$ is the fixed travel demand for class $m$ between OD pair $(r,s)$. This can be written compactly as:
\begin{equation}
    B_m f_m = d_m, \quad m \in \{gv, ev\},
\end{equation}
where $B_m$ is the OD-path incidence matrix for class $m$. The feasible region decomposes as:
\begin{equation}
    \Omega = \Omega_{gv} \times \Omega_{ev}, \quad \text{where} \quad \Omega_m := \left\{ f_m \in \mathbb{R}^{n_m}_+ : B_m f_m = d_m \right\}.
\end{equation}

\textbf{Link Flow Aggregation.} The total flow on link $a \in \mathcal{A}$ aggregates contributions from both classes:
\begin{equation}\label{eq:link_flow}
    x_a = \sum_{m \in \mathcal{M}} \sum_{(r,s) \in \mathcal{W}} \sum_{k \in \psi_{rs}} f_{k,m}^{rs} \cdot \delta_{ka}^{rs} = \Delta_{gv} f_{gv} + \Delta_{ev} f_{ev},
\end{equation}
where $\Delta_m$ is the link-path incidence matrix for class $m$, and $\delta_{ka}^{rs} = 1$ if path $k$ uses link $a$.

\textbf{Generalized Path Cost.} The travel cost on link $a$ follows the BPR function:
\begin{equation}\label{eq:BPR}
    t_a(x_a) = t_a^0 \left[ 1 + \alpha \left( \frac{x_a}{c_a} \right)^\beta \right],
\end{equation}
where $t_a^0$ is the free-flow travel time, $c_a$ is the capacity, and $\alpha, \beta > 0$ are calibration parameters. The generalized path cost for class $m$ on path $k$ between $(r,s)$ is:
\begin{equation}\label{eq:path_cost}
    C_{k,m}^{rs}(f) = \sum_{a \in \mathcal{A}} \left( \gamma_m \cdot t_a(x_a(f)) + C_m \cdot l_a \right) \delta_{ka}^{rs},
\end{equation}
where $\gamma_m$ is the value of time for class $m$, $C_m$ is the per-kilometer operating cost, and $l_a$ is the length of link $a$. We assume $\gamma_{gv} = \gamma_{ev} = \gamma$ (common value of time across vehicle classes), consistent with the assumption that GV vs. EV choice reflects vehicle ownership rather than income differentiation.

Define the path cost operator:
\begin{equation}
    F(f) = \begin{pmatrix} F_{gv}(f) \\ F_{ev}(f) \end{pmatrix}, \quad \text{where} \quad (F_m(f))_k^{rs} = C_{k,m}^{rs}(f).
\end{equation}

\textbf{Wardrop Equilibrium Conditions.} The fixed-class MUE is characterized by the following conditions: for each class $m$ and OD pair $(r,s)$, there exists an equilibrium cost $\pi_m^{rs,*}$ such that:
\begin{equation}\label{eq:wardrop}
    \begin{cases}
        C_{k,m}^{rs}(f^*) = \pi_m^{rs,*}, & \text{if } f_{k,m}^{rs,*} > 0, \\
        C_{k,m}^{rs}(f^*) \geq \pi_m^{rs,*}, & \text{if } f_{k,m}^{rs,*} = 0.
    \end{cases}
\end{equation}

\begin{proposition}[Variational Inequality Characterization]\label{prop:VI}
The flow vector $f^* \in \Omega$ is a fixed-class MUE if and only if it solves the variational inequality $\text{VI}(F, \Omega)$:
\begin{equation}\label{eq:VI}
    \langle F(f^*), f - f^* \rangle \geq 0, \quad \forall f \in \Omega.
\end{equation}
\end{proposition}

\begin{proof}
The Wardrop conditions \eqref{eq:wardrop} are equivalent to:
\begin{equation}
    \sum_{m \in \mathcal{M}} \sum_{(r,s) \in \mathcal{W}} \sum_{k \in \psi_{rs}} C_{k,m}^{rs}(f^*) \cdot (f_{k,m}^{rs} - f_{k,m}^{rs,*}) \geq 0, \quad \forall f \in \Omega.
\end{equation}
This follows because for any feasible $f$, we have $\sum_k f_{k,m}^{rs} = \sum_k f_{k,m}^{rs,*} = d_m^{rs}$, so:
\begin{eqnarray}
    \sum_k C_{k,m}^{rs}(f^*) \cdot f_{k,m}^{rs} - \sum_k C_{k,m}^{rs}(f^*) \cdot f_{k,m}^{rs,*} 
    &\geq& \pi_m^{rs,*} \sum_k f_{k,m}^{rs} - \sum_k C_{k,m}^{rs}(f^*) \cdot f_{k,m}^{rs,*} \nonumber \\
    &=& \pi_m^{rs,*} \cdot d_m^{rs} - \pi_m^{rs,*} \cdot d_m^{rs} \\
    &=& 0,
\end{eqnarray}
where the inequality uses $C_{k,m}^{rs}(f^*) \geq \pi_m^{rs,*}$ for all $k$, with equality when $f_{k,m}^{rs,*} > 0$.
\end{proof}

\subsubsection{Existence, Uniqueness, and Monotonicity}

We establish the well-posedness of the fixed-class MUE problem through careful analysis of the cost operator's properties.

\begin{lemma}[Monotonicity of the Cost Operator]\label{lem:monotone}
Under the BPR cost function \eqref{eq:BPR} with $\beta \geq 1$, the path cost operator $F: \Omega \to \mathbb{R}^n$ is monotone:
\begin{equation}
    \langle F(f) - F(f'), f - f' \rangle \geq 0, \quad \forall f, f' \in \Omega.
\end{equation}
Moreover, if $\beta > 1$ and the network is connected, $F$ is strictly monotone on the subspace orthogonal to the kernel of demand constraints.
\end{lemma}

\begin{proof}
The Jacobian of $F$ with respect to $f$ can be decomposed as:
\begin{equation}
    \nabla_f F(f) = \begin{pmatrix} \nabla_{f_{gv}} F_{gv} & \nabla_{f_{ev}} F_{gv} \\ \nabla_{f_{gv}} F_{ev} & \nabla_{f_{ev}} F_{ev} \end{pmatrix}.
\end{equation}
For the $(k,m)$-th component corresponding to path $k$ of class $m$:
\begin{equation}
    \frac{\partial C_{k,m}^{rs}}{\partial f_{k',m'}^{r's'}} = \gamma_m \sum_{a \in \mathcal{A}} t_a'(x_a) \cdot \delta_{ka}^{rs} \cdot \delta_{k'a}^{r's'},
\end{equation}
where $t_a'(x_a) = \frac{\alpha \beta t_a^0}{c_a} \left( \frac{x_a}{c_a} \right)^{\beta-1} \geq 0$.

The Jacobian can be written as:
\begin{equation}
    \nabla_f F(f) = \Gamma \cdot \Delta^\top \cdot \text{diag}(t'(x)) \cdot \Delta,
\end{equation}
where $\Gamma = \text{diag}(\gamma_{gv} \mathbf{1}_{n_{gv}}, \gamma_{ev} \mathbf{1}_{n_{ev}})$ and $\Delta = (\Delta_{gv}, \Delta_{ev})$. Since $t_a'(x_a) \geq 0$ for all $a$, we have $\text{diag}(t'(x)) \geq 0$, implying:
\begin{equation}
    v^\top \nabla_f F(f) v = \sum_{a \in \mathcal{A}} t_a'(x_a) \left( \sum_m \gamma_m (\Delta_m v_m)_a \right)^2 \geq 0,
\end{equation}
establishing monotonicity. Strict monotonicity follows when $\beta > 1$ (ensuring $t_a'(x_a) > 0$ for $x_a > 0$) and when the network is connected (ensuring $\Delta$ has sufficient rank).
\end{proof}

\begin{theorem}[Existence and Uniqueness]\label{thm:existence}
The fixed-class MUE problem admits at least one solution. If the cost operator $F$ is strictly monotone, the equilibrium link flow $x^*$ is unique. If $F$ is strongly monotone, the path flow $f^*$ is also unique.
\end{theorem}

\begin{proof}
\textit{Existence:} The feasible region $\Omega = \Omega_{gv} \times \Omega_{ev}$ is nonempty (assuming feasible demands), compact (bounded by total demand), and convex. The cost operator $F$ is continuous. By the existence theorem for variational inequalities \citep{facchinei2003finite}, VI$(F, \Omega)$ has at least one solution.

\textit{Uniqueness of link flows:} Suppose $f^*$ and $f^{**}$ are two equilibria. By the VI characterization:
\begin{align}
    \langle F(f^*), f^{**} - f^* \rangle &\geq 0, \\
    \langle F(f^{**}), f^* - f^{**} \rangle &\geq 0.
\end{align}
Adding these inequalities:
\begin{equation}
    \langle F(f^*) - F(f^{**}), f^* - f^{**} \rangle \leq 0.
\end{equation}
By strict monotonicity, this implies $\Delta(f^* - f^{**}) = 0$, i.e., $x^* = x^{**}$.

\textit{Uniqueness of path flows:} Under strong monotonicity, $\langle F(f^*) - F(f^{**}), f^* - f^{**} \rangle \geq \mu \|f^* - f^{**}\|^2$ for some $\mu > 0$, which combined with the above yields $f^* = f^{**}$.
\end{proof}

\subsubsection{Lagrangian Formulation and Saddle Point Characterization}

To develop efficient algorithms and incorporate system-level constraints, we reformulate the fixed-class MUE within a Lagrangian framework. Consider additional constraints of the form:
\begin{equation}\label{eq:capacity_constraint}
    Af \leq c,
\end{equation}
representing link capacity limits, environmental regulations, or policy constraints. Here $A \in \mathbb{R}^{p \times n}$ and $c \in \mathbb{R}^p$.

\begin{definition}[Augmented Lagrangian]
The augmented Lagrangian for the constrained fixed-class MUE is:
\begin{equation}\label{eq:lagrangian}
    \mathcal{L}(f, \lambda, \mu) = \Phi(f) + \lambda^\top (Af - c) + \mu_{gv}^\top (B_{gv} f_{gv} - d_{gv}) + \mu_{ev}^\top (B_{ev} f_{ev} - d_{ev}),
\end{equation}
where $\Phi(f) = \sum_{a \in \mathcal{A}} \int_0^{x_a(f)} \bar{t}_a(\omega) d\omega$ with $\bar{t}_a(\omega) = \bar{\gamma} \cdot t_a(\omega) + \bar{C} \cdot l_a$ being an averaged cost function, $\lambda \geq 0$ are dual variables for capacity constraints, and $\mu_m$ are dual variables for demand constraints.
\end{definition}

\begin{theorem}[Saddle Point Equivalence]\label{thm:saddle}
The point $(f^*, \lambda^*, \mu^*)$ is a saddle point of the Lagrangian $\mathcal{L}$:
\begin{equation}\label{eq:saddle_point}
    \mathcal{L}(f^*, \lambda, \mu) \leq \mathcal{L}(f^*, \lambda^*, \mu^*) \leq \mathcal{L}(f, \lambda^*, \mu^*), \quad \forall f \geq 0, \lambda \geq 0, \mu,
\end{equation}
if and only if $f^*$ solves the constrained fixed-class MUE and $(\lambda^*, \mu^*)$ are the associated Lagrange multipliers.
\end{theorem}

\begin{proof}
The saddle point conditions yield the KKT system:
\begin{equation}\label{eq:KKT_full}
    \boxed{
    \begin{aligned}
        & \text{(Stationarity)} && F(f^*) + A^\top \lambda^* + \tilde{B}^\top \mu^* - \nu^* = 0, \\
        & \text{(Primal feasibility)} && B_m f_m^* = d_m, \ f^* \geq 0, \ Af^* \leq c, \\
        & \text{(Dual feasibility)} && \lambda^* \geq 0, \ \nu^* \geq 0, \\
        & \text{(Complementarity)} && \lambda^* \perp (c - Af^*), \ \nu^* \perp f^*,
    \end{aligned}
    }
\end{equation}
where $\tilde{B} = \text{diag}(B_{gv}, B_{ev})$ and $\nu^*$ enforces non-negativity. The stationarity condition states that at equilibrium, the effective path cost (including shadow prices) equals the OD potential for active paths—precisely the generalized Wardrop condition.
\end{proof}

The saddle point formulation admits a min-max interpretation:
\begin{equation}\label{eq:minmax}
    \min_{f \in \Omega, f \geq 0} \max_{\lambda \geq 0} \mathcal{L}(f, \lambda) = \max_{\lambda \geq 0} \min_{f \in \Omega, f \geq 0} \mathcal{L}(f, \lambda),
\end{equation}
where strong duality holds under Slater's condition (existence of strictly feasible $f$).

\subsubsection{Primal-Dual Algorithm with Convergence Guarantees}

We now present our main algorithmic contribution: a primal-dual method for solving the fixed-class MUE with provable convergence guarantees.

\begin{algorithm}[ht]
\caption{Primal-Dual Gradient Method for Fixed-Class MUE}
\label{alg:primal_dual}
\begin{algorithmic}[1]
\REQUIRE Step sizes $\alpha, \beta > 0$; initial $(f^0 \in \Omega, \lambda^0 \geq 0)$; tolerance $\epsilon > 0$
\ENSURE Approximate equilibrium $(f^*, \lambda^*)$
\STATE Initialize $k \leftarrow 0$
\REPEAT
    \STATE \textbf{// Primal Update (Class-wise Projection)}
    \FOR{each class $m \in \{gv, ev\}$}
        \STATE Compute effective costs: $\tilde{C}_m^k = F_m(f^k) + A_m^\top \lambda^k$
        \STATE Update: $\hat{f}_m^{k+1} = f_m^k - \alpha \tilde{C}_m^k$
        \STATE Project: $f_m^{k+1} = \Pi_{\Omega_m}(\hat{f}_m^{k+1})$
    \ENDFOR
    \STATE \textbf{// Dual Update (Constraint Enforcement)}
    \STATE Compute residual: $r^{k+1} = Af^{k+1} - c$
    \STATE Update: $\lambda^{k+1} = \Pi_{\mathbb{R}^p_+}(\lambda^k + \beta r^{k+1})$
    \STATE \textbf{// Convergence Check}
    \STATE Compute gap: $G^{k+1} = \|f^{k+1} - f^k\|^2 + \|\lambda^{k+1} - \lambda^k\|^2$
    \STATE $k \leftarrow k + 1$
\UNTIL{$G^k < \epsilon$}
\RETURN $(f^k, \lambda^k)$
\end{algorithmic}
\end{algorithm}

The key feature of Algorithm \ref{alg:primal_dual} is the \textit{class-wise projection} in the primal update. Unlike standard projected gradient methods, each class's flow is projected onto its own feasible set $\Omega_m$ independently, respecting the fixed-class structure.

\begin{lemma}[Projection onto Simplex]\label{lem:projection}
For each class $m$ and OD pair $(r,s)$, the projection $\Pi_{\Omega_m}$ onto the probability simplex scaled by demand $d_m^{rs}$ can be computed in $O(|\psi_{rs}| \log |\psi_{rs}|)$ time via sorting.
\end{lemma}

\begin{theorem}[Global Convergence]\label{thm:convergence}
Let $(f^k, \lambda^k)$ be the sequence generated by Algorithm \ref{alg:primal_dual}. Suppose:
\begin{enumerate}
    \item[(A1)] The cost operator $F$ is monotone and $L_F$-Lipschitz continuous;
    \item[(A2)] The constraint matrix $A$ satisfies $\|A\| \leq L_A$;
    \item[(A3)] The step sizes satisfy $0 < \alpha < \frac{1}{L_F + L_A^2/\beta}$ and $\beta > 0$.
\end{enumerate}
Then the sequence $(f^k, \lambda^k)$ converges to a saddle point $(f^*, \lambda^*)$ of the Lagrangian $\mathcal{L}$.
\end{theorem}

\begin{proof}
We establish convergence through a Lyapunov function argument.

\textit{Step 1: Define the Lyapunov function.} Let $(f^*, \lambda^*)$ be any saddle point. Define:
\begin{equation}
    V^k = \frac{1}{2\alpha}\|f^k - f^*\|^2 + \frac{1}{2\beta}\|\lambda^k - \lambda^*\|^2.
\end{equation}

\textit{Step 2: Analyze the primal update.} By the projection property:
\begin{equation}
    \langle f^{k+1} - f^*, \hat{f}^{k+1} - f^{k+1} \rangle \geq 0.
\end{equation}
Substituting $\hat{f}^{k+1} = f^k - \alpha(F(f^k) + A^\top \lambda^k)$:
\begin{align}
    \|f^{k+1} - f^*\|^2 &\leq \|f^k - f^* - \alpha(F(f^k) + A^\top \lambda^k)\|^2 \nonumber \\
    &= \|f^k - f^*\|^2 - 2\alpha \langle F(f^k) + A^\top \lambda^k, f^k - f^* \rangle \nonumber \\
    &\quad + \alpha^2 \|F(f^k) + A^\top \lambda^k\|^2. \label{eq:primal_bound}
\end{align}

\textit{Step 3: Analyze the dual update.} Similarly:
\begin{align}
    \|\lambda^{k+1} - \lambda^*\|^2 &\leq \|\lambda^k - \lambda^*\|^2 + 2\beta \langle Af^{k+1} - c, \lambda^k - \lambda^* \rangle \nonumber \\
    &\quad + \beta^2 \|Af^{k+1} - c\|^2. \label{eq:dual_bound}
\end{align}

\textit{Step 4: Combine using saddle point properties.} From the saddle point characterization:
\begin{align}
    \langle F(f^*), f^k - f^* \rangle &\geq 0, \label{eq:sp1} \\
    \langle Af^* - c, \lambda^k - \lambda^* \rangle &\leq 0. \label{eq:sp2}
\end{align}
By monotonicity of $F$:
\begin{equation}
    \langle F(f^k) - F(f^*), f^k - f^* \rangle \geq 0. \label{eq:mono}
\end{equation}

\textit{Step 5: Establish descent.} Combining \eqref{eq:primal_bound}–\eqref{eq:mono} and using Young's inequality:
\begin{align}
    V^{k+1} - V^k &\leq -\langle F(f^k), f^k - f^* \rangle - \langle A^\top \lambda^k, f^k - f^* \rangle \nonumber \\
    &\quad + \langle Af^{k+1} - c, \lambda^k - \lambda^* \rangle + \frac{\alpha}{2}\|F(f^k) + A^\top \lambda^k\|^2 \nonumber \\
    &\quad + \frac{\beta}{2}\|Af^{k+1} - c\|^2.
\end{align}

Using the saddle point inequalities and the identity:
\begin{equation}
    \langle Af^{k+1} - c, \lambda^k - \lambda^* \rangle - \langle A^\top \lambda^k, f^k - f^* \rangle = \langle A(f^{k+1} - f^k), \lambda^k \rangle + \text{(boundary terms)},
\end{equation}
we obtain after algebraic manipulation:
\begin{equation}
    V^{k+1} - V^k \leq -\delta \left( \|f^{k+1} - f^k\|^2 + \|\lambda^{k+1} - \lambda^k\|^2 \right),
\end{equation}
for some $\delta > 0$ depending on $\alpha, \beta, L_F, L_A$ under condition (A3).

\textit{Step 6: Conclude convergence.} The sequence $\{V^k\}$ is non-increasing and bounded below by 0, hence convergent. The descent inequality implies:
\begin{equation}
    \sum_{k=0}^\infty \left( \|f^{k+1} - f^k\|^2 + \|\lambda^{k+1} - \lambda^k\|^2 \right) < \infty,
\end{equation}
so $(f^k, \lambda^k)$ is a Cauchy sequence converging to some $(f^\infty, \lambda^\infty)$. By continuity of the KKT conditions, $(f^\infty, \lambda^\infty)$ is a saddle point.
\end{proof}

\begin{theorem}[Convergence Rate]\label{thm:rate}
Under the conditions of Theorem \ref{thm:convergence}, the algorithm achieves:
\begin{enumerate}
    \item[(i)] \textbf{Sublinear rate:} $\min_{0 \leq k \leq K} G^k = O(1/K)$;
    \item[(ii)] \textbf{Linear rate (under strong monotonicity):} If $F$ is $\mu$-strongly monotone, then $V^k \leq (1-\rho)^k V^0$ for some $\rho \in (0,1)$.
\end{enumerate}
\end{theorem}

\begin{proof}
\textit{(i)} Summing the descent inequality over $k = 0, \ldots, K-1$:
\begin{equation}
    \delta \sum_{k=0}^{K-1} G^k \leq V^0 - V^K \leq V^0.
\end{equation}
Hence $\min_{k \leq K-1} G^k \leq V^0 / (\delta K) = O(1/K)$.

\textit{(ii)} Under $\mu$-strong monotonicity, the descent inequality strengthens to:
\begin{equation}
    V^{k+1} \leq (1 - 2\alpha\mu) V^k,
\end{equation}
yielding geometric convergence with $\rho = 2\alpha\mu$.
\end{proof}

\subsubsection{Extra-Gradient Variant for Enhanced Stability}

For networks with highly heterogeneous cost structures between GVs and EVs, we propose an extra-gradient variant that provides improved stability.

\begin{algorithm}[ht]
\caption{Extra-Gradient Method for Fixed-Class MUE}
\label{alg:extra_gradient}
\begin{algorithmic}[1]
\REQUIRE Step size $\tau > 0$; initial $(f^0 \in \Omega, \lambda^0 \geq 0)$
\ENSURE Approximate equilibrium $(f^*, \lambda^*)$
\REPEAT
    \STATE \textbf{// Extrapolation Step}
    \STATE $\tilde{f}^k = \Pi_\Omega(f^k - \tau[F(f^k) + A^\top \lambda^k])$
    \STATE $\tilde{\lambda}^k = \Pi_{\mathbb{R}^p_+}(\lambda^k + \tau(Af^k - c))$
    \STATE \textbf{// Main Update Step}
    \STATE $f^{k+1} = \Pi_\Omega(f^k - \tau[F(\tilde{f}^k) + A^\top \tilde{\lambda}^k])$
    \STATE $\lambda^{k+1} = \Pi_{\mathbb{R}^p_+}(\lambda^k + \tau(A\tilde{f}^k - c))$
    \STATE $k \leftarrow k + 1$
\UNTIL{convergence}
\end{algorithmic}
\end{algorithm}

\begin{theorem}[Convergence of Extra-Gradient Method]\label{thm:eg_convergence}
Let $T: (f, \lambda) \mapsto (F(f) + A^\top \lambda, c - Af)$ be the saddle point operator. If $T$ is monotone and $L$-Lipschitz, then Algorithm \ref{alg:extra_gradient} with step size $\tau \in (0, 1/L)$ converges to a saddle point $(f^*, \lambda^*)$.
\end{theorem}

\begin{proof}
The extra-gradient method can be viewed as a forward-backward splitting applied to the monotone inclusion $0 \in T(z) + N_{\Omega \times \mathbb{R}^p_+}(z)$ where $z = (f, \lambda)$. Following \citet{korpelevich1976extragradient}, define:
\begin{equation}
    \|z^{k+1} - z^*\|^2 \leq \|z^k - z^*\|^2 - (1 - \tau^2 L^2)\|z^k - \tilde{z}^k\|^2.
\end{equation}
For $\tau < 1/L$, the coefficient $(1 - \tau^2 L^2) > 0$, ensuring monotonic decrease of distance to the solution set. The sequence is Fejér monotone and hence converges.
\end{proof}

\subsubsection{Stability Analysis under Penetration Shift}

We now analyze how the equilibrium evolves as the EV penetration rate changes, providing theoretical foundations for the empirically observed ``plateau phases.''

Let $R_{e}\in [0,1]$ denote the EV penetration rate. The class-specific demands are:
\begin{equation}
    d_{ev}^{rs}(R_{e}) = R_{e} \cdot d^{rs}, \quad d_{gv}^{rs}(R_{e}) = (1-R_{e}) \cdot d^{rs},
\end{equation}
where $d^{rs}$ is the total OD demand.

\begin{definition}[Active Path Set]
For equilibrium $f^*(R_{e})$, the active path set for class $m$ between $(r,s)$ is:
\begin{equation}
    \mathcal{A}_m^{rs}(R_{e}) := \left\{ k \in \psi_{rs} : f_{k,m}^{rs,*}(R_{e}) > 0 \right\}.
\end{equation}
\end{definition}

\begin{theorem}[Equilibrium Sensitivity to Penetration Rate]\label{thm:plateau}
Suppose $F$ is strictly monotone and continuously differentiable. Then:
\begin{enumerate}
    \item[(i)] \textbf{(Local Lipschitz Continuity)} The equilibrium link flow $x^*(R_{e})$ is Lipschitz continuous in $R_{e}$:
    \begin{equation}
        \|x^*(R_e) - x^*(R'_e)\| \leq L \cdot |R_e - R'_e|,
    \end{equation}
    where $L > 0$ depends on network structure and cost parameters.

    \item[(ii)] \textbf{(Reduced Sensitivity Condition)} Define the \textit{path overlap ratio} as:
    \begin{equation}
        \rho(R_e) = \frac{|\mathcal{A}_{gv}(R_e) \cap \mathcal{A}_{ev}(R_e)|}{|\mathcal{A}_{gv}(R_e) \cup \mathcal{A}_{ev}(R_e)|},
    \end{equation}
    where $\mathcal{A}_m(R_e) = \bigcup_{(r,s)} \mathcal{A}_m^{rs}(R_e)$. When $\rho(R_e) \to 1$ (high path overlap), the sensitivity $\|\partial x^*/\partial R_e\|$ diminishes.

    \item[(iii)] \textbf{(Transition Points)} Equilibrium pattern transitions occur when some path switches between active and inactive status, i.e., when $f_{k,m}^{rs,*}(R_{e}) \to 0$ for some $(k, m, r, s)$.
\end{enumerate}
\end{theorem}

\begin{proof}
\textit{Step 1:} Define the equilibrium mapping $f^*: [0,1] \to \Omega$ implicitly by the KKT conditions $G(f^*(R_e), R_e) = 0$.

\textit{Step 2:} By the Implicit Function Theorem, $f^*(R_e)$ is differentiable where the Jacobian $\partial G/\partial f$ is non-singular, which is guaranteed by strict monotonicity of $F$.

\textit{Step 3:} The Jacobian becomes singular precisely when some complementarity condition $f_{k,m}^{rs,*} \cdot (C_{k,m}^{rs} - \pi_m^{rs,*}) = 0$ transitions between binding regimes, corresponding to paths switching between active and inactive.

\textit{Step 4:} Between such transition points, $f^*(R_e)$ is smooth, and we can bound $\|\partial x^*/\partial R_e\|$ using:
\begin{equation}
    \frac{\partial x^*}{\partial R_e} = -\left(\frac{\partial G}{\partial f}\right)^{-1} \frac{\partial G}{\partial R_e}.
\end{equation}
The Lipschitz constant $L$ is bounded by $\|(\partial G/\partial f)^{-1}\| \cdot \|\partial G/\partial R_e\|$.

\textit{Step 5:} When $\mathcal{A}_{gv} \approx \mathcal{A}_{ev}$ (high overlap, $\rho \to 1$), the term $\partial G/\partial R_e$ becomes small because demand redistribution between classes has minimal effect on aggregate flows. Specifically, if both classes use similar paths, shifting demand from GV to EV merely relabels flows without substantially changing link loads, reducing $\|\partial x^*/\partial R_e\|$.
\end{proof}
% === END MODIFIED THEOREM ===

\begin{corollary}[Critical Penetration Thresholds]\label{cor:threshold}
Suppose at $R_{e} = 0$ (pure GV), the active path sets are $\mathcal{P}_{gv}^0$ for each OD pair. Define the critical thresholds:
\begin{equation}
    R_{e,c}^{(i)} := \inf\left\{ R_{e} > R_{e,c}^{(i-1)} : \exists k \in \mathcal{P}_{gv}^{i-1}, \ f_{k,gv}^{rs,*}(R_{e}) = 0 \right\},
\end{equation}
with $R_{e,c}^{(0)} = 0$. Then the system exhibits plateau phases in $[0, R_{e,c}^{(1)}], [R_{e,c}^{(1)}, R_{e,c}^{(2)}], \ldots$, with qualitative equilibrium changes occurring only at the boundaries.
\end{corollary}

\begin{remark}[Interpretation of Plateau Phases and Cascading Effects]
This theorem provides the theoretical underpinning for the transition zone (Definition \ref{def:transition_zone}) observed in numerical experiments. Crucially, plateau phases are typically brief because EVs displace GVs route-by-route: even minimal EV introduction can fully capture certain cost-advantaged routes, initiating incremental displacement immediately. At $R_e = 0$, the marginal system benefit is zero. As $R_e$ increases past a critical threshold, a cascading effect emerges---each additional EV simultaneously displaces GV flows across multiple congested links, triggering network-wide flow redistribution and rapidly amplifying system benefits. At medium to high penetration ($R_{e} \gtrsim 0.5$), EVs dominate and marginal returns diminish as fewer GV routes remain available for displacement. 
\end{remark}

\begin{remark}[Policy Implications]
Theorem \ref{thm:plateau} suggests that EV adoption subsidies may exhibit threshold effects: uniform incentives yield limited congestion benefits until penetration crosses critical thresholds, at which point traffic redistribution delivers substantial gains. Targeted policies to push penetration past specific thresholds may be more cost-effective than gradual uniform incentives.
\end{remark}

\subsection{Evaluation Metrics of Congestion Patterns}

To systematically evaluate the characteristics of congestion patterns exhibited by road systems following MUE assignment in the context of widespread EV penetration, we design a series of multi-dimensional evaluation metrics, as detailed below.

\subsubsection{Average Travel Time}

The average travel time of a system is the most fundamental and intuitive metric for measuring the overall operational efficiency of road networks. It directly reflects the time cost incurred by commuters to complete their trips under the prevailing traffic conditions. Specifically, under MUE conditions, the travel time between different OD pairs can be generated by summing the link travel time determined along the shortest path for the corresponding assigned traffic flow, as described in Eq. \eqref{7}. Then, under free flow conditions, the travel time between OD pairs is the travel time associated with the shortest path, disregarding congestion on the route segments. Therefore, the average travel time (in minutes) for all OD pairs under MUE conditions and free flow conditions can be expressed as follows:

\begin{equation}\label{57}
 T_{MUE}\left ( R_{e} \right ) =\frac{ {\textstyle \sum_{r}^{N}} {\textstyle \sum_{s}^{N}}c_{MUE} ^{rs}\left (  R_{e}\right )  q^{rs}  }{{\textstyle \sum_{r}^{N}} {\textstyle \sum_{s}^{N}}q^{rs} }  ,  
\end{equation}
\begin{equation}\label{58}
T_{FF} =\frac{ {\textstyle \sum_{r}^{N}} {\textstyle \sum_{s}^{N}}c_{FF} ^{rs}q^{rs}  }{{\textstyle \sum_{r}^{N}} {\textstyle \sum_{s}^{N}}q^{rs} }   ,  
\end{equation}
where $T_{MUE}\left ( R_{e} \right )$ denotes the average travel time of the system under MUE conditions when EV penetration is $R_{e}$ ; $T_{FF}$ denotes the average travel time of the system under free flow conditions; $c_{MUE} ^{rs}\left ( R_{e}\right )$ denotes the travel time between origin $r$ and destination $s$ under MUE conditions when EV penetration is $R_{e}$; $c_{FF} ^{rs}$ denotes the travel time between origin $r$ and destination $s$ under free flow conditions. 

Additionally, by comparing simulation results under EV penetration of 0 (baseline scenario) and of 1 (full penetration scenario), we devise metrics for the absolute change in travel time and relative change to quantify the impact of EV penetration on traffic congestion. The absolute change directly reflects the actual magnitude of efficiency alterations, while the relative change reveals the proportional relationship between traffic improvement/deterioration levels. Combining these two metrics enables effective identification of differentiated response patterns to EV penetration across cities with distinct congestion characteristics. The absolute change $\Delta T_{abs}$ and relative change $\Delta T_{rel}\left ( \% \right )$ can be expressed as:

\begin{equation}\label{59}
 \Delta T_{abs} =T_{MUE}\left ( 1 \right )- T_{MUE}\left ( 0 \right ),
\end{equation}
\begin{equation}\label{60}
 \Delta T_{rel}\left ( \% \right )  =\frac{T_{MUE}\left ( 1 \right )- T_{MUE}\left ( 0 \right )}{T_{MUE}\left ( 0 \right )}\times 100\% ,
\end{equation}
where the relative change is more suitable for horizontal comparisons between cities of different scales.

\subsubsection{Potential Savings}

To investigate the travel time saving potential that EVs can deliver at different stages of penetration, we develop the metric of potential savings. This metric aims to quantify the proportion of travel time savings already achieved by the system relative to its maximum theoretical savings potential under a given level of EV penetration. The specific calculation can be defined as:
\begin{equation}\label{PS}
PS\left (R_{e}\right ) =\frac{ T_{max}-T\left ( R_{e}\right )}{T_{max}-T_{min}}\times 100\%  ,
\end{equation}
where $PS\left (R_{e}\right )$ denotes the potential savings (\%) at the EV penetration of $R_{e}$; $T_{max}$ and $T_{min}$ denote the maximum and minimum average travel times of the system, respectively; $T\left ( R_{e} \right )$ denotes the average travel time of the system at the EV penetration of $R_{e}$. The metric of potential savings provides a standardized measure that not only tracks the progress of improvements, but more importantly, displays bottlenecks or plateau phases in the effectiveness of those improvements. We then propose the metric of the difference in potential savings to further characterize the patterns and unevenness of benefit shifts during EV penetration. It is defined as the increase or decrease in the potential savings ($PS$) value within the EV penetration variation range, specifically as follows:

\begin{equation}\label{62}
\triangle PS\left ( R_{e_{i} }  \to R_{e_{j} }    \right ) = PS\left ( R_{e_{i} }  \right )-PS\left ( R_{e_{j} }  \right ) ,
\end{equation}
where $R_{e_{i} } $ and $R_{e_{j} }  $ represent two consecutive levels of EV presentation (e.g., from 0 to 0.05); $\triangle PS\left ( R_{e_{i} }  \to R_{e_{j}} \right )$ denotes the difference in potential savings realized within this range. This metric can observe phased increases/decreases in traffic conditions, thereby showing the key aspects of the impact of EV penetration.

\subsubsection{Volume over Capacity (VOC)}

To evaluate the congestion characteristics of road systems at the micro level of network load, we introduce the key metric of VOC. VOC is defined as the ratio of traffic volume over capacity for each link, which is measured as: 
\begin{equation}\label{63}
 VOC_{a} =\frac{v_{a} }{c_{a} } ,
\end{equation}
where $v_{a}$ and $c_{a} $ are the traffic flow and capacity of road $a$. This metric directly reflects the traffic load and saturation degree of the link during the time period. The $VOC_{a}$ value of less than 1 indicates that the traffic flow on link $a$ has not yet reached its theoretical capacity, leaving sufficient room to accommodate more vehicles. Whereas, the $VOC_{a}$ value greater than 1 indicates that link $a$ is already in a state of oversaturated congestion.

To capture the total congestion load across the entire road network at the system macro level, we further define the total VOC metric, denoted as $VOC_{total}$. The calculation method is the sum of the VOC values of all links in the road network, i.e.,

\begin{equation}\label{VOC_{total}}
VOC_{total} = \sum_{a}^{} VOC_{a}.
\end{equation}

$VOC_{total}$ quantifies the relative relationship between the total traffic demand carried by the entire transportation network and its total theoretical capacity. By observing the trend of this metric during the EV penetration increase, it can effectively demonstrate whether traffic flow has undergone an effective spatial redistribution.

\subsubsection{Road Utilization Rate}

To evaluate the breadth and balance of road network capacity utilization, we introduce the road utilization rate metric. It can be defined as the ratio of the number of road segments (links) in active use (i.e., with the VOC value greater than 0) to the total number of road segments under specific traffic conditions. The formal definition of this metric is as follows:

\begin{equation}\label{65}
RUR=\frac{N_{voc> 0} }{N_{total}},
\end{equation}
where $RUR$ denotes road utilization rate; $N_{voc> 0}$ denotes the number of road segments in active use; $N_{total}$ denotes the total number of all road segments in the network. This metric can show the effective utilization of road network resources in terms of spatial breadth. It clearly reveals whether traffic flows are concentrated and congested on a few roads or evenly distributed across most available paths. 

\subsubsection{Link Congested Time}

Another key metric is link congested time. This refers to the actual travel time a vehicle spends passing through a specific link (road) under congested conditions, which can be obtained using the BPR function in formula Eq. \eqref{eq:BPR}. This metric provides a clear indication of the link's operational performance under traffic load and its congestion status.

\subsubsection{Difference in Delay Factor}

To explore the micro-level impact of EV penetration on traffic congestion patterns from the dimensions of time reliability and fairness, we propose the difference in delay factor as a metric. This metric aims to visualize the relative change in congestion levels on specific links or roads before and after the introduction of EVs. We first define the delay factor of link $a$ as the ratio of its actual travel time $t_{a}$ under congested conditions to its baseline travel time $t_{a}^{0}$ under free flow conditions, which can be expressed as:

\begin{equation}\label{DF_{a}}
DF_{a} =\frac{t_{a} }{t_{a}^{0} }.
\end{equation}

This metric directly reflects the impact of traffic congestion on travel time deterioration, with values closer to 1 indicating that the link's operational state is approaching ideal free flow status. Based on this metric, the difference in delay factor can be further derived. It represents the difference in delay factor for the same link between two scenarios: full EV penetration  and a baseline scenario without EVs :

\begin{equation}\label{67}
\bigtriangleup DF_{a} =DF_{a}^{1} -DF_{a}^{0},
\end{equation}
where $\bigtriangleup DF_{a}$ represents the difference in delay factor for link $a$; $DF_{a}^{1}$ and $DF_{a}^{1}$ represent the delay factor for link a when EV penetration is 1 and 0, respectively. This metric attributes macroscopic changes in system efficiency to microscopic roads, which can expose how EV penetration influences road network operations by reshaping the spatial distribution of traffic flow.

\section{Case Study}\label{Case_Study}

In this section, using 10 representative U.S. cities as examples, the mixed traffic flow of GVs and EVs is assigned to the actual road network based on peak-hour OD (travel demand) data. This research explores the immediate impact of EV penetration on traffic conditions and the evolutionary patterns of traffic congestion.

\subsection{Study Area}
To validate the effectiveness and universality of the method, we consider factors such as geographic location, urban scale, topographical characteristics, and traffic conditions during commuting hours to select 10 representative cities in the United States. Specifically, these cities include San Francisco, Portland, Las Vegas, New Orleans, Dallas, Milwaukee, Boston, Philadelphia, Denver, and Honolulu, whose geographical locations are shown in Fig. \ref{fig:3}. These include several large cities such as Dallas, Philadelphia, and San Francisco, mid-sized cities like Milwaukee and Las Vegas, as well as smaller cities like Honolulu. The topographical characteristics of these cities also differ. For example, San Francisco is divided by the bay terrain, with traffic relying on key bridges and tunnels to maintain connectivity along commuter corridors. While Dallas and Las Vegas feature relatively flat and continuous terrain, forming grid-like transportation systems dominated by surface roads. Basic information for the selected 10 representative U.S. cities is shown in Table \ref{tab:city_stats}, with population and land area data from the year 2020\footnote{\url{https://www.census.gov/}} and congestion ranking from the year 2022\footnote{\url{https://www.tomtom.com/traffic-index/ranking/}}. It can be observed that among these 10 cities, Philadelphia and Dallas have the largest populations, Dallas and New Orleans have the largest land areas, San Francisco and Boston have the highest population density, while the congestion rankings of San Francisco, Boston, and Philadelphia are the highest.

\begin{figure}[ht]
    \centering
    \includegraphics[width=\textwidth]{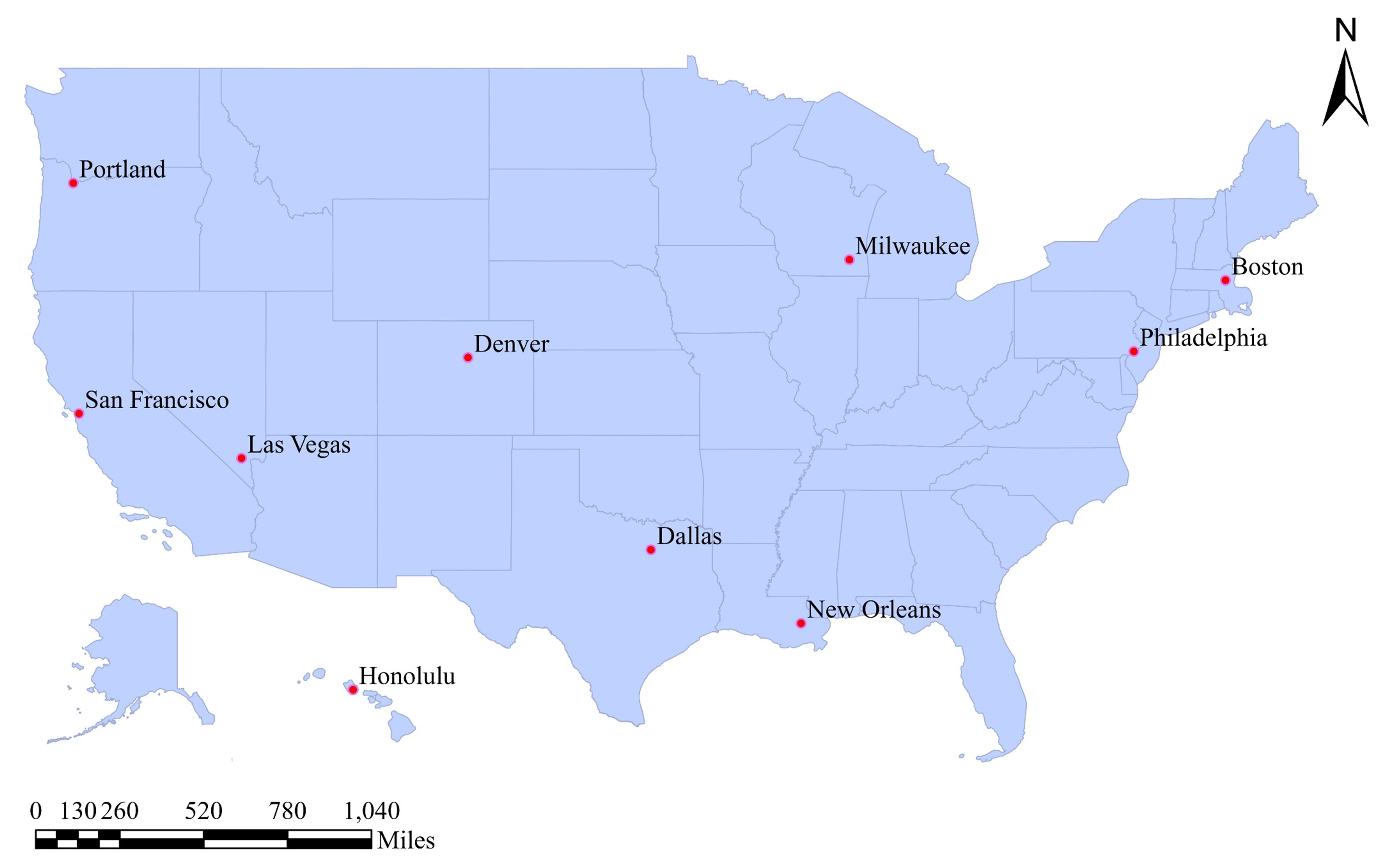}
    \caption{The geographic locations of the selected 10 representative U.S. cities}
    \label{fig:3}
\end{figure}

\begin{table}[!ht]
\centering
\caption{Basic information for the selected 10 representative U.S. cities.}
\fontsize{9}{14}\selectfont 
\begin{tabularx}{\textwidth}{XXXXXXp{2.5cm}p{2.5cm}} 
\toprule
\textbf{No.} & \textbf{City} & \textbf{State} & \textbf{Census 2020} & \textbf{Land Area ($km^{2}$)} & \textbf{Population Density} & \textbf{Congestion Ranking} \\
\midrule
1 & San Francisco & California & 873,965 & 121.5 & 7,195 & 3 \\
2 & Portland & Oregon & 652,503 & 345.8 & 1,887 & 16 \\
3 & Las Vegas & Nevada & 641,903 & 367.3 & 1,748 & 25 \\
4 & New Orleans & Louisiana & 383,997 & 439.0 & 875 & 19 \\
5 & Dallas & Texas & 1,304,379 & 879.6 & 1,483 & 33 \\
6 & Milwaukee & Wisconsin & 577,222 & 249.2 & 2,300 & 28 \\
7 & Boston & Massachusetts & 675,647 & 125.1 & 5,401 & 4 \\
8 & Philadelphia & Pennsylvania & 1,603,797 & 348.1 & 4,607 & 8 \\
9 & Denver & Colorado & 715,522 & 396.5 & 1,805 & 14 \\
10 & Honolulu & Hawaii & 350,964 & 156.7 & 2,240 & 13 \\
\bottomrule
\end{tabularx}

\label{tab:city_stats}
\end{table}

\subsection{Data}

\subsubsection{Data Acquisition}
In this study, we obtain the road network and travel demand data for these 10 cities. The urban road network data are derived from the OpenStreetMap (OSM) database\footnote{\url{https://www.openstreetmap.org/}}. This collaborative open-source mapping platform provides comprehensive geospatial information covering road infrastructure, including network topology structure, road attributes, and connectivity relationships. After data cleansing and integration, the processed network data can serve as the primary input for the traffic assignment model. We extract the nodes and links from the road networks of 10 cities in the OSM data, as shown in Fig. \ref{fig:4}(a). Each node represents the intersection of two links and contains a unique identifier along with the node's latitude and longitude information. By leveraging the correspondence between nodes and links, the network topology can be effectively constructed, and relevant road attributes can be extracted. System comparisons reveal that Dallas possesses the most complex road network (number of nodes: 21,389; number of links: 77,818), followed by Philadelphia (number of nodes: 10,410; number of links: 38,641), while Honolulu has the smallest network scale (number of nodes: 2,982; number of links: 1,120).

Additionally, travel demand data are estimated based on the LODES dataset\footnote{\url{https://lehd.ces.census.gov/data/lodes/}} provided by the U.S. Census Bureau. The LODES dataset has been widely used in related research \citep{mckinney2021total}, encompassing commuting data for the workforce across all U.S. states over multiple years. The collection process for LODES data involves employers reporting employee information to state workforce agencies, including work location and residence location. The U.S. Census Bureau collaborates with state agencies to process and anonymize these data, ultimately generating OD pairs. This dataset provides census block codes for workplaces and residences at the finest block level, accompanied by corresponding total employment figures. In practice, the LODES dataset offers an excellent representation of the commuting distribution of the U.S. workforce, suitable for constructing OD matrices. Therefore, this study collects commuting OD demand data at the block level for 10 U.S. cities in 2019. The specific total demand data for 10 U.S. cities are presented in Fig. \ref{fig:4}(b). The comparative analysis of aggregate travel demand demonstrates significant disparities across the 10 cities, with variations extending multiple orders of magnitude. In particular, Philadelphia records the highest trip volume at 399,290, followed closely by Dallas with 345,369 trips, while New Orleans (104,261 trips) and Las Vegas (132,511 trips) demonstrate the lowest and second-lowest demand, respectively.  

\begin{figure}[ht]
    \centering
    \includegraphics[width=\textwidth]{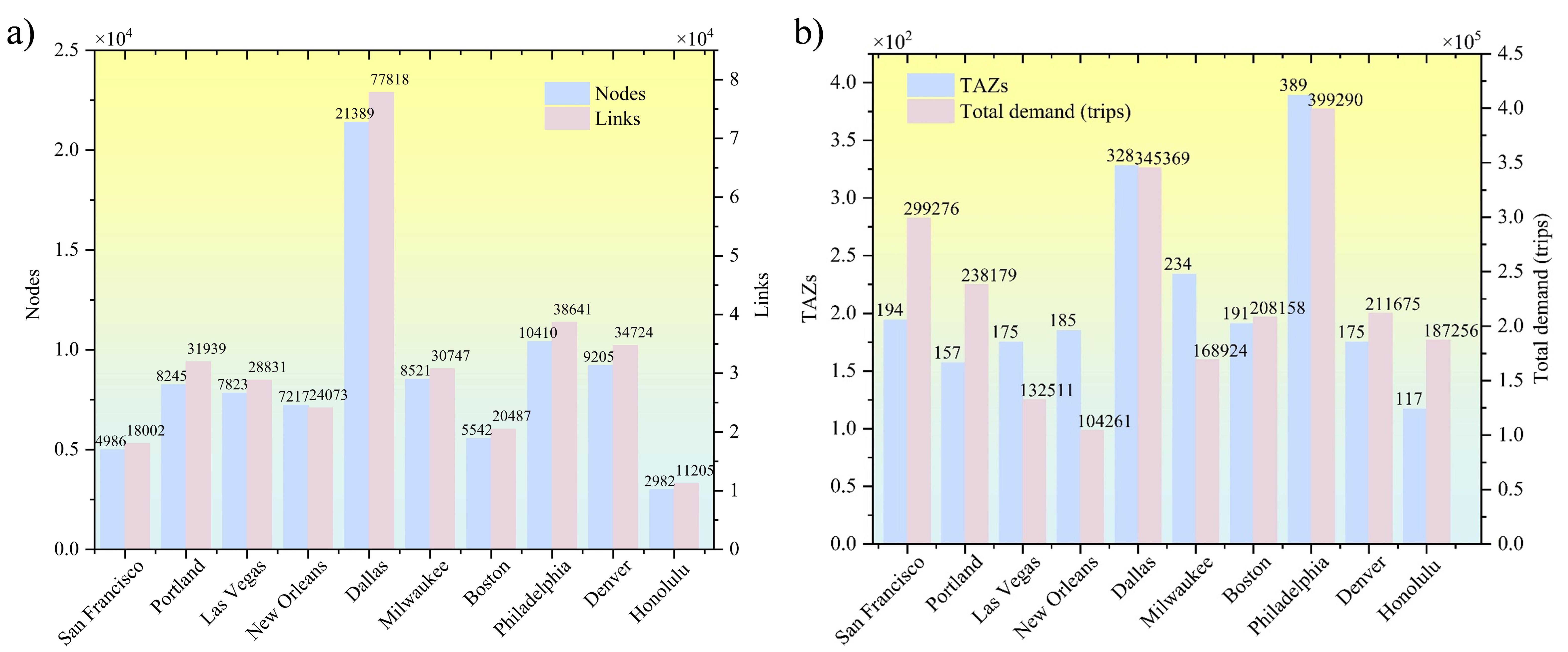}
    \caption{Road network and travel demand data for 10 U.S. cities. a) Number of nodes and links. b) Number of TAZs and total demand (trips).}
    \label{fig:4}
\end{figure}

\subsubsection{Data Integration}
This section integrates road network and travel demand data into a computable dataset. OD pairs in the travel demand matrix lack direct correspondence with nodes in the road network. Therefore, we then establish a geospatial association process to map each block to its nearest network node for subsequent traffic flow assignment. We first aggregate the OD travel demand data from the fundamental block-level granularity to the tract level. As defined by the United State Bureau \citep{smith2021census}, blocks represent the smallest statistical units, typically containing 600-3,000 residents, while tracts comprise multiple blocks with populations ranging from 1,200 to 8,000 individuals. This hierarchical aggregation achieves an optimal balance between computational tractability and spatial resolution accuracy, consistent with established practices in transportation research \citep{abdel2013geographical}. Consequently, we adopt tracts as Traffic Analysis Zones (TAZs) for subsequent traffic assignment procedures. Then, we establish network connectivity by computing TAZ centroids, i.e., the average coordinates of all the blocks within a tract, and generating specialized connector links that bridge these centroids to nearby road nodes. This engineered linkage enables complete trip routing: vehicles originate at origin TAZs, access the main road network through connectors, traverse suitable paths, and finally reach destination TAZs via terminal connectors. The number of TAZs across 10 U.S. cities is displayed in Fig. \ref{fig:4}(b). Substantial variation is observed in TAZ quantities across cities, with Philadelphia (389) and Dallas (328) containing the highest counts, while Honolulu (117) and Portland (157) maintain the lowest volumes.

\subsection{Parameter Settings}

In this section, we implement a systematic parameter calibration process to account for critical factors influencing traffic assignment results. Through iterative refinement of road attributes and link performance functions—particularly BPR function parameters—we achieve convergence to UE conditions. This calibration ensures the model can capture the supply-demand interactions within the urban road network, with the specific parameter settings described below. 

\subsubsection{Road Attributes}

We classify the urban road network into three hierarchical categories: expressways, arterial highways, and local roads. For each road category, we calibrate two fundamental parameters—road capacity and free flow speed—based on the experimental results. Based on previous research \citep{xu2024unified}, the parameter settings for each type of road in 10 U.S. cities are presented in Table \ref{tab:road_capacity} of \ref{Parameter_Settings}.

\subsubsection{MUE Model Parameters}

In the MUE model, the parameter settings for GV and EV cost functions, as well as gasoline and electricity prices, are shown in Tables \ref{tab:cost_params} and \ref{tab:bpr_price_params}. Vehicle operating cost components (maintenance, depreciation, insurance) are derived from AAA's Your Driving Costs study (2023) and U.S. DOE Alternative Fuels Data Center. Environmental externality costs follow \citet{parry2007should} and EPA social cost of carbon estimates. Gasoline prices ($P_{gas}$) are sourced from AAA Gas Prices and GasBuddy for 2023 city-level averages. Electricity prices ($P_{ele}$) are derived from U.S. Energy Information Administration (EIA) state-level residential rates. Additionally, We set the parameters $\alpha$ and $\beta $ in the the BPR function  based on research findings \citep{xu2024unified}. We find that satisfactory results are obtained when parameter $\alpha$ varies between 0.15 and 0.6 and parameter $\beta $ varies between 1.2 and 3. The specific parameters for the BPR function in 10 U.S. cities are shown in Table \ref{tab:bpr_price_params}.

\subsection{Results}

\subsubsection{Probability Distributions of Commuting Distances}

We first tabulate the probability distributions of commuting distances across 10 U.S. cities, as shown in Fig. \ref{fig:5}. Probability distributions and fitted curves demonstrate that the straight-line (Euclidean) commuting distances follow a lognormal distribution. This phenomenon is consistent with findings in urban travel behavior literature \citep{horner2004spatial, yang2012patterns}, where this distribution emerges from the multiplicative nature of spatial accessibility and residential location choices. Unlike the symmetrical normal distribution, the lognormal distribution is right-skewed. This implies that most commuting distances are concentrated within a shorter range, yet a significant “long tail” persists—meaning a considerable portion of residents still require long-distance commutes. In reality, most individuals choose to work near their residences (short commutes); however, due to factors such as housing costs and the distribution of job opportunities, some inhabitants are compelled to undertake long commutes, forming a long-tail distribution. Additionally, the distributions of commuting distances for these cities generally peak between 2.5 and 7.5 $km$, indicating that medium-to-short commutes are the predominant choice for urban residents, consistent with fundamental characteristics of urban travel. Meanwhile, long-distance commuters residing in outlying areas—such as those living in suburbs and working downtown—though constituting a relatively small proportion, exert a significant impact on the capacity of road transport systems, peak-hour congestion, and the quality of life for commuters.

\begin{figure}[ht]
    \centering
    \includegraphics[width=\textwidth]{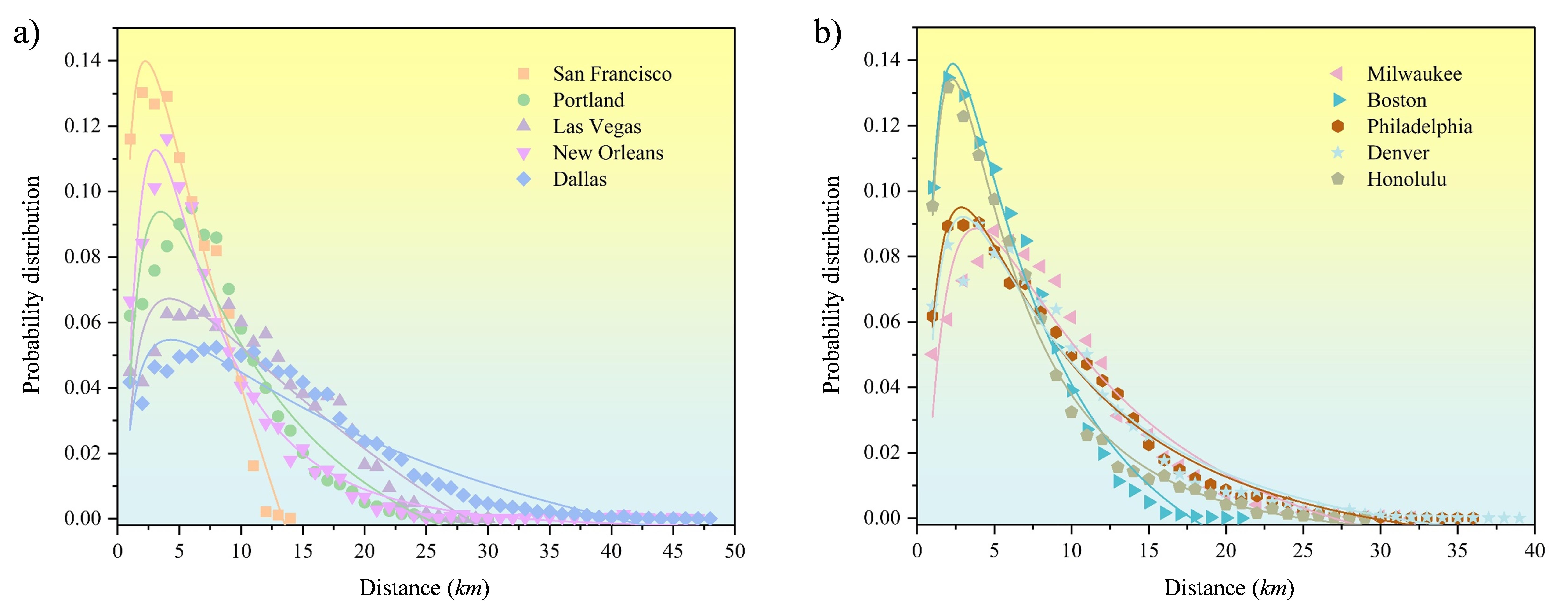}
    \caption{Probability distributions of commuting distances for 10 U.S. cities. a) Probability distributions of commuting distances for San Francisco, Portland, Las Vegas, New Orleans, and Dallas. b) Probability distributions of commuting distances for Milwaukee, Boston, Philadelphia, Denver, and Honolulu.}
    \label{fig:5}
\end{figure}

\subsubsection{Average Travel Time}\label{sec:ave_time}
As the EV penetration increases, the average travel time variations across 10 U.S. cities based on the developed MUE model are illustrated in Fig. \ref{fig:6}. In the baseline scenario where EVs have not been introduced (i.e., EV penetration is 0), the average travel time across these 10 cities exhibits significant heterogeneity in traffic conditions, clearly delineating distinct levels of congestion patterns. Dallas and Philadelphia are prime examples of “congestion hotspots,” with the average travel time reaching 40.53 and 39.18 $min$, respectively—far exceeding other cities. This indicates that commuters in these large cities experienced extremely low commuting efficiency during this period, likely enduring substantial traffic delays. Portland, New Orleans, and Honolulu form the second tier, with the average travel time ranging from 25 to 33 $min$. These cities face notable traffic pressure, though congestion levels are generally less severe than in the first tier. Milwaukee and Denver boast the most efficient traffic conditions, with the average travel time clocking in at just 16.75 and 16.15 $min$, respectively. This indicates that their road transport systems maintain high commuting efficiency during peak hours with relatively few delays. Moderately congested cities exhibit an “efficiency optimization” pattern. For instance, New Orleans achieved the highest relative improvement (10.80\%), realizing the greatest proportion of efficiency gains. Portland and Las Vegas similarly benefit from this pattern, with relative improvement rates approaching 10\%.

\begin{figure}[ht]
    \centering
    \includegraphics[width=\textwidth]{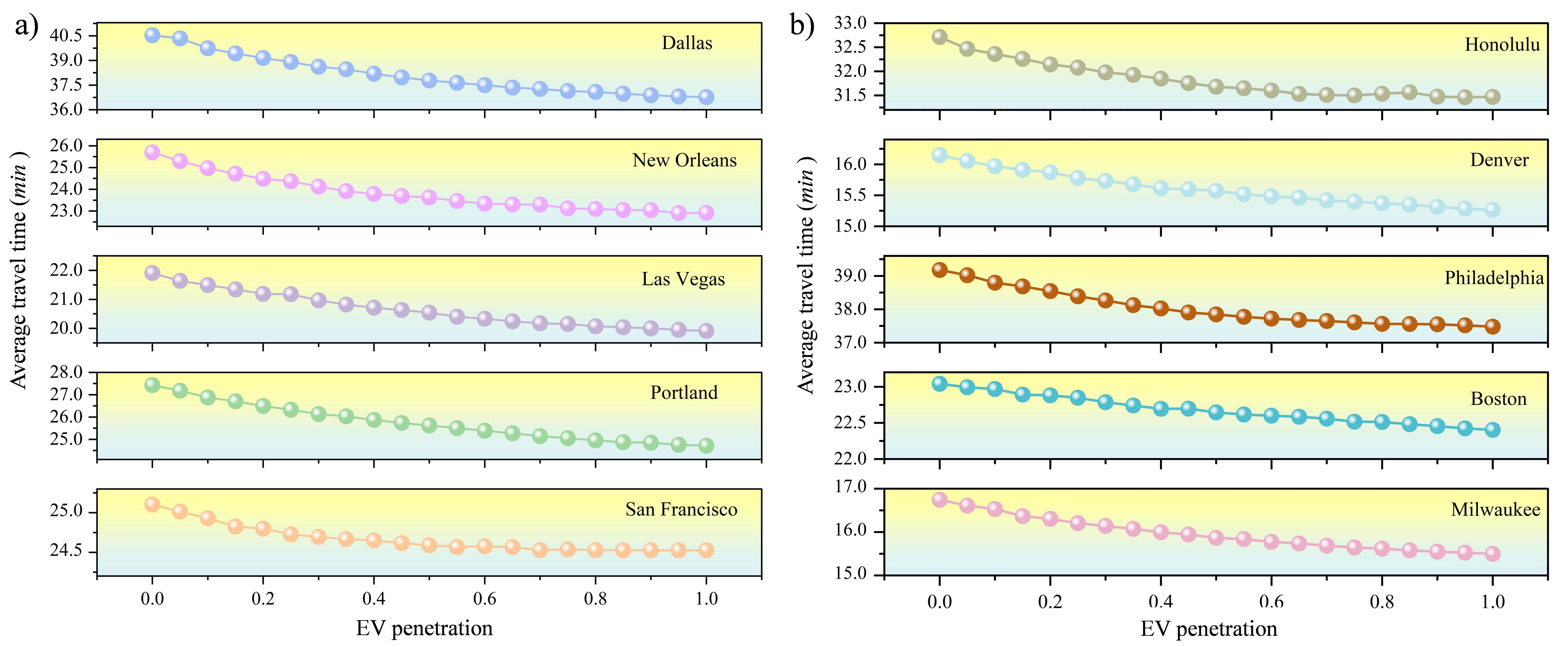}
    \caption{Average travel time for 10 U.S. cities. a) Average travel time for San Francisco, Portland, Las Vegas, New Orleans, and Dallas. b) Average travel time for Milwaukee, Boston, Philadelphia, Denver, and Honolulu.}
    \label{fig:6}
\end{figure}

As the EV penetration increases from 0 to 1, the average travel time decreases across all cities; however, the extent of improvement is closely linked to their initial congestion patterns, revealing distinct optimization mechanisms. Table \ref{tab:travel_time_changes} shows the average travel time, absolute change ($min$), and relative change (\%) for 10 U.S. cities as EV penetration varies. Cities with high congestion intensity typically exhibit an “absolute benefit” pattern, with Dallas serving as a prime example of this pattern. As the city with the longest initial average travel time, it achieve the largest absolute reduction (3.76 $min$), with the introduction of EVs saving significant total travel time for the entire commuting population. Although its relative improvement rate (9.28\%) is not the highest, the actual benefits it delivers are the most substantial. Cities with relatively smooth traffic exhibit a “marginal benefit” pattern. For cities like San Francisco, Boston,  and Denver, where the initial travel time is already relatively short, the savings from EV penetration are limited (absolute reductions are all less than 1 $min$, with relative reductions under 5.6\%). This may stem from their road systems approaching their fundamental capacity, where congestion problems are less pronounced, or their congestion patterns are more entrenched (e.g., constrained by road planning), resulting in smaller marginal benefits from EV technology. The above phenomena indicate that with the introduction of EVs, drivers have become less sensitive to travel distance and less concerned about the economic implications of taking longer routes. Some travelers no longer prioritize finding the shortest spatial route but instead begin to choose routes that are shorter in time and smoother in driving, even if this means covering greater mileage.

Although the increasing of EV penetration can continue to reduce the average travel time, its improvement efficiency is not linear. Analysis reveals that in the later stages when penetration approaches saturation (e.g., exceeding 0.6), progress across all cities slows considerably, exhibiting the characteristic features of a plateau phase (Definition \ref{def:plateau_phase}). Consistent with our theoretical framework, these plateau phases are brief: because EVs displace GVs route-by-route, even minimal EV introduction can fully capture certain cost-advantaged routes, initiating the cascading effect (Definition \ref{def:transition_zone}) almost immediately. In contrast, during the initial phase of penetration from 0, the overall decline is relatively obvious, with the plateau phase being less pronounced and only slightly evident in cities such as Dallas, Philadelphia, and Boston. This phenomenon can be explained by the sequential route displacement mechanism: in the pure GV scenario, due to high marginal travel costs, commuters heavily rely on the shortest routes. This reliance may lead to excessive traffic flow concentrating in core areas such as city centers, thereby exhibiting severe central congestion patterns. The introduction of EVs enables travelers to shift from “congested shortcuts” to “uncongested longer routes,” such as utilizing urban fringes or ring expressways. Once EV penetration crosses a critical threshold, a cascading effect emerges: each additional EV simultaneously displaces GV flows across multiple congested links, triggering network-wide flow redistribution. Therefore, in the initial phase, the system can achieve relatively significant overall congestion relief with minimal EV penetration. This manifests in macro data as a noticeable reduction in travel time during the early stages. As EV penetration reaches a high level, the most easily optimized traffic flow have been addressed, and the marginal benefits of further progress begin to diminish---fewer GV routes remain available for displacement. At this stage, congestion improvements are far less immediate than the initial passenger flow redistribution, leading to a widespread and pronounced plateau phase in the later stages.

\begin{table}[!ht]
\centering
\caption{The average travel time, absolute change ($min$), and relative change (\%) for 10 U.S. cities as EV penetration varies.}
\fontsize{9}{14}\selectfont
\begin{tabularx}{\textwidth}{>{\raggedright\arraybackslash}X 
                             >{\raggedright\arraybackslash}X 
                             >{\raggedright\arraybackslash}X 
                             >{\raggedright\arraybackslash}X 
                             >{\raggedright\arraybackslash}X} 
\toprule
\textbf{Cities} & \multicolumn{2}{l}{\textbf{Average travel time ($min$)}} & \textbf{Absolute change ($min$)} & \textbf{Relative change (\%)}  \\
\cmidrule(lr){2-3}
 & \textbf{EV penetration=0} & \textbf{EV penetration=1} &  &  \\
\midrule
San Francisco & 25.1 & 24.53 & -0.57 & -2.27 \\
Portland & 27.43 & 24.71 & -2.72 & -9.92 \\
Las Vegas & 21.91 & 19.91 & -2.00 & -9.13 \\
New Orleans & 25.69 & 22.92 & -2.77 & -10.78 \\
Dallas & 40.53 & 36.77 & -3.76 & -9.28 \\
Milwaukee & 16.75 & 15.5 & -1.25 & -7.46 \\
Boston & 23.04 & 22.4 & -0.64 & -2.78 \\
Philadelphia & 39.18 & 37.48 & -1.70 & -4.34 \\
Denver & 16.15 & 15.26 & -0.89 & -5.51 \\
Honolulu & 32.71 & 31.47 & -1.24 & -3.79 \\
\bottomrule
\end{tabularx}
\label{tab:travel_time_changes}
\end{table}

\textbf{Empirical Validation of Theorem \ref{thm:plateau}.} The travel time curves in Fig. \ref{fig:6} provide direct empirical validation of the Equilibrium Sensitivity to Penetration Rate (Theorem \ref{thm:plateau}). According to the theorem, equilibrium flows are piecewise smooth in $R_e$, with qualitative transitions occurring at critical thresholds where active path sets change. We classify the 10 cities into three categories based on curve morphology:

\begin{itemize}
    \item \textbf{Type I (Pronounced Transition $\to$ Plateau):} Dallas, New Orleans, Portland, and Philadelphia exhibit steep initial descent ($R_e \in [0, 0.3]$) followed by marked flattening. Estimated critical thresholds: $R_{e,c}^{(1)} \approx 0.25$--$0.35$. This pattern indicates that the first few critical thresholds are crossed early, triggering rapid cascading displacement, after which active sets stabilize.

    \item \textbf{Type II (Gradual Transition):} Las Vegas, Denver, Milwaukee, and Honolulu display approximately linear decline throughout. This suggests either (i) critical thresholds are evenly distributed across $[0,1]$, or (ii) urban topology limits the scope of each cascading event, resulting in continuous but moderate redistribution.

    \item \textbf{Type III (Persistent Plateau):} San Francisco and Boston show minimal variation ($<$3\% reduction), with curves nearly flat from $R_e = 0$. Type III cities exhibit high path overlap ($\rho \to 1$) due to topological constraints (peninsula, historic grid), which per Theorem \ref{thm:plateau}(ii) implies reduced sensitivity $\|\partial x^*/\partial R_e\|$. Thus, `plateau' behavior in these cities reflects the network structure producing high overlap, consistent with both theory and observation.
\end{itemize}

For Type I/II cities, during the initial phase of penetration from 0, the overall decline is relatively obvious. However, Type III cities exhibit nearly linear, gradual decline without pronounced initial benefits, reflecting their topological constraints. These empirical patterns directly correspond to Theorem \ref{thm:plateau}'s predictions: Type I cities exhibit clear ``transition points'' (part iii) with subsequent ``plateau intervals''; Type III cities remain in extended plateau due to shared active paths. The threshold $\epsilon = 0.07$ min per 10\% penetration in Definition \ref{def:plateau_phase} is calibrated to distinguish these categories: Type I cities exceed this rate during $R_e \in [0, 0.3]$, while Type III cities remain below it throughout.

Importantly, the city-type classification is not determined by the GV-to-EV cost ratio $R$ alone. Using the cost parameters from Tables \ref{tab:cost_params}--\ref{tab:bpr_price_params}, we compute $R$ values ranging from 2.33 (Honolulu) to 2.85 (Portland). Yet cities with similar $R$ exhibit different response patterns: San Francisco ($R = 2.82$) shows persistent plateau (Type III), while Portland ($R = 2.86$) exhibits pronounced transition (Type I). This indicates that network topology---not cost differential---appears to play a significant mediating role in congestion response patterns. While our cross-sectional comparison of 10 cities suggests topology matters, we cannot rule out confounding factors (transit mode share, parking policy). Establishing definitive causality would require controlled natural experiments or instrumental variable approaches beyond this study's scope. Specifically, topological constraints (peninsula geography, limited river crossings, dense historic grids) restrict the set of alternative routes available for EV-induced flow redistribution, thereby suppressing the cascading displacement mechanism regardless of cost advantage.

\subsubsection{Potential Savings}

We statistically analyze the potential savings and the difference in potential savings for 10 U.S. cities, as shown in Fig. \ref{fig:7}. The values of potential savings ($PS$) across all cities exhibit an overall upward trend as EV penetration increases, reaching or approaching 100\% when near full penetration (penetration$=1$). This indicates that the comprehensive introduction of EVs can effectively unlock the optimization potential for alleviating traffic congestion. However, there are relatively remarkable differences in the improvement trajectories across different cities. In cities like San Francisco and Honolulu, $PS$ values grow rapidly during the early periods of EV penetration (penetration$<  0.3$), demonstrating that their road networks can swiftly respond to the optimization effects brought by the introduction of EVs, i.e., early investments could yield high returns. Cities such as Portland, Denver, and Milwaukee have experienced more linear and steady increases in the $PS$ values, suggesting that the benefits of improvement are being consistently and stably realized throughout the entire implementation process. Represented by Dallas and Boston, these cities are generally less effective during the early stages of penetration but subsequently show a more pronounced increase in $PS$ values. This pattern aligns with the cascading effect described in Definition \ref{def:transition_zone}: these cities require reaching a critical threshold before EVs can simultaneously displace GV flows across multiple congested links, triggering network-wide flow redistribution. It is worth noting that the growth curves of some cities did not follow a monotonically increasing trend, but instead exhibited slight fluctuations. For example, in San Francisco, the PS value declines slightly from 0.93 to 0.91 at a penetration of 0.6. Honolulu also experience a slight dip when penetration reached 0.8 and 0.85. The fluctuation may stem from the following: when EV penetration reaches a certain threshold, a large number of travelers simultaneously shift to new “optimal routes” that may rapidly form new congestion points on these alternative routes.  Meanwhile, congestion on the original routes has not yet been sufficiently alleviated, triggering a temporary, minor decline in the overall system efficiency. This phenomenon is also consistent with Theorem \ref{thm:plateau}(iii): this penetration level corresponds to a transition point where path active sets reconfigure. The temporary efficiency loss occurs as flows “search” for new equilibrium routes before settling into improved configurations at higher $R_e$. Despite localized fluctuations, the overall upward trend across all cities is clear and dominant. Furthermore, these results hold significant value for transportation planners. Planners should recognize that the benefits of EV penetration are not released uniformly, requiring a focus on critical threshold periods (such as the early stages in San Francisco and Honolulu). While promoting EVs, dynamic traffic management strategies (such as adaptive signal control and congestion pricing) should be implemented to smooth fluctuations during the transition period and accelerate the system's convergence toward a more optimal new equilibrium state.

\begin{figure}[!ht]
    \centering
    \includegraphics[width=\textwidth]{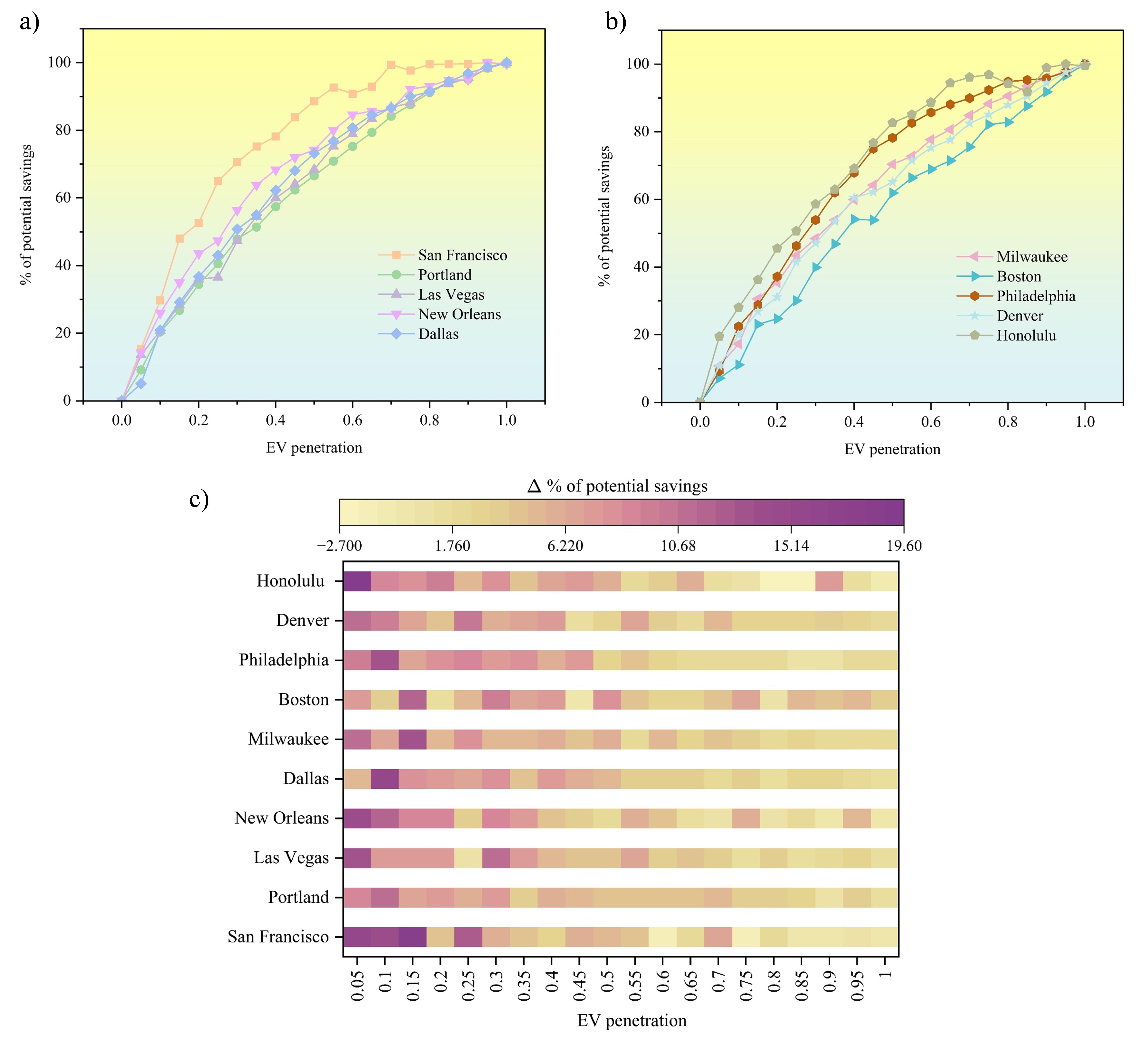}
    \caption{Potential savings and the difference in potential savings for 10 U.S. cities. a) Potential savings for San Francisco, Portland, Las Vegas, New Orleans, and Dallas. b) Potential savings for Milwaukee, Boston, Philadelphia, Denver, and Honolulu. c) Difference in potential savings for 10 U.S. cities. (Note that Fig. \ref{fig:7}(c) uses per-city normalization to highlight each city's internal variation pattern. Cross-city comparisons should reference Table \ref{tab:travel_time_changes} (absolute values) rather than color intensity.)}
    \label{fig:7}
\end{figure}

Fig. \ref{fig:7}(c) displays the difference in potential savings for 10 U.S. cities, exhibiting a similar pattern of variation to previous results. Overall, for most cities, the early to medium stages (penetration$< 0.5$) corresponds to the transition zone (Definition \ref{def:transition_zone}), where cascading displacement effects maximize system benefits. Consistent with Definition \ref{def:plateau_phase}, plateau phases are brief because EVs displace GVs route-by-route---even minimal EV introduction can fully capture certain cost-advantaged routes. In the medium to late stages (penetration$> 0.5$), nearly all cities show a marked slowdown in growth, as the network approaches EV-dominated equilibrium and fewer GV routes remain available for displacement (Definition \ref{def:plateau_phase}). Multiple cities exhibit high growth peaks at specific stages, such as San Francisco reaching a difference in $PS$ of 18.32\% when EV penetration reaches 0.15, and Portland achieving a difference of 11.19\% at 0.1---these peaks reflect the cascading effect where each additional EV simultaneously displaces GV flows across multiple congested links. This empirically validates the “diminishing returns” principle in transportation system optimization: as penetration increases, optimization potential gradually saturates and may even experience temporary efficiency declines due to traffic rebalancing. This suggests planners should target the transition zone ($R_e < 0.5$) for maximum policy impact, where cascading effects amplify returns on EV promotion investments.

\subsubsection{VOC}

Fig. \ref{Total VOC} illustrates the change in the total VOC with increasing EV penetration for 10 U.S. cities. It is noteworthy that as EV penetration increases, the overall VOC values across the networks of the 10 cities show an upward trend, though the magnitude of this varies significantly. Portland, Dallas, New Orleans, and Philadelphia show the largest increases (approximately 8\%–14\%), indicating that in these cities with robust expressway networks and traffic redundancy, the introduction of EVs has drawn more traffic toward a few long-distance, high-grade corridors (expressways and highways), resulting in stronger traffic concentration on the backbone network. The overall VOC increases in San Francisco, Las Vegas, Boston, and Honolulu are moderate (about 5\%–7\%), indicating that while main roads also face rising pressure, the migration from short routes to longer ones is limited due to terrain constraints or the overall compactness of the road network. Denver and Milwaukee exhibit the smallest increase (about 2\%–3\%), indicating that their relatively regular, low-heterogeneity grid networks undergo only minor rebalancing under increasing EV penetration. It should be emphasized that an increase in total VOC does not necessarily imply more severe congestion. In our MUE framework, the equilibrium condition minimizes travelers’ generalized travel cost rather than the system-wide total VOC. As the operating cost of EVs decreases, the strategy of “travelling slightly farther in exchange for higher speeds” becomes more attractive, prompting travelers to shift from numerous short, intersection-dense urban streets that are highly sensitive to congestion toward a limited set of long-distance expressways and highways with large capacity and lower congestion sensitivity. This redistribution of flows raises the VOC on these backbone corridors and consequently increases the total VOC. However, because these corridors can still provide substantially lower travel times than local roads even at relatively high saturation levels, the system-wide average travel time continues to decline.

\begin{figure}[!ht]
    \centering
    \includegraphics[width=\textwidth]{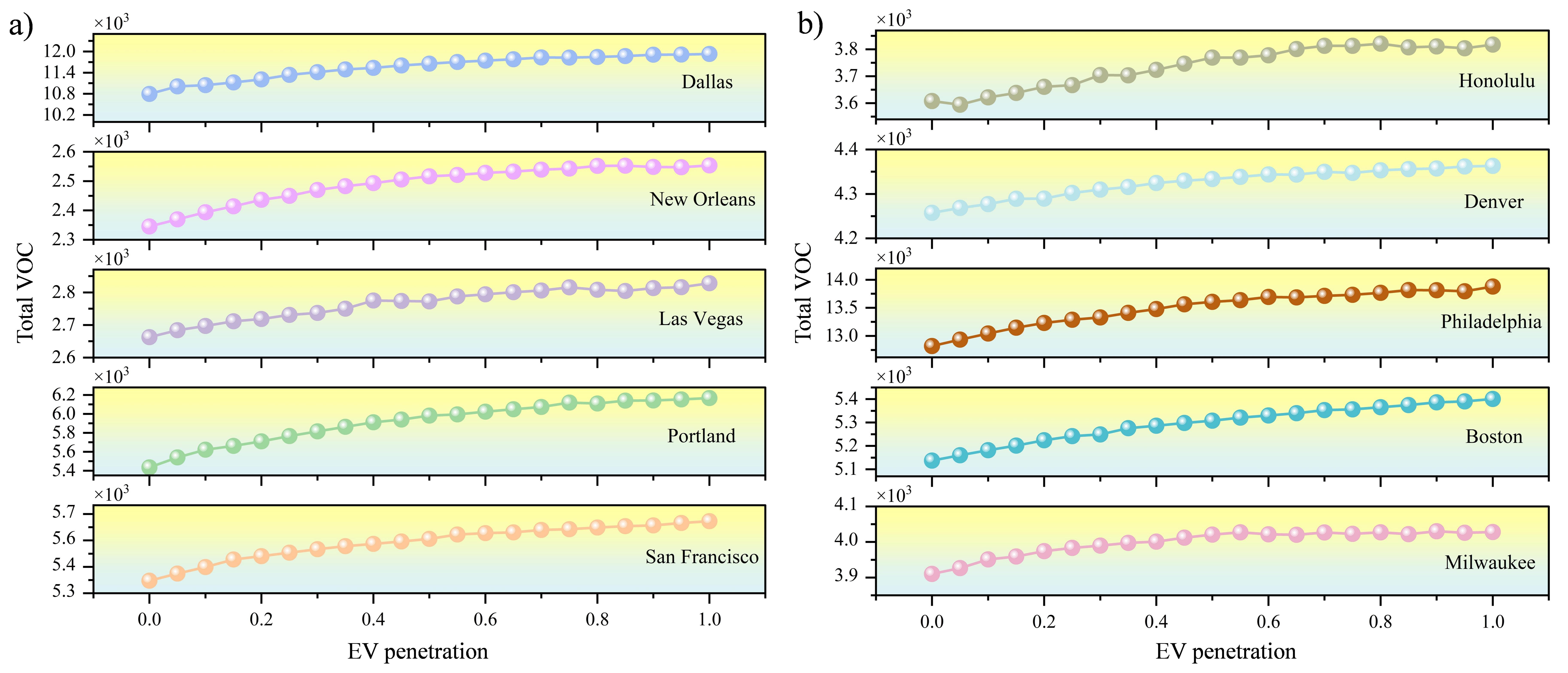}
    \caption{Total VOC for 10 U.S. cities. a) Total VOC for San Francisco, Portland, Las Vegas, New Orleans, and Dallas. b) Total VOC for Milwaukee, Boston, Philadelphia, Denver, and Honolulu. (Note: Y-axes differ across subplots to show each city's full range. For cross-city comparisons, refer to percentage change values rather than visual slopes.) }
    \label{Total VOC}
\end{figure}

To more clearly illustrate changes in VOC distributions, Fig. \ref{fig:VOC_10cities} in \ref{VOC_Difference_Distributions} plots link-level VOC difference distributions across four scenarios with EV penetration $R_{e}$ of 0.25, 0.5, 0.75, and 1, compared to the baseline scenario ($R_{e}$=0, all-full GVs). The labels indicate VOC difference values, with colors ranging from blue to red (purple) representing a clear decrease to a significant increase in VOCs, respectively.

Building upon the Type I/II/III classification established in Section \ref{sec:ave_time}, we identify three distinct spatial redistribution patterns that provide mechanistic explanations for why certain cities exhibit pronounced travel time improvements while others show persistent plateaus.

\textbf{Pattern A: Backbone Reinforcement} (high network redundancy $\to$ Type I response).

Cities with well-developed highway frameworks and multiple radial-loop corridors exhibit clear ``backbone reinforcement,'' corresponding to the steep travel time descent observed in Type I cities. Dallas provides the clearest example: as $R_{e}$ increases, multiple radial highways progressively ``light up,'' forming continuous warm-colored corridors extending from city center to suburbs, while surrounding grid roads shift to blue (VOC reduction). EV penetration redirects cross-district trips from dense local networks onto the highway backbone; total VOC increases substantially, yet higher highway speeds yield pronounced system-wide travel time reduction. Portland and Philadelphia exhibit similar patterns: Portland's north–south highway corridors absorb cross-river trips (total VOC increase 8\%–14\%), while Philadelphia's central radial corridor strengthens as riverfront highways are relieved. The common mechanism: high network redundancy enables the cascading displacement effect (Definition \ref{def:transition_zone})---EVs simultaneously capture multiple alternative corridors, triggering network-wide redistribution that manifests as steep travel time improvements at $R_e \in [0, 0.3]$.

\textbf{Pattern B: Bottleneck Concentration} (topological constraints $\to$ Type I/III mixed response).

Cities constrained by geography or infrastructure bottlenecks exhibit spatially concentrated VOC changes confined to limited critical links. New Orleans exemplifies this pattern: VOC changes concentrate around Mississippi River crossings and CBD, with bridge approaches showing pronounced warming while interior grids vary minimally. Limited crossings constrain redistribution to a small link set, reorganizing local bottlenecks while the overall pattern remains river–CBD focused. San Francisco (peninsula restricts cross-bay alternatives), Honolulu (narrow coastal belt confines flow to single east–west corridor), Boston (historic grid with weak hierarchy), and Las Vegas (simple radial layout) show similar concentration effects. These constraints suppress cascading: even with favorable cost ratios ($R \approx 2.6$–$2.8$), limited alternatives prevent network-wide redistribution. This explains why San Francisco ($R = 2.82$) exhibits Type III plateau despite cost advantage similar to Portland ($R = 2.85$, Type I).

\textbf{Pattern C: Marginal Rebalancing } (homogeneous grid $\to$ Type II response).

Cities with regular grid layouts and minimal road hierarchy heterogeneity exhibit only marginal VOC changes. Milwaukee's network shows VOC differences within $[-0.3, 0.3]$ regardless of $R_e$, with only slight warming on highway connectors. EV penetration triggers no substantial path reconfiguration; the network undergoes minor rebalancing only. Denver exhibits similar behavior: moderate freeway-local performance gap and abundant substitution possibilities prevent traffic concentration. In homogeneous grids, absence of clear ``backbone vs. feeder'' differentiation means EVs cannot selectively exploit high-capacity corridors, resulting in gradual distributed improvement (Type II linear decline) rather than concentrated cascading.

\textbf{Synthesis.} The correspondence between spatial patterns and travel time response types confirms that network topology governs EV-induced congestion reshaping: Pattern A (high redundancy) enables cascading $\to$ Type I steep improvement; Pattern B (bottleneck constraints) confines redistribution $\to$ mixed Type I/III depending on constraint severity; Pattern C (homogeneous grids) distributes changes uniformly $\to$ Type II gradual improvement. This has direct policy implications: cities should assess network topology before projecting EV benefits, as cost-based incentives alone cannot overcome structural constraints.

\subsubsection{Road Utilization Rate}
As the EV penetration increases, the road utilization rates in the 10 U.S. cities exhibit discernible and consistent patterns, as illustrated in the Fig. \ref{Road Utilization Rate}. Road utilization rates in various cities rise rapidly in the initial phase, enter a period of high-level narrow fluctuations with a slight upward trend in the middle phase, and ultimately experience a certain degree of decline in the later phase. The introduction of EVs immediately triggers a spatial redistribution of traffic flow, as travelers shift from congested "shortcuts" to underutilized alternative routes, rapidly "activating" road network resources within a short timeframe---this reflects the sequential route displacement mechanism where EVs capture cost-advantaged routes one by one. In the medium stage, road network resource utilization approaches saturation, with traffic flow stabilizing at a new equilibrium state. Road utilization rates show a slight increase, while the marginal benefits of further optimization diminish as fewer GV routes remain available for displacement. The later appearing decline indicates that road utilization peaks when traditional GVs and EVs coexist, confirming the travel behavior where commuters explore and utilize the maximum available routes under mixed traffic conditions. When all vehicles transition to EVs, decision-making patterns may become more homogeneous, converging instead toward a more concentrated yet efficient traffic state. This overall change pattern demonstrates that the widespread adoption of EVs can effectively enhance the utilization efficiency of road network resources. However, with fixed OD demand, road network optimization faces natural upper limits. In later stages, complementary management measures (e.g., congestion pricing, lane management) should be implemented to further improve road utilization rates. Additionally, cities like San Francisco and Philadelphia exhibit the highest road utilization rates (consistently$>$ 0.34); cities such as Portland and Boston show moderate utilization rates (between 0.27 and 0.33); while cities including Las Vegas, New Orleans, and Denver have the lowest road utilization rates ($<$ 0.27).

\begin{figure}[!ht]
    \centering
    \includegraphics[width=\textwidth]{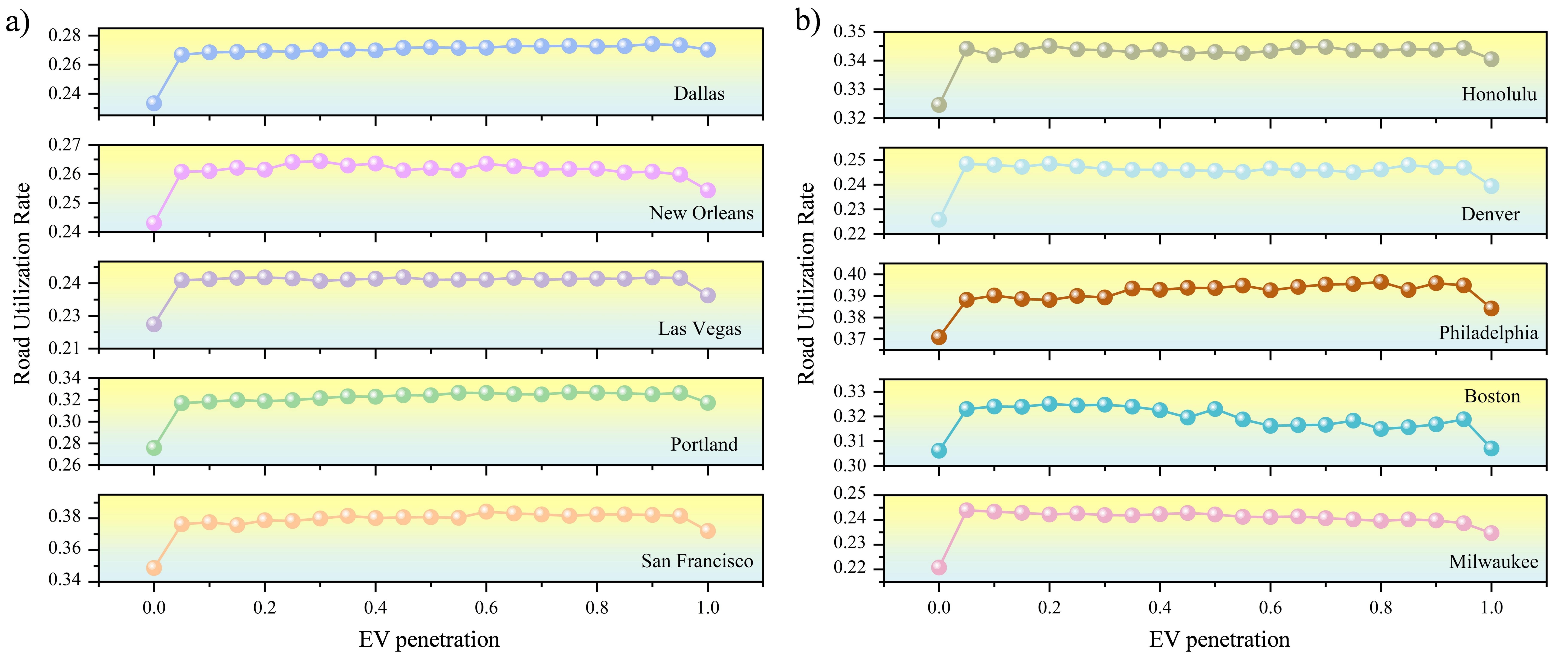}
    \caption{Road utilization rates for 10 U.S. cities. a) Road utilization rates for San Francisco, Portland, Las Vegas, New Orleans, and Dallas. b) Road utilization rates for Milwaukee, Boston, Philadelphia, Denver, and Honolulu.}
    \label{Road Utilization Rate}
\end{figure}

\subsubsection{Link Congested Time}
Fig. \ref{Link Congested Time 1} and \ref{Link Congested Time 2} illustrate the distribution characteristics of link congested time versus link length under different EV penetration scenarios for U.S. cities. Each subgraph corresponds to a city, with scatter points of different shapes and colors representing five typical scenarios of EV penetration: $R_{e}$ =0, 0.25, 0.5, 0.75, and 1. For comparative purposes, we calculate the average congested time for each city's links across different scenarios after binning by length. These results are then overlaid on the same coordinate system, enabling a visual assessment of how EV penetration affects congestion levels across links of varying lengths. From an overall perspective, the curves for all cities exhibit a distinct upward-right trend, meaning that longer links correlate with higher average congested times: Short links ($<0.5 km$) generally have average congested times below 1 minute, medium-length links (approximately $0.5\text{-}2 km$) cluster within the 1–3 minute range, while long links ($>2 km$, reaching over 8 km in some cities) can see average congested times rise to 5–10 minutes. Compared to differences across varying lengths, scatter plots within the same length interval show higher overlap across different $R_{e}$ scenarios, exhibiting only slight upward or downward shifts. This indicates that under given OD demand and road network structures, changes in EV penetration do not substantially alter congested time classes at the link level. The effect of EV penetration is smaller than the role played by link length and urban structure itself.

\begin{figure}[!ht]
    \centering
    \includegraphics[width=\linewidth,
    trim=1mm 0.7mm 1.5mm 1mm]{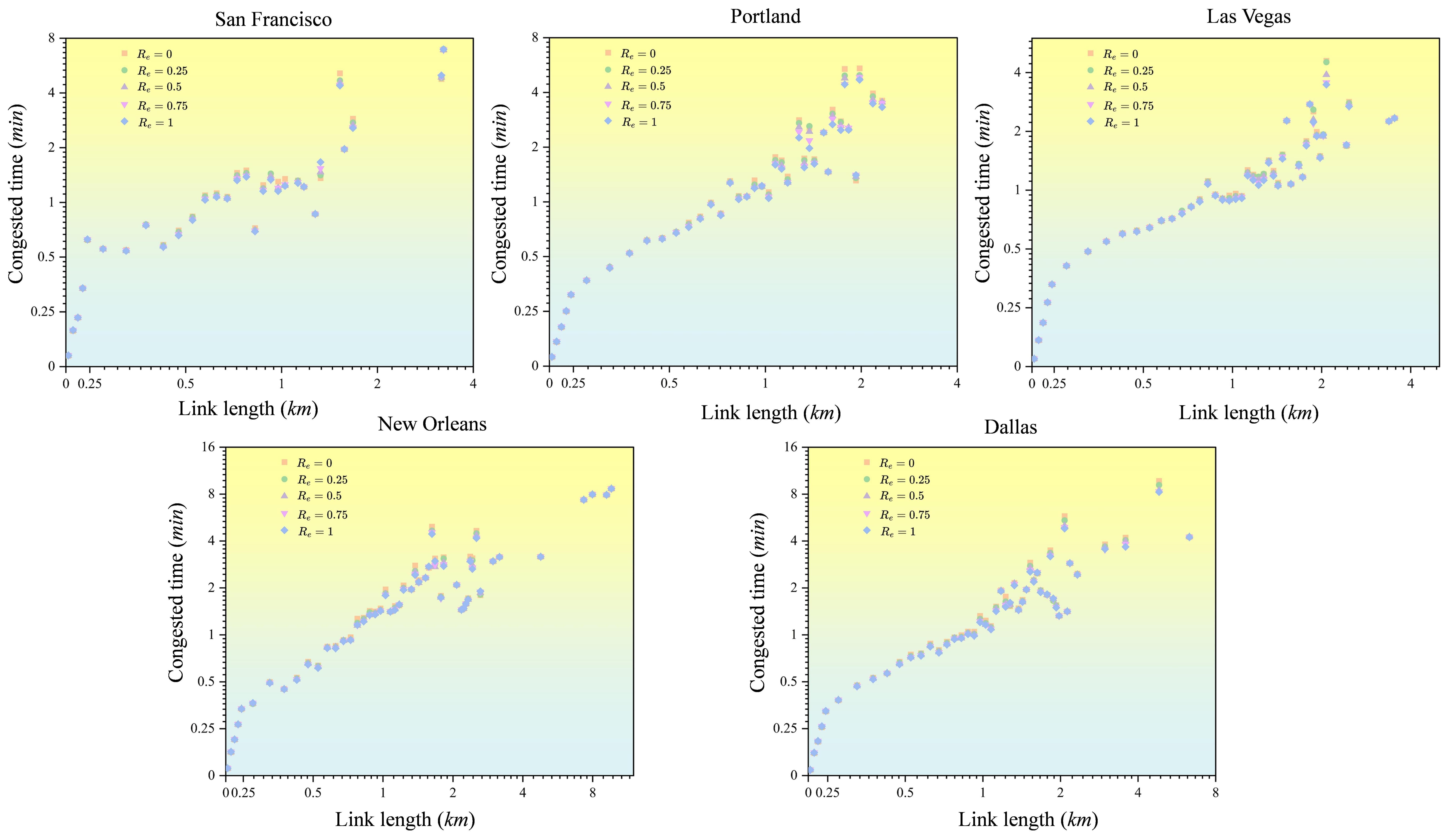}
    \caption{Link congested time for San Francisco, Portland, Las Vegas, New Orleans, and Dallas.}
    \label{Link Congested Time 1}
\end{figure}

\begin{figure}[!ht]
    \centering
    \includegraphics[width=\linewidth,
    trim=1mm 0.7mm 1.5mm 1mm]{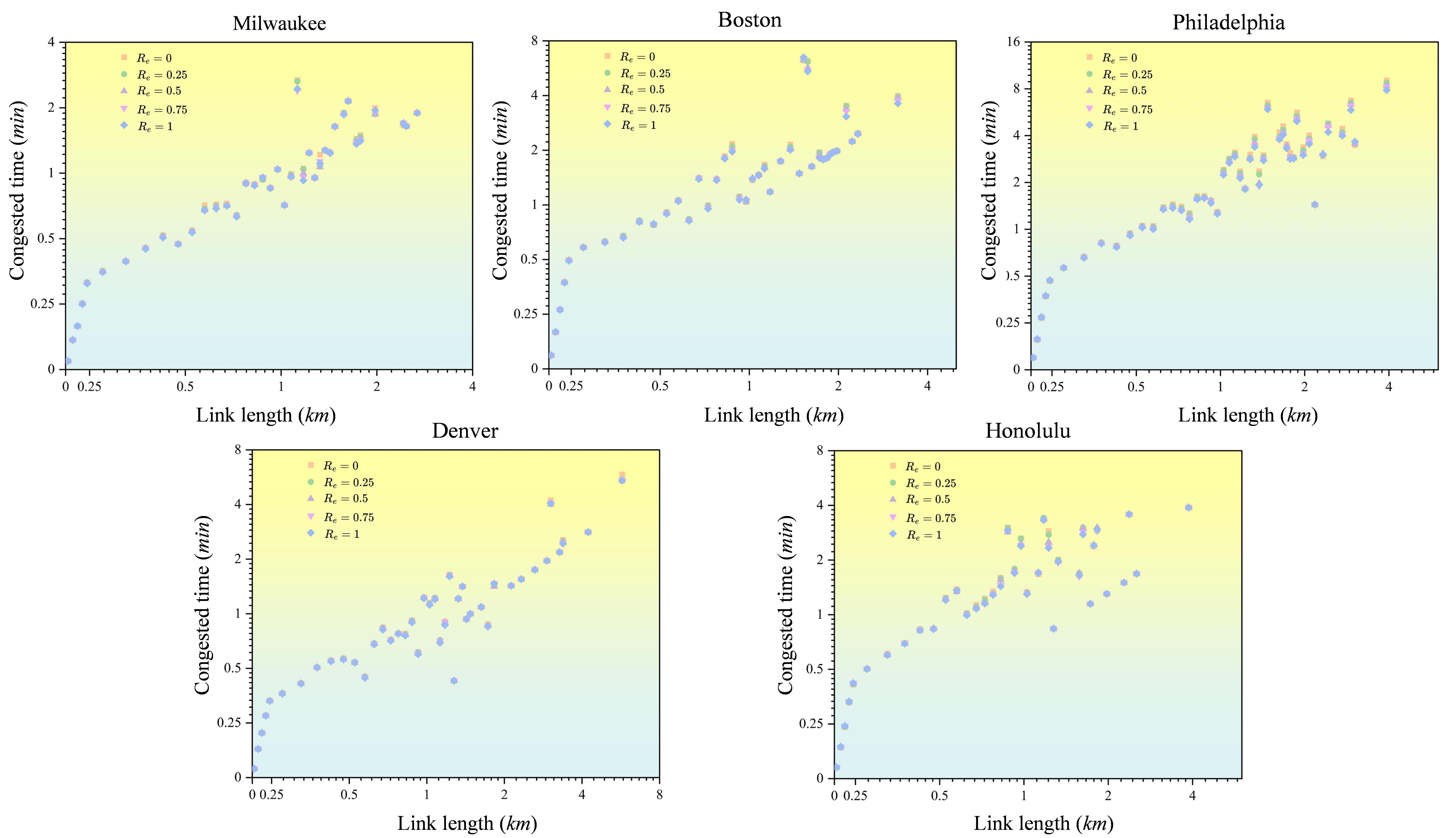}
    \caption{Link congested time for Milwaukee, Boston, Philadelphia, Denver, and Honolulu.}
    \label{Link Congested Time 2}
\end{figure}

Then, by further comparing different EV penetration scenarios, we observe that for short-length links, data points across all five $R_{e}$  values nearly overlap completely, indicating that congested time shows little sensitivity to EV penetration rates. For medium-length links, as $R_{e}$  increases from 0 to 1, average congested time slightly decreases in some cities, but the reduction is typically less than 0.2–0.3 minutes. For long-length links, differences driven by EV popularity become relatively more pronounced. This is particularly evident in cities like Dallas and New Orleans, with numerous long-distance arterial roads and higher congestion improvement potential, where points with $R_{e}=1$  are generally slightly lower than those with $R_{e}=0$ for the same distance. This indicates that the operational cost advantage of EVs partially redistributes demand from highly congested long links to relatively suboptimal but smoother routes through path selection, thereby producing a moderate “peak-shaving” effect on heavily loaded long links. 

Furthermore, examining inter-city heterogeneity reveals that cities with more traditional high-density, central-radiating structures (e.g., San Francisco, Boston, Milwaukee) exhibit nearly identical $R_{e}$ scenario curves across all length intervals. This suggests that within such networks, limited choices for main roads and minimal detour redundancy mean that even electric vehicles with significant cost advantages struggle to substantially alleviate link-level congestion through route restructuring. Conversely, cities with higher network redundancy and stronger functional mixing (such as Dallas and New Orleans) exhibit a clearer trend of congestion shifting to later hours on long-distance corridors, where EV penetration delivers more pronounced congestion-alleviating effects. This finding aligns with previous  conclusions based on the average travel time and VOC distributions: the improvement in urban congestion effects shows pronounced city-specific heterogeneity, influenced both by the degree of EV penetration and the redundancy and spatial structure of the urban road network.

\subsubsection{Difference in Delay Factor}
The results of the difference in delay factors for 10 U.S. cities are shown in Fig. \ref{Difference in Delay Factor}. This metric is calculated by determining the difference in delay factors of each link between two scenarios: full EV penetration (penetration$= 1$) and a baseline without EV (penetration$= 0$). It can be observed that the difference in delay factors of links exhibits widespread instability, with all cities showing both positive and negative values simultaneously. In particular, positive values indicate that the delay factor of the links during full GVs exceeds the baseline (full GVs), i.e., performance becomes worse (travel delays are more severe), while negative values indicate the opposite. This indicates that no city's road network consistently exhibits link performance that is either better or worse than the benchmark level. While links for different cities display distinct performance tendencies (i.e., overall shifts in the difference in delay factors). Certain cities (such as Portland and New Orleans) exhibit relatively more or larger positive values; whereas other cities (such as Milwaukee) show negative or larger negative values with greater relative frequency. Although the difference in delay factors for all cities fluctuates around 0, there are significant variations in the magnitude of fluctuations and “tail risk” (the likelihood of extreme values) across different cities. Based on this, we can broadly categorize cities into three types: (i) Cities with moderate fluctuations and relatively stable delay factors. The values for representative cities such as Portland, Las Vegas, and Milwaukee are overwhelmingly concentrated within the narrow range between -1 and 1. Even when deviations occur, extreme values are rare. This indicates that on most links in these cities, the travel time delays for fully GVs and fully EVs are relatively comparable. (ii) Cities with high volatility and links exhibiting significantly higher EV delays. These cities possess a large number of near-zero values, yet also show remarkable deviations with high positive values. For example, New Orleans record an extreme positive value of 4.28, while Boston's highest value reach 2.17. (iii) Cities with high volatility and links exhibiting significantly lower EV delays. These cities experience notably low negative values. For example, San Francisco appears with a minimum value of -2.75; Dallas's minimum value also amounts to -2.13. These extremely low negative values reflect the great potential for optimizing these links through the introduction of EVs.

\begin{figure}[!ht]
    \centering
    \includegraphics[width=\textwidth]{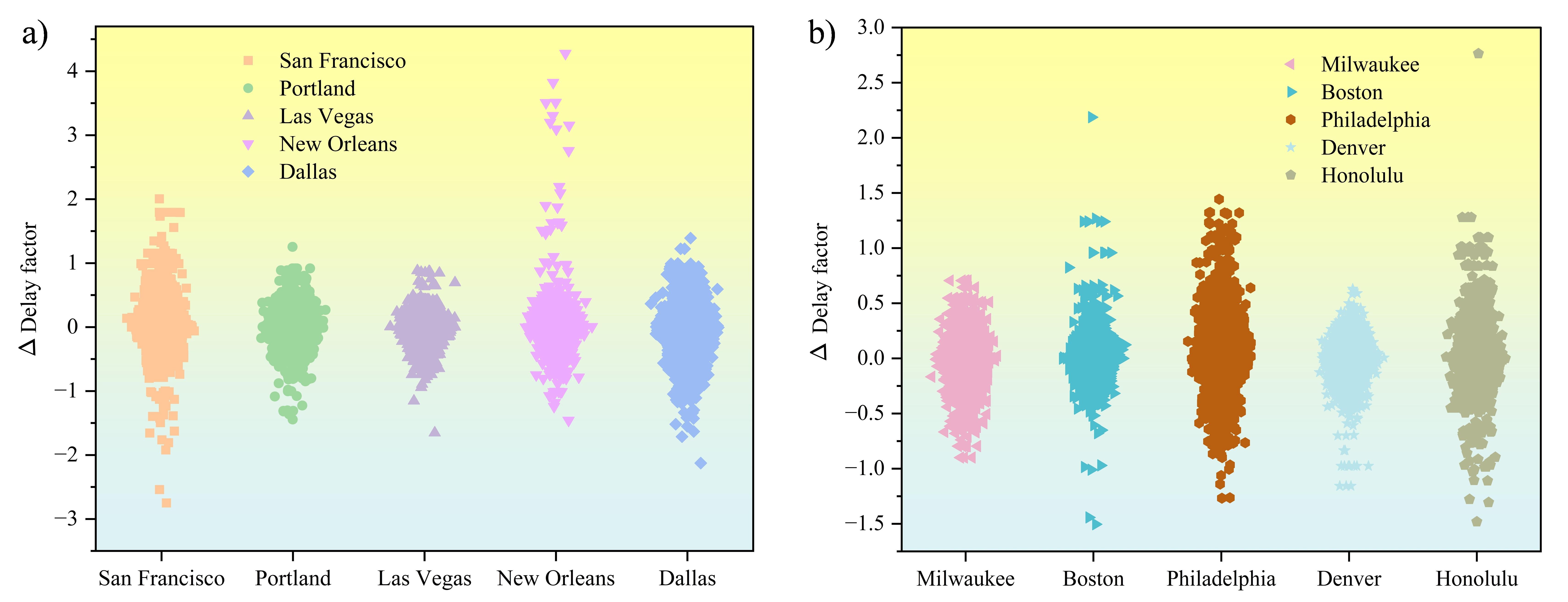}
    \caption{Difference in delay factors for 10 U.S. cities. a) Difference in delay factors for San Francisco, Portland, Las Vegas, New Orleans, and Dallas. b) Difference in delay factors for Milwaukee, Boston, Philadelphia, Denver, and Honolulu.}
    \label{Difference in Delay Factor}
\end{figure}

\section{Conclusion and Future Work}

This study constructs a modeling framework integrating MUE assignment with urban heterogeneity analysis for mixed-traffic scenarios involving GVs and EVs. We apply the bi-conjugate Frank-Wolfe algorithm to the GV-EV MUE problem and prove convergence under our specific formulation. Using actual road networks and commuting demand data from 10 representative U.S. cities, it reveals the congestion reshaping mechanisms across different cities under varying EV penetration levels. The research results show that: (i) EV penetration reduces congestion in most cities, with benefits ranging from 2.27\% (San Francisco) to 10.78\% (New Orleans). Cities with constrained topologies (Type III) show limited response regardless of penetration level, suggesting that congestion benefits are network-structure dependent. (ii) EVs displace GVs through a sequential route displacement mechanism: benefits emerge immediately as EVs capture cost-advantaged routes, accelerate through the transition zone ($R_e \in [0, 0.5]$) where cascading effects trigger network-wide flow redistribution, and eventually diminish in the plateau phase ($R_e > 0.5$) as fewer GV routes remain available for displacement. (iii) The expansion of EV penetration does not alleviate traffic congestion in a uniform manner. Instead, it typically reshapes congestion patterns (as defined in Definition \ref{def:congestion_pattern}) by concentrating traffic flow onto a few expressway/highway corridors, resulting in distinct traffic evolution trends across different cities: some cities experience increased congestion on main corridors but significantly reduce overall travel time, while others undergo localized bottleneck restructuring, and still others see only minor changes. 

This paper can also provide policy recommendations for sustainable urban transportation planning, design, and management. Specifically: (i) Implement EV promotion and planning strategies aligned with urban structure. For Type I/II cities with network redundancy, targeting the transition zone ($R_e < 0.5$) yields maximum policy impact. For Type III cities with topological constraints, EV promotion alone offers limited congestion relief; complementary infrastructure investments (new crossings, capacity expansion at bottlenecks) may be necessary prerequisites. (ii) Advance infrastructure upgrades centered on capacity optimization and intelligent guidance. On high-traffic roads, measures such as dynamic lane management, reduction of roadside parking, and optimized ramp design can be deployed. Concurrently, users should be provided with real-time, reliable “smooth-flow route” guidance. (iii) Establish a data-driven dynamic congestion management and collaborative governance mechanism. The platform for short-term congestion forecasting and early warning should be established, with dynamic congestion pricing or incentive measures applied by time slot and road segment. Moreover, EV penetration policies should be coordinated with public transportation upgrades, slow-moving system improvements, and work-residence balance planning.

Although this work has made some progress in characterizing the reshaping of urban congestion patterns against the backdrop of EV penetration, certain limitations remain. Specifically: (i) While existing models lack explanatory power for the specific GV-EV cost differential mechanism, we acknowledge that our model also abstracts away certain factors such as charging behavior and range anxiety. (ii) Our city selection prioritizes geographic and topological diversity within data availability constraints. However, the sample excludes Midwest agricultural regions, small cities ($<$100k population). (iii) Our static peak-hour analysis abstracts away temporal dynamics including intraday demand fluctuation, weekday/weekend variation, and seasonal effects. Additionally, en-route EV charging delays are not modeled, an important factor for longer commutes. (iv) Our analysis treats $R_e$ as exogenous. In reality, EV adoption is endogenous: congested cities may adopt stronger EV incentives. This simultaneity does not affect our comparative static equilibrium analysis, but does mean our results should not be interpreted as forecasting dynamic adoption paths. Overall, these limitations will be explored in depth in future work.

\section*{Declaration of Interest}
The authors declare no conflict of interest.

\section*{Data Availability }
The processed network data, OD matrices, and implementation code will be made available upon publication. The Frank-Wolfe algorithm converges to a unique equilibrium due to strict convexity (Theorem \ref{thm:existence}), ensuring reproducibility regardless of initialization.

\section*{Author Contributions}
Xiaohan Xu: Conceptualization, Data Curation, Formal Analysis, Methodology, Validation, Visualization, Writing - Original Draft, Revision; Wei Ma: Data Curation, Formal Analysis, Methodology, Validation, Visualization, Revision;
Zhiheng Shi: Formal Analysis, Methodology, Writing - Review \& Editing, Revision ; Xiaotong Xu: Methodology, Writing - Review \& Editing, Revision; Bin He: Supervision, Writing - Original Draft; Kairui Feng: Methodology, Resources, Supervision, Writing - Original Draft, Revision.

\section*{Acknowledgements} 
 Kairui Feng is supported by Shanghai Municipal Science and Technology Commission Explorers Program (24TS1401600) and Xiaomi Foundation. Kairui Feng and Bin He are supported by the National Natural Science Foundation of China (Grant No. 62088101). This support is gratefully acknowledged.

\appendix

\section{Parameter Settings}\label{Parameter_Settings}
Model parameter settings are shown in Tables \ref{tab:road_capacity}, \ref{tab:cost_params}, and \ref{tab:bpr_price_params}.

\section{VOC Difference Distributions for 10 U.S. Cities}\label{VOC_Difference_Distributions}

VOC difference distributions for 10 U.S. cities are shown in Fig. \ref{fig:VOC_10cities}.

\renewcommand{\thetable}{A\arabic{table}}
\setcounter{table}{0}
\newpage
\begin{table}[!ht]
\centering
\caption{Capacity ($veh/h$) and free flow speed ($km/h$) for each type of roads in 10 U.S. cities.}
\fontsize{9}{14}\selectfont
\begin{tabularx}{\textwidth}{l l *{6}{>{\raggedright\arraybackslash}X}} 
\toprule
\multirow{2}{*}{\textbf{No.}} & \multirow{2}{*}{\textbf{City}} & \multicolumn{3}{l}{\textbf{Capacity}} & \multicolumn{3}{l}{\textbf{Free Flow Speed}} \\
\cmidrule(lr){3-5} \cmidrule(lr){6-8}
 & & \textbf{Expressways} & \textbf{Highways} & \textbf{Local Roads} & \textbf{Expressways} & \textbf{Highways} & \textbf{Local Roads} \\
\midrule
1 & San Francisco & 2,200 & 2,000 & 1,400 & 90 & 60 & 40 \\
2 & Portland & 2,200 & 2,000 & 1,400 & 90 & 65 & 45 \\
3 & Las Vegas & 2,200 & 2,000 & 1,400 & 90 & 60 & 40 \\
4 & New Orleans & 2,200 & 2,000 & 1,400 & 90 & 60 & 40 \\
5 & Dallas & 2,200 & 2,000 & 1,400 & 90 & 65 & 45 \\
6 & Milwaukee & 2,200 & 2,000 & 1,400 & 90 & 65 & 45 \\
7 & Boston & 2,200 & 2,000 & 1,300 & 60 & 45 & 30 \\
8 & Philadelphia & 2,000 & 1,800 & 1,200 & 90 & 60 & 30 \\
9 & Denver & 2,000 & 1,800 & 1,300 & 90 & 60 & 35 \\
10 & Honolulu & 2,200 & 2,000 & 1,400 & 90 & 60 & 40 \\
\bottomrule
\end{tabularx}
\label{tab:road_capacity}
\end{table}

\begin{table}[!ht]
    \centering
    \caption{Parameter settings for GV and EV cost functions. (All cost components are in \$/mile.)}
    \label{tab:cost_params}
    \fontsize{9}{12}\selectfont
    \renewcommand{\arraystretch}{1.3} 
    \begin{tabular}{@{} l l @{\hspace{8pt}} l l @{\hspace{8pt}} l l @{\hspace{8pt}} l l @{\hspace{8pt}} l l @{}}
        \toprule
        \textbf{Parameter} & \textbf{Value} & 
        \textbf{Parameter} & \textbf{Value} & 
        \textbf{Parameter} & \textbf{Value} & 
        \textbf{Parameter} & \textbf{Value} & 
        \textbf{Parameter} & \textbf{Value} \\
        \midrule
        $C_{gv}^{maint}$ & 0.101 & $C_{gv}^{fix}$ & 0.18 & $C_{gv}^{dep}$ & 0.25 & 
        $C_{gv}^{ins}$ & 0.08 & $C_{gv}^{add}$ & 0.02 \\
        $C_{gv}^{env}$ & 0.055 & $r_{dis}$ & 1.609 & $MPG_{gv}$ & 25 & 
        $C_{ev}^{maint}$ & 0.064 & $C_{ev}^{fix}$ & 0.08 \\
        $C_{ev}^{dep}$ & 0.10 & $C_{ev}^{ins}$ & 0.08 & $C_{ev}^{add}$ & 0.01 & 
        $C_{ev}^{env}$ & 0.01 & $C_{ev}^{sub}$ & -0.075 \\
        $\kappa_{gal}$ & 33.7 & $MPGe_{ev}$ & 110 & & & & & & \\
        \bottomrule
    \end{tabular}
\end{table}

\begin{table}[!ht]
    \centering
    \caption{Parameters for energy prices and the BPR function in 10 U.S. cities.}
    \fontsize{9}{14}\selectfont
    \begin{tabular}{p{1cm}p{2.2cm}p{1.2cm}p{1.2cm}p{1.2cm}p{1.2cm}}
        \toprule
        \textbf{No.} & \textbf{City} & \textbf{$\alpha$} & \textbf{$\beta$} & \textbf{$P_{gas}$} & \textbf{$P_{ele}$} \\
        \midrule
        1 & San Francisco & 0.5 & 1.8 & 5.100 & 0.1534 \\
        2 & Portland & 0.5 & 1.2 & 4.057 & 0.0918 \\
        3 & Las Vegas & 0.5 & 1.3 & 4.354 & 0.1215 \\
        4 & New Orleans & 0.6 & 1.8 & 3.103 & 0.0776 \\
        5 & Dallas & 0.6 & 1.3 & 3.132 & 0.1098 \\
        6 & Milwaukee & 0.5 & 1.5 & 3.246 & 0.1395 \\
        7 & Boston & 0.25 & 2 & 3.415 & 0.1491 \\
        8 & Philadelphia & 0.5 & 1.2 & 3.487 & 0.1290 \\
        9 & Denver & 0.5 & 1.5 & 3.111 & 0.1161 \\
        10 & Honolulu & 0.5 & 1.5 & 4.774 & 0.3510 \\
        \bottomrule
    \end{tabular}
    \label{tab:bpr_price_params}
\end{table}

\clearpage
\setcounter{figure}{0}
\renewcommand{\thefigure}{A\arabic{figure}}

\begin{figure}[H]
    \centering

    % First subfigure (top)
    \begin{subfigure}{\textwidth}
        \centering
        \includegraphics[height=0.47\textheight,keepaspectratio]{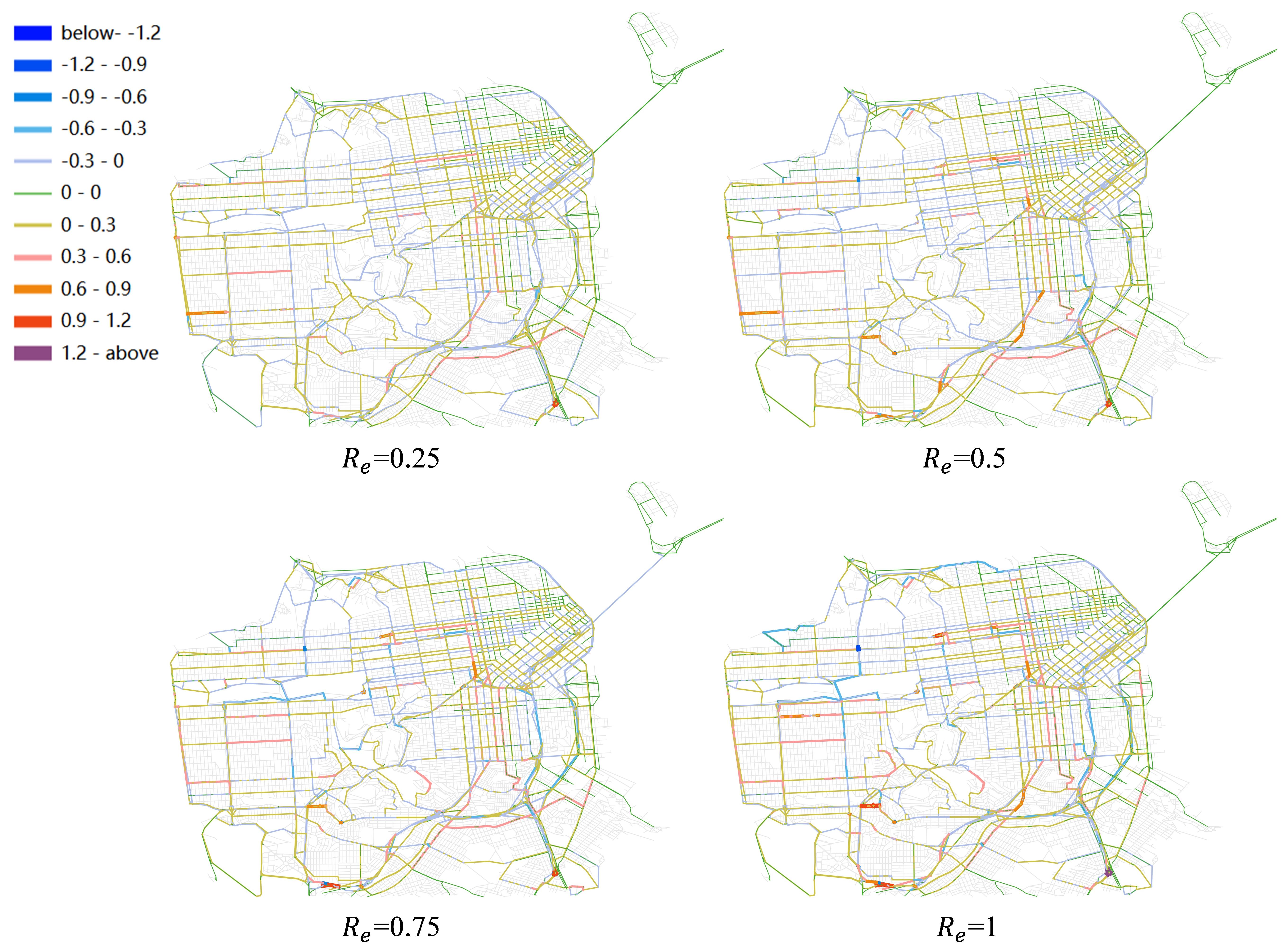}
        \caption{VOC difference distributions for San Francisco.}
        \label{fig:d_SanFrancisco}
    \end{subfigure}

    \vfill

    % Second subfigure (bottom)
    \begin{subfigure}{\textwidth}
        \centering
        \includegraphics[height=0.43\textheight,keepaspectratio]{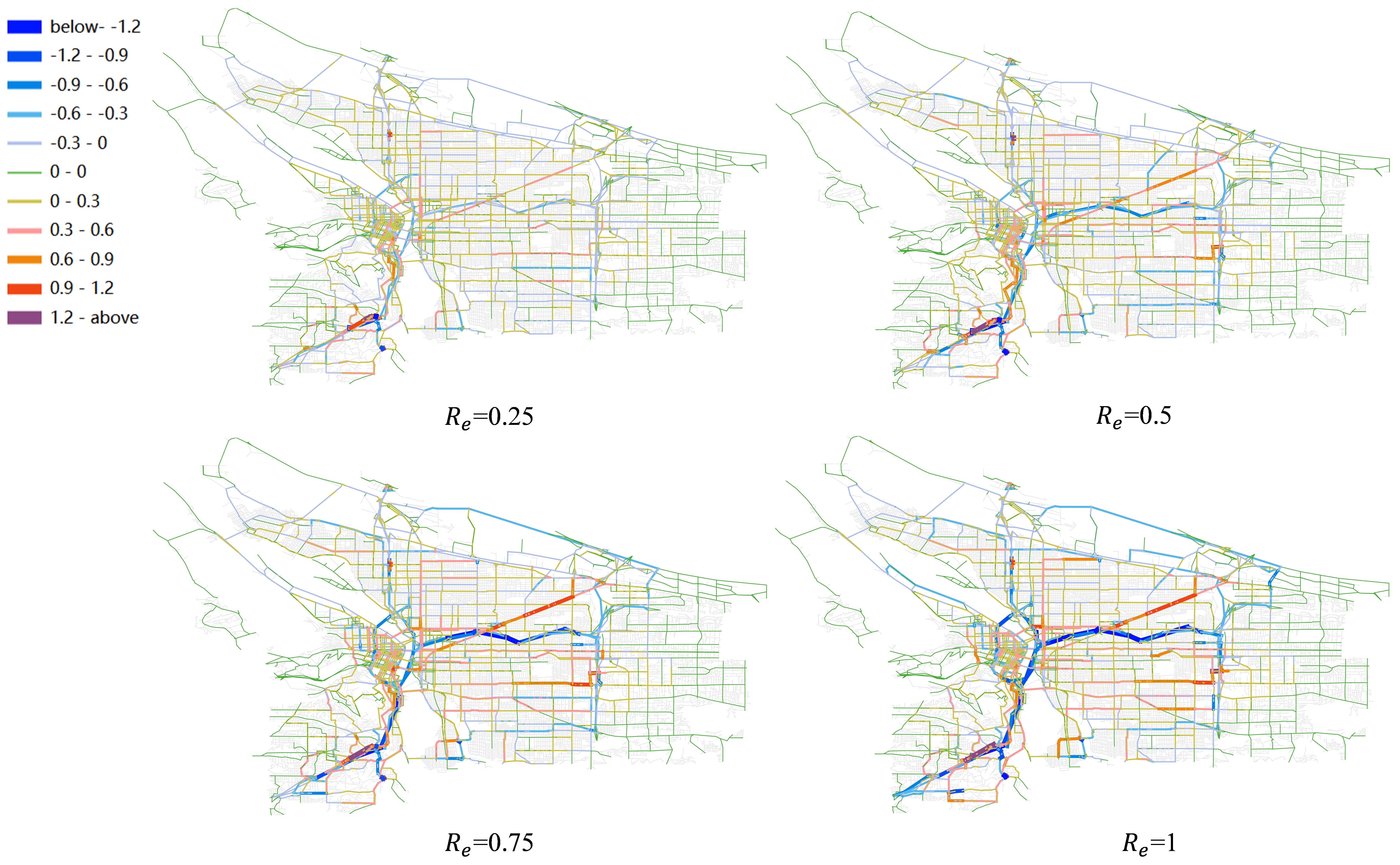}
        \caption{VOC difference distributions for Portland.}
        \label{fig:d_Portland}
    \end{subfigure}

    \caption{VOC difference distributions for 10 U.S. cities.}
    \label{fig:VOC_10cities}
\end{figure}

\begin{figure}[H]\ContinuedFloat
    \centering

    % Third subfigure (top)
    \begin{subfigure}{\textwidth}
        \centering
        \includegraphics[height=0.44\textheight,keepaspectratio]{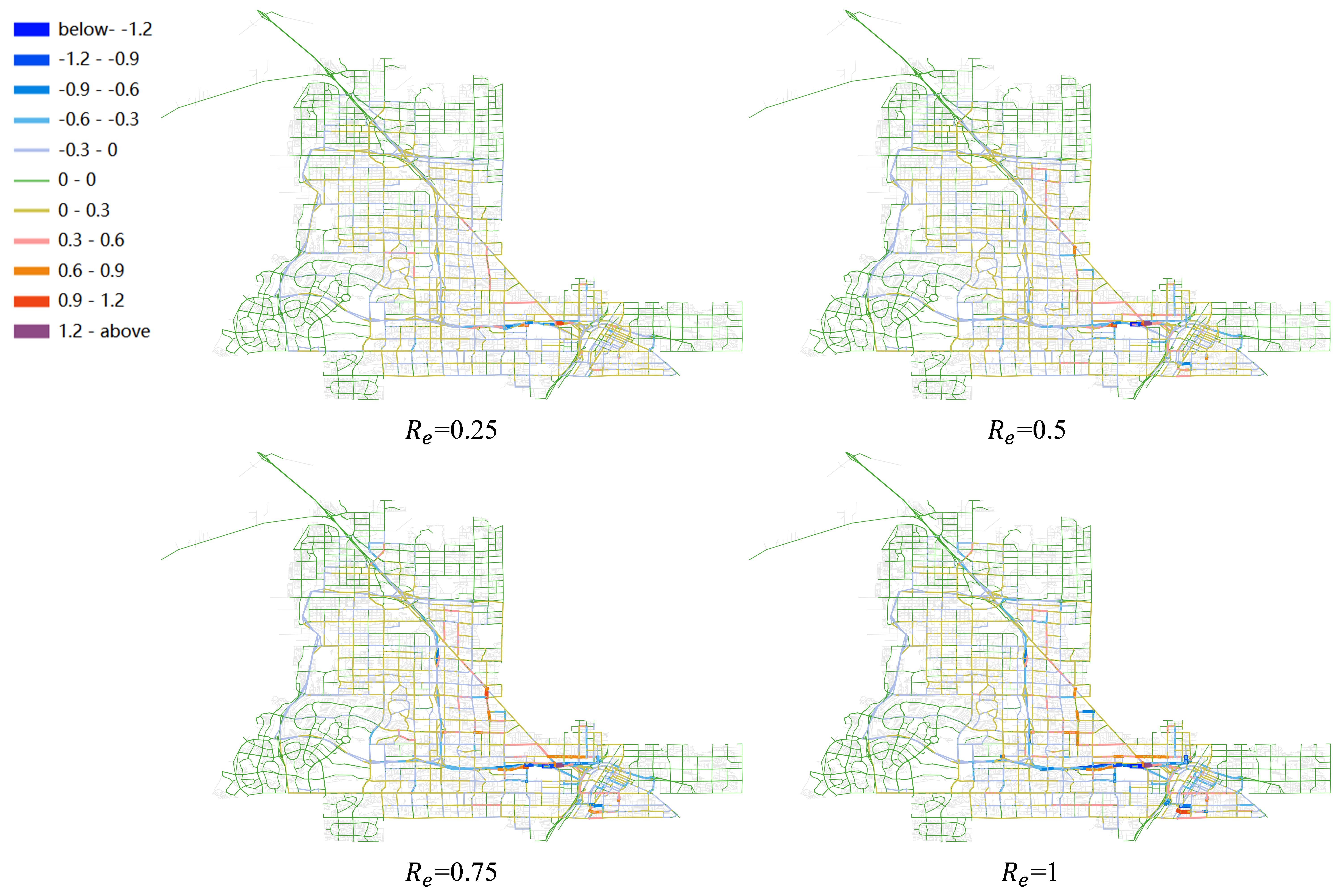}
        \caption{VOC difference distributions for Las Vegas.}
        \label{fig:d_LasVegas}
    \end{subfigure}

    \vfill

    % Fourth subfigure (bottom)
    \begin{subfigure}{\textwidth}
        \centering
        \includegraphics[height=0.46\textheight,keepaspectratio]{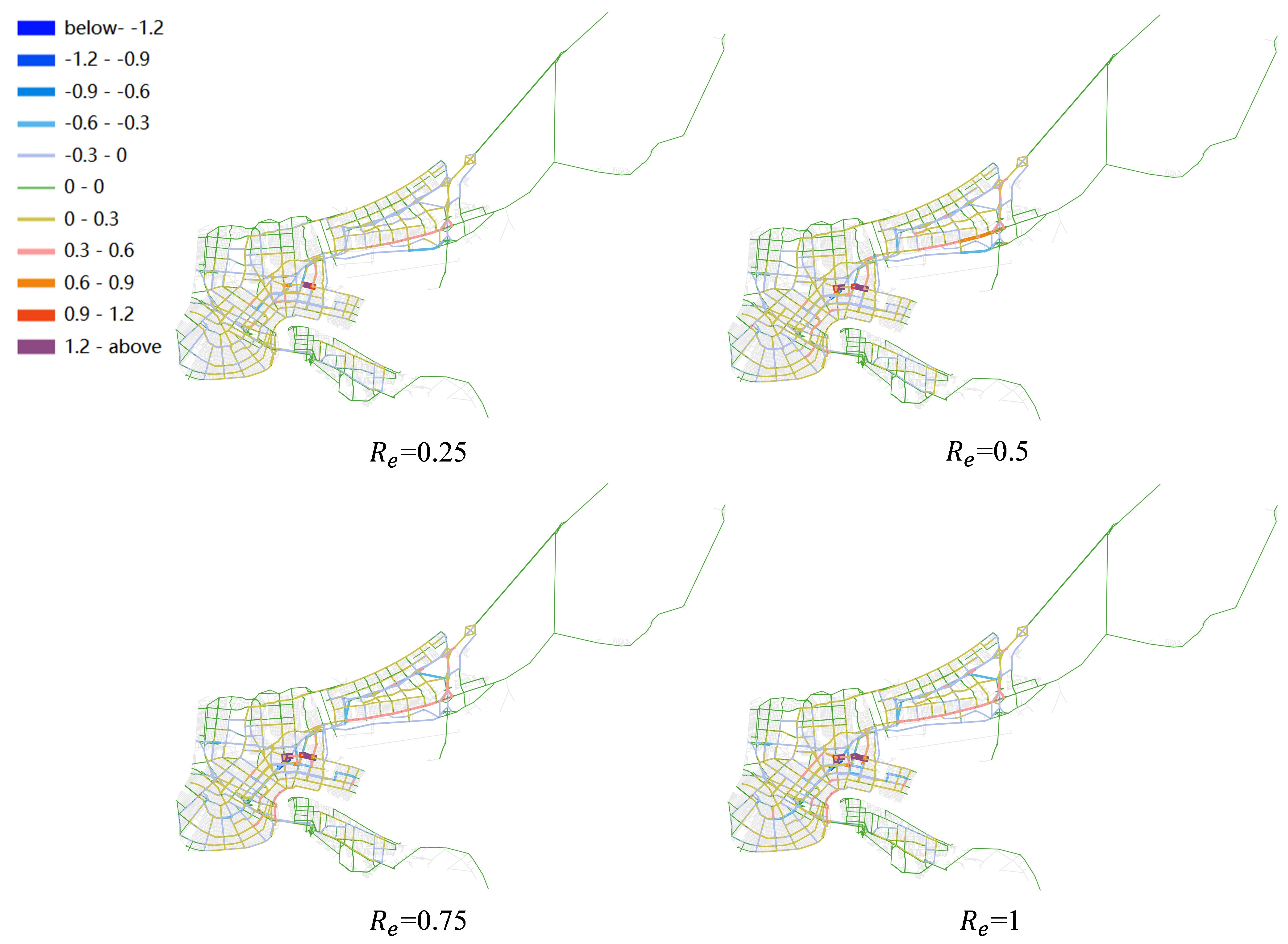}
        \caption{VOC difference distributions for New Orleans.}
        \label{fig:d_NewOrleans}
    \end{subfigure}

    \caption{VOC difference distributions for 10 U.S. cities (continued).}
\end{figure}

\begin{figure}[H]\ContinuedFloat
    \centering

    % Fifth subfigure (top)
    \begin{subfigure}{\textwidth}
        \centering
        \includegraphics[height=0.44\textheight,keepaspectratio]{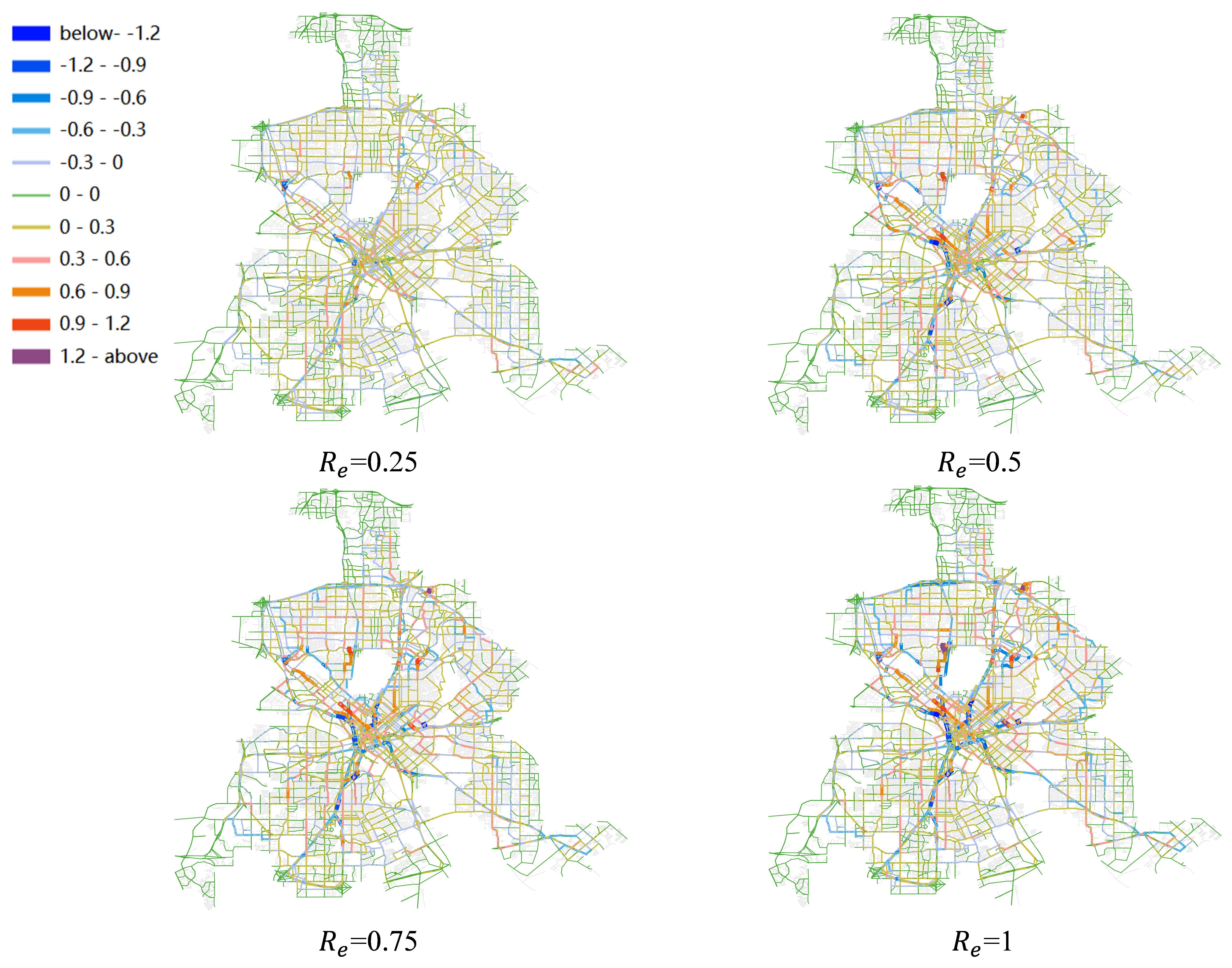}
        \caption{VOC difference distributions for Dallas.}
        \label{fig:d_Dallas}
    \end{subfigure}

    \vfill

    % Sixth subfigure (bottom)
    \begin{subfigure}{\textwidth}
        \centering
        \includegraphics[height=0.47\textheight,keepaspectratio]{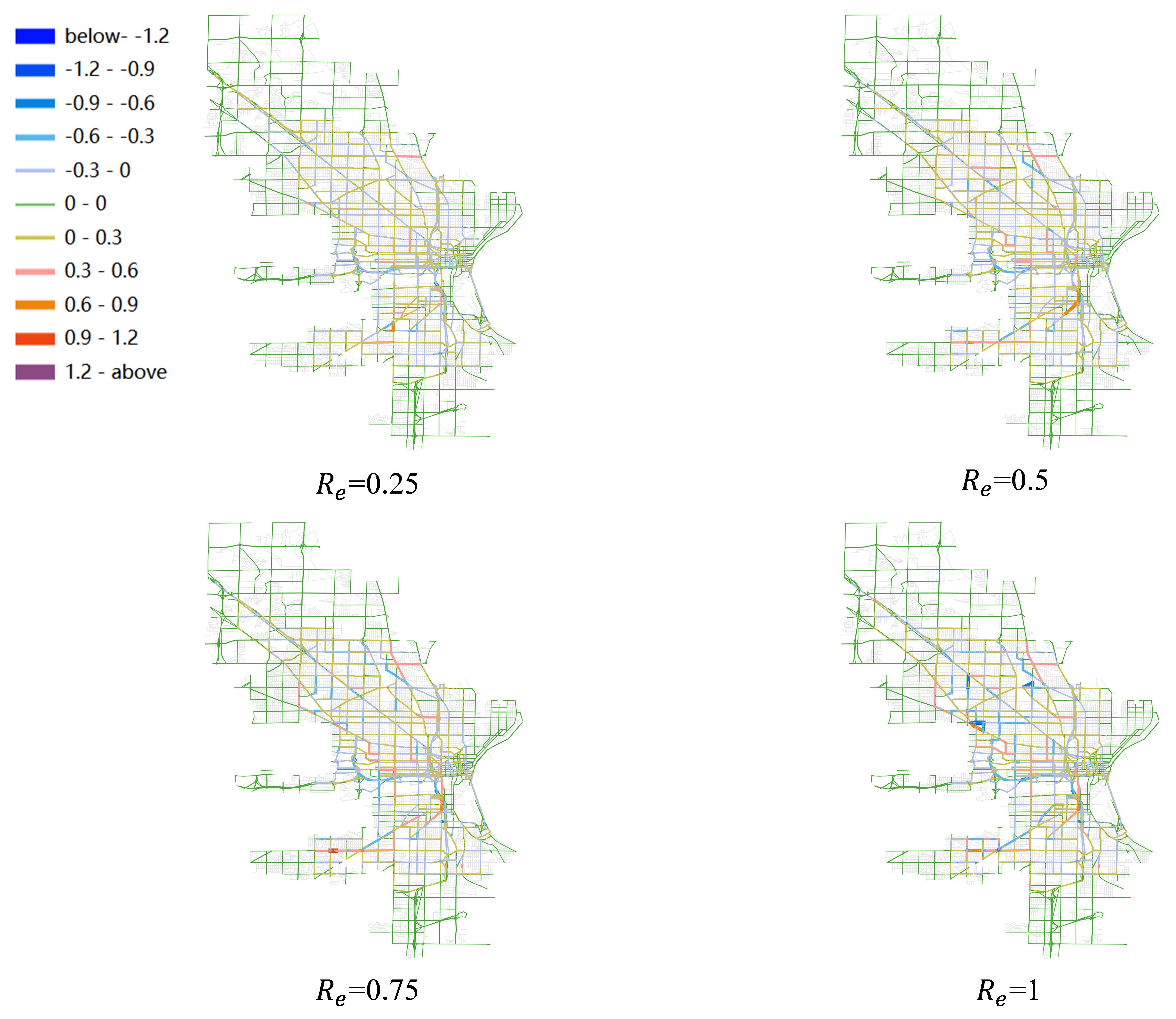}
        \caption{VOC difference distributions for Milwaukee.}
        \label{fig:d_Milwaukee}
    \end{subfigure}

    \caption{VOC difference distributions for 10 U.S. cities (continued).}
\end{figure}

\begin{figure}[H]\ContinuedFloat
    \centering

    % Seventh subfigure (top)
    \begin{subfigure}{\textwidth}
        \centering
        \includegraphics[height=0.45\textheight,keepaspectratio]{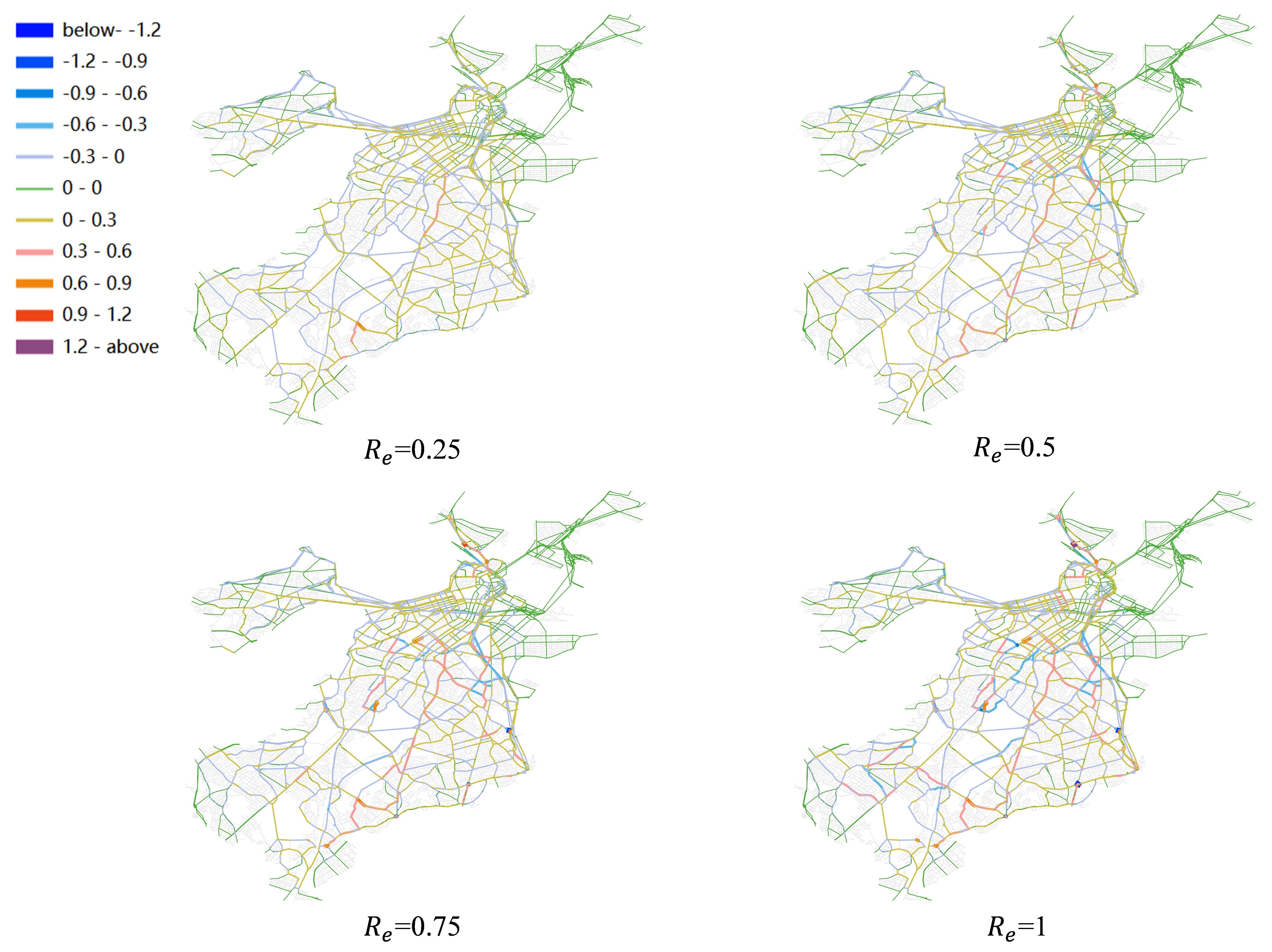}
        \caption{VOC difference distributions for Boston.}
        \label{fig:d_Boston}
    \end{subfigure}

    \vfill

    % Eighth subfigure (bottom)
    \begin{subfigure}{\textwidth}
        \centering
        \includegraphics[height=0.45\textheight,keepaspectratio]{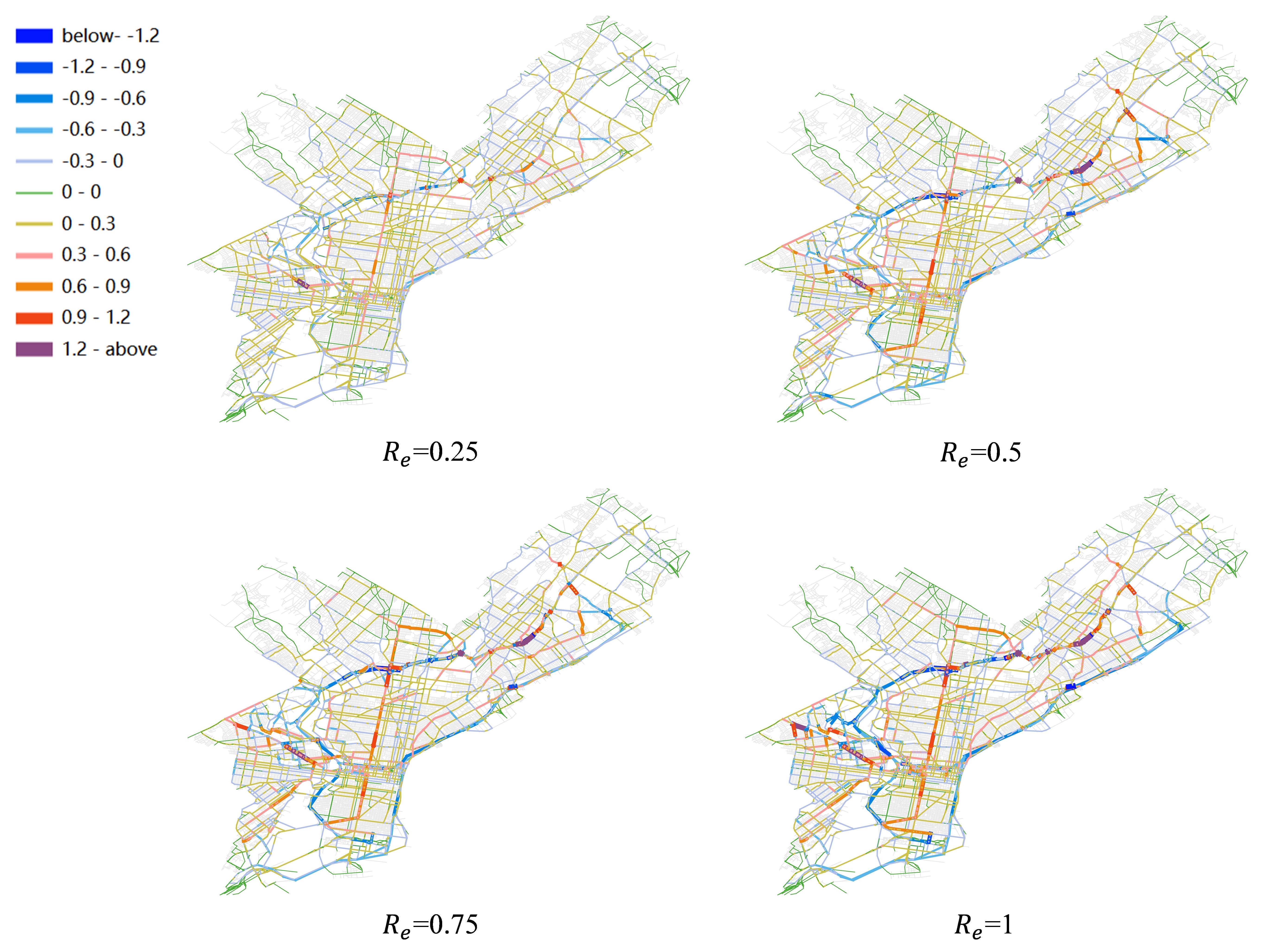}
        \caption{VOC difference distributions for Philadelphia.}
        \label{fig:d_Philadelphia}
    \end{subfigure}

    \caption{VOC difference distributions for 10 U.S. cities (continued).}
\end{figure}

\begin{figure}[H]\ContinuedFloat
    \centering

    % Ninth subfigure (top)
    \begin{subfigure}{\textwidth}
        \centering
        \includegraphics[height=0.44\textheight,keepaspectratio]{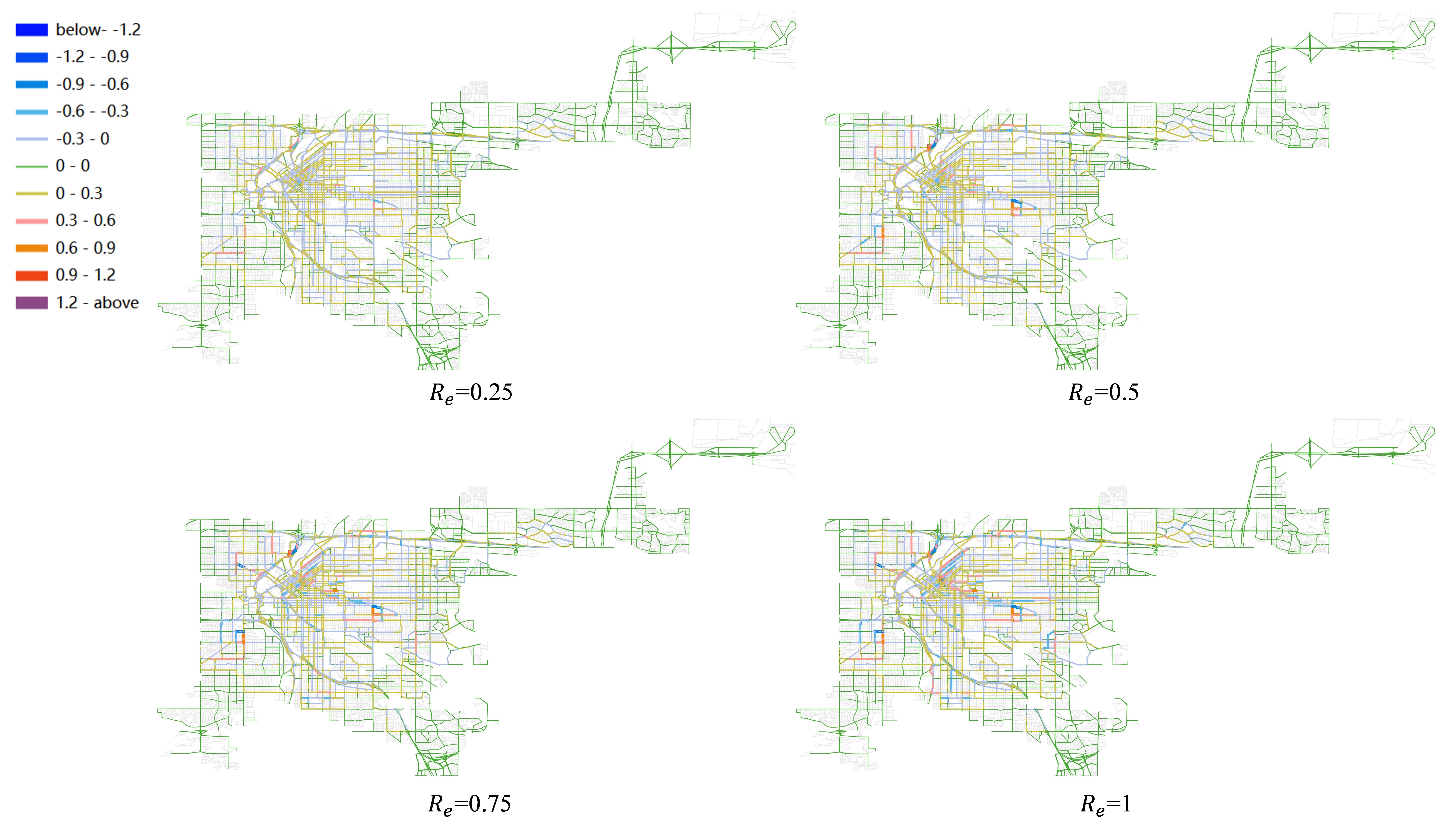}
        \caption{VOC difference distributions for Denver.}
        \label{fig:d_Denver}
    \end{subfigure}

    \vfill

    % Tenth subfigure (bottom)
    \begin{subfigure}{\textwidth}
        \centering
         \includegraphics[width=\textwidth]{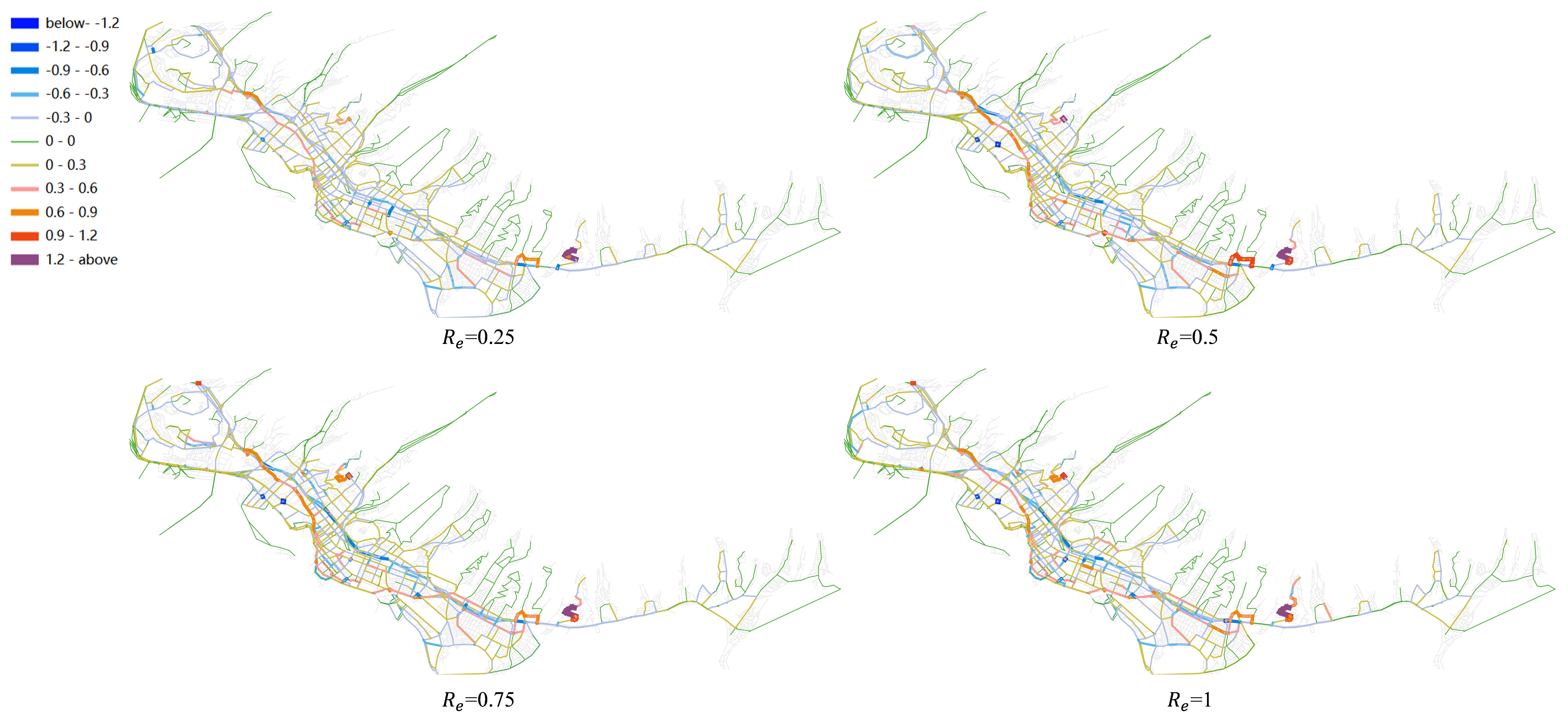}
        \caption{VOC difference distributions for Honolulu.}
        \label{fig:d_Honolulu}
    \end{subfigure}

    \caption{VOC difference distributions for 10 U.S. cities (continued).}
    
\end{figure}

 Note: (a) Color scale is unified across cities for direct comparison. Cities with smaller VOC changes (e.g., Denver) show muted colors, accurately reflecting smaller magnitudes. (b) Spatial patterns in sparse networks (Honolulu) should be interpreted cautiously due to limited link coverage. (c) Edge links may show artifacts from OD truncation at study area boundaries; core network patterns are more reliable.

%\appendix

%\section{Section in Appendix}

%Sample text. Sample text. Sample text. Sample text. Sample text. Sample text. 
%Sample text. Sample text. Sample text. Sample text. Sample text. Sample text. 
%Sample text. 

%% References


%\bibliographystyle{elsarticle-num}
% \bibliographystyle{elsarticle-harv}
% \bibliographystyle{elsarticle-num-names}
% \bibliographystyle{model1a-num-names}
% \bibliographystyle{model1b-num-names}
% \bibliographystyle{model1c-num-names}
% \bibliographystyle{model2-names}
% \bibliographystyle{model3a-num-names}
% \bibliographystyle{model3-num-names}
% \bibliographystyle{model4-names}
% \bibliographystyle{model5-names}
% \bibliographystyle{model6-num-names}

\bibliography{sample}

\end{document}